%% file: technical_report.tex
\documentclass[letterpaper]{article} 
\usepackage{aaai24}  
\usepackage{times}  
\usepackage{helvet}  
\usepackage{courier}  
\usepackage[hyphens]{url}  
\usepackage{graphicx} 
\usepackage{natbib}  
\usepackage{caption} 
\usepackage{amsmath, amsthm, mathtools}
\usepackage{amsfonts}
\usepackage{dsfont}
\usepackage[ruled,linesnumbered]{algorithm2e}
\usepackage{tabularx}
\usepackage{comment}
\usepackage{soul}
\usepackage{enumerate}
\usepackage{subcaption}
\newcommand{\ceil}[1]{\left \lceil #1 \right \rceil}
\newenvironment{proof sketch}{\begin{proof}[Proof sketch]}{\end{proof}}
\newtheorem{theorem}{Theorem}[section]
\newtheorem{lemma}{Lemma}[section]
\newtheorem{remark}{Remark}[section]

\newtheorem{definition}{Definition}[section]

\newif\ifhideannotation
\hideannotationfalse

\ifhideannotation
\newcommand{\kz}[1]{}
\newcommand{\bj}[1]{}
\newcommand{\zd}[1]{}
\newcommand{\liu}[1]{}
\newcommand{\ky}[1]{}
\else
\newcommand{\ky}[1]{{\color{orange}#1}}
\newcommand{\kz}[1] {{\footnotesize\color{cyan}[Keyuan: #1]}}
\newcommand{\bj}[1] {{\footnotesize\color{red}[Bo: #1]}}
\newcommand{\zd}[1] {{\footnotesize\color{orange}[Zhongdong: #1]}}
\newcommand{\liu}[1] {{\footnotesize\color{blue}[Zhongdong: #1]}}
\fi
\title{Learning-augmented Online Algorithm for Two-level Ski-rental Problem}
\author {
    Keyuan Zhang\textsuperscript{\rm 1},
    Zhongdong Liu\textsuperscript{\rm 1},
    Nakjung Choi\textsuperscript{\rm 2},
    Bo Ji\textsuperscript{\rm 1}
}
\affiliations {
    \textsuperscript{\rm 1}Virginia Tech, Blacksburg, VA, USA\\
    \textsuperscript{\rm 2}Nokia Bell Labs\\
    keyuanz@vt.edu, 
    zhongdong@vt.edu, nakjung.choi@nokia-bell-labs.com,
    boji@vt.edu
}

\begin{document}
\maketitle

\input{files/Abstract}
\input{files/Introduction}
\input{files/System_Model}
\input{files/Online_Algorithm}
\input{files/LA_Algorithm}
\input{files/Experiment}
\input{files/Conclusion}

\bibliography{references}
\include{files/supplementary}

\end{document}

%% file: files/Abstract.tex
\begin{abstract}
In this paper, we study the two-level ski-rental problem, where a user needs to fulfill a sequence of demands for multiple items by choosing one of the three payment options: paying for the on-demand usage (i.e., rent), buying individual items (i.e., single purchase), and buying all the items (i.e., combo purchase). Without knowing future demands, the user aims to minimize the total cost (i.e., the sum of the rental, single purchase, and combo purchase costs) by balancing the trade-off between the expensive upfront costs (for purchase) and the potential future expenses (for rent). We first design a robust online algorithm (RDTSR) that offers a worst-case performance guarantee. While online algorithms are robust against the worst-case scenarios, they are often overly cautious and thus suffer a poor average performance in typical scenarios. On the other hand, Machine Learning (ML) algorithms typically show promising average performance in various applications but lack worst-case performance guarantees. To harness the benefits of both methods, we develop a learning-augmented algorithm (LADTSR) by integrating ML predictions into the robust online algorithm, which outperforms the robust online algorithm under accurate predictions while ensuring worst-case performance guarantees even when predictions are inaccurate. Finally, we conduct numerical experiments on both synthetic and real-world trace data to corroborate the effectiveness of our approach.
\end{abstract}

%% file: files/Introduction.tex
\section{Introduction}

Decision-making under uncertainty is crucial in many real-world scenarios, such as buying or leasing a car, purchasing a daily or annual parking permit, and so on.
Such problems are often modeled as online rent-or-buy problems. One classic example is the \emph{ski-rental} problem, where a skier goes skiing for an unknown number of days and can choose to rent skis at a lower daily cost or buy them at a higher price to ski freely thereafter. The ski-rental problem and its variants have been extensively studied and found applications in
cloud computing~\citep{khanafer2013constrained}, power management~\citep{antoniadis2021learning}, serverless edge computing~\citep{pan2022retention}, etc.

However, the ski-rental problem only involves making a binary decision (i.e., rent or buy for one item only), which makes it inadequate to model more sophisticated scenarios where multiple levels of decisions are involved. 
Consider cloud computing as an example. An organization may request multiple types of server instances (e.g., HPC instances for complex system simulations, Memory instances for storing data, and GPU instances for AI tasks) from the cloud vendors such as Amazon EC2,
Microsoft Azure,
and Google Cloud~\citep{Miguel_2023}.
The cloud vendors usually offer multiple pricing options: i) pay-as-you-go: the cost only depends on the demand; ii) instance reservation: paying an upfront cost to cover the demand for that type of instance within a certain period; iii) combo offers: offering increased discounts for reserving multiple types of instances.
Hence, making choices wisely among these plans is crucial to reducing expenses on cloud resources.
Additionally, similar scenarios are common for individual users. For example, when subscribing to various digital services (e.g., Amazon offers video, music, and games), users often face choices among multiple plans: renting specific content, subscribing to one single service, or subscribing to a combo membership that grants full access to multiple services. 
Other applications, such as mobile data plans or Bahncard reservations, also offer similar payment options~\citep{wu2021competitive}.

To model multiple levels of decisions, we consider the two-level ski-rental problem introduced in~\citet{wu2021competitive}. Specifically, it considers scenarios where a user needs to fulfill the demands for multiple items in a given but \emph{unknown} time horizon. Upon the arrival of one demand for some item, the user can cover the demand by choosing one of the three payment options: i) rent, ii) single purchase, and iii) combo purchase. 
Rent is an on-demand payment with no upfront cost;  
a single purchase incurs an upfront cost (denoted by 
$C_s$) but covers the demands for that specific item since then and onward; a combo purchase has a higher upfront cost (denoted by $C_c \geq C_s$)
while covering the demands for all the items till the end. 
While the ski-rental problem only considers rent and single purchases, a unique challenge in the two-level ski-rental problem lies in the intricate decision-making between single purchases and a combo purchase covering multiple items.
To minimize the total cost (i.e., the sum of the rental, single purchase, and combo purchase costs),~\citet{wu2021competitive} propose online algorithms that make decisions without knowing any future information. \emph{Competitive ratio (CR)} is employed to evaluate the algorithms, which is defined as the ratio between the cost of the online algorithm and that of the optimal offline algorithm, over all the feasible inputs (see the formal definition of CR in Section~\ref{sec::system_model}). 
While they prove a CR for their deterministic algorithm, the CR can be unbounded when multiple units of demand arrive at the same time (see Remark~\ref{remark:comparison}).
We propose our robust online algorithm to bridge this gap. 

Although online algorithms are robust against the worst-case situation, they often exhibit excessive caution towards uncertainty, resulting in poor average performance in many real-world scenarios. In contrast, machine learning (ML) algorithms show promising average performance in many applications by using historical data to construct useful prediction models. However, they lack robustness guarantees as the predictions can be highly inaccurate due to distribution drift~\citep{lu2018learning} or adversarially chosen test data~\citep{szegedy2014intriguing}.

In the second part of this work, we first show that a simple algorithm that blindly follows ML predictions cannot guarantee robustness in worst-case scenarios. Then, we design a novel learning-augmented online algorithm by integrating ML predictions into our online algorithm. 
Specifically, we show that our algorithm achieves two desired properties: i) with an accurate prediction, it offers a better CR than our robust online algorithm \emph{(consistency)} and ii) even if the prediction is inaccurate,
it still retains a worst-case guarantee \emph{(robustness)}. 
Our learning-augmented online algorithm takes the predictions as the black-box oracles without prior knowledge of the prediction quality.

We summarize our main contributions in the following.

\emph{First}, 
we propose an online algorithm for the two-level ski-rental problem. We prove that it achieves a CR of $3 - \frac{1}{C_s} - \frac{1}{C_c}(2 - \frac{1}{C_s})$, providing a stronger performance guarantee compared to the State-of-the-Art algorithm introduced in~\citet{wu2021competitive} (see Remark~\ref{remark:comparison}).

\emph{Second}, we design a novel learning-augmented algorithm that integrates ML predictions into robust online decision-making. We prove that the proposed algorithm can guarantee both consistency and robustness. To the best of our knowledge, this is the first work that designs learning-augmented algorithms for the two-level ski-rental problem.

\emph{Finally}, we conduct numerical experiments to evaluate the performance of our online algorithm and learning-augmented algorithm using both synthetic data and real-world trace data. The results verify our theoretical analyses and corroborate that our learning-augmented algorithm guarantees both consistency and robustness.


\paragraph{Related Work.}
The ski-rental problem is introduced in~\citet{karlin1988competitive},
which gives the best deterministic online algorithm with a CR of 2. Then, an optimal randomized online algorithm achieving $e/(e-1)$-competitiveness is proposed in~\citet{karlin1994competitive} ($e$ is the Euler's number). 
Several variants of the ski-rental problem are later studied.
In the following, we briefly discuss the variants that are most relevant to ours. \citet{fleischer2001bahncard} studies the \emph{Bahncard problem}, where buying a Bahncard provides discounts on subsequent railway tickets for a certain duration. \citet{ai2014multi} study the \emph{multi-shop ski-rental problem}, where a user needs to choose when and where to buy from multiple shops with different prices for renting and purchasing. Recently, \citet{wu2022competitive} study the \emph{multi-commodity ski-rental problem}, where a user requires a bundle of commodities altogether and needs to choose payment options for each commodity. 

On the other hand, researchers have designed learning-augmented algorithms for a variety of online optimization problems. In the seminal work~\citep{pmlr-v80-lykouris18a}, 
they show that by incorporating ML predictions, the classic online Marker algorithm can achieve both consistency and robustness for the caching problem. Learning-augmented algorithm design is also studied for the ski-rental problem and its variants~\citep{purohit2018improving, bamas2020primal, wang2020online}. For example, both single and multiple ML predictions are incorporated into the online algorithm for the multi-shop ski-rental problem~\citep{wang2020online}. We refer interested readers to the surveys~\citep{boyar2017online, mitzenmacher2022algorithms} for comprehensive coverage. 
In our work, however, the combo purchase can cover all the single purchases and thus result in coupled decisions. This ``two-level'' payment option is more practical in quite a few real-world applications but introduces unique challenges in the algorithm design.

%% file: files/System_Model.tex
\section{System Model and Problem Formulation} \label{sec::system_model}

\paragraph{System Model.}
We consider a time-slotted model with time index $t = 1,2, \dots, T$, where $T$ is the length of a finite time horizon \emph{unknown} to the user (i.e., the decision maker). The user needs to fulfill demands for $K$ items arriving over time. Upon the arrival of demands in each time-slot, the user can fulfill them by choosing from different options (rent or purchase).
Without loss of generality, we assume that the demands arriving in the same time-slot are for one item only.

Let $d(t):=(i(t), a(t))$ denote the demand arriving in time-slot $t$, where $i(t) \in \{1,2,\dots,K\}$
is the index of the item and $a(t)$ is a non-negative integer representing the amount of the demand.
In particular, $d(t)=(0, 0)$ indicates no demand in time-slot $t$.
Let $\mathbf{D}:=\{d(t)\}_{t=1}^T$ denote the sequence of demands over the entire time horizon $T$.

To fulfill the demand in each time-slot, the user can choose one of the three payment options: i) rent, ii) single purchase, and iii) combo purchase. 
\begin{enumerate}[i)]
    \item \textbf{Rent}: The user pays a unit rental cost to cover one unit of demand in time-slot $t$. 
    \item \textbf{Single purchase}: 
    The user pays an upfront fee $C_s > 1$, which covers all the demands for item $i(t)$ from time $t$ to $T$. For ease of analysis, we assume the same single purchase cost $C_s$ for all the items.
    \item \textbf{Combo purchase}: 
    The user pays a higher upfront fee $C_c \in (C_s, K\cdot C_s)$, which covers the demands for all the items from time $t$ to $T$.
\end{enumerate}
We assume that both $C_s$ and $C_c$ are integers and that once the payment is made in some time-slot $t$, the associated demands can be fulfilled immediately. 
While choosing to rent may offer a lower immediate cost, it may result in a higher future cost when confronted with increased future demands. In contrast, choosing to purchase (either single or combo) incurs a higher upfront cost, but it covers the potential future demands.

\subsubsection{Problem Formulation.} Let $r(t)$ be the binary rental decision in time-slot $t$, where $r(t) = 1$ if the user decides to rent, otherwise $r(t) = 0$. 
Let $t_k$ denote the time when the user makes a single purchase for item $k$. We set $t_k = \infty$ when a single purchase is never made. Similarly, let $t_c$ denote the time when the user makes a combo purchase, and we set $t_c=\infty$ when a combo purchase is never made. 
We consider \emph{online} algorithm $\pi$ that determines $\{ {r^{\pi}}(t)~\text{for}~ t=1,2,\dots,T, \{ t_k^{\pi}\}_{k=1}^K, t_c^{\pi}\}$
only using information available in time-slot $t$: the rental decision history $\{r^{\pi}(\tau)\} _{\tau = 1}^{t-1}$, single purchase history $\{t^\pi_k\}_{k=1}^K$, combo purchase history $t^\pi_c$, demand history $\{d(\tau)\} _{\tau = 1}^t$, number of instances $K$, single purchase cost $C_s$, and combo purchase cost $C_c$; no future information is available to the algorithm.
For notational simplicity, we drop the superscript $\pi$
in the rest of the paper whenever the context is clear.

Let $\mathcal{C}(\mathbf{D},\pi)$ be 
the total cost under Algorithm $\pi$:
\begin{equation}
    \mathcal{C}(\mathbf{D},\pi):=\sum_{t=1}^Ta(t)r(t) + \sum_{k=1}^K C_s \mathds{1}_{\{t_k \le T\}} + C_c\mathds{1}_{\{t_c \le T\}},
\end{equation}
where $\mathds{1}_{\{\cdot\}}$ is the indicator function and the three terms correspond to the total rental cost, total single purchase cost, and combo purchase cost, respectively. 
Given a demand sequence $\mathbf{D}$, the objective is to minimize the total cost:
\begin{subequations}
\begin{align}
    &\min_{r(t), t_k, t_c} & & \mathcal{C}(\mathbf{D},\pi)  \\
    &\hspace{0.4cm} \text{subject to} 
    & & \begin{aligned}[t]
        r(t) + \mathds{1}_{\{t \ge t_{i(t)}\}} &+ \mathds{1}_{\{t \ge t_c\}}\geq 1, \\ &t\in\{1,2,\dots,T\}
    \end{aligned} \label{fulfill} \\
    &&&  r(t)\in \{0,1\}, \ t\in \{1,2,\dots,T\}\\    
    &&& t_k > 0, \quad k \in \{1, 2, \dots, K\} \\
    &&& t_c > 0,
\end{align} 
\label{prob:1}
\end{subequations}
\!\!where Contraint~\eqref{fulfill} requires that each demand must be fulfilled via either renting or purchasing (single or combo).

\paragraph{Competitive Ratio.} Without knowledge of the future demand, it is usually difficult for an online algorithm to attain the same minimum total cost achieved by an optimal offline algorithm, which requires knowledge of the entire demand sequence $\mathbf{D}$ beforehand.  
To measure the performance of an online algorithm, we consider
a widely adopted metric called \emph{competitive ratio} (CR), defined as the ratio between the cost of the online algorithm and that of the optimal offline algorithm over all possible inputs. Formally, 
an online algorithm $\pi$ is called $c$-competitive if there exists a constant $c \ge 1$ such that for any demand sequence $\mathbf{D}$, there is
\begin{equation}
    \mathcal{C}(\mathbf{D}, \pi) \leq c \cdot \mathit{OPT}(\mathbf{D}),
\end{equation}
where $\mathit{OPT}(\mathbf{D})$ is the cost of the optimal offline algorithm under the demand sequence $\mathbf{D}$. 

%% file: files/Online_Algorithm.tex
\section{Robust Online Algorithm}
\label{sec::online}
In this section, we introduce our robust online algorithm and prove its competitive ratio. 
To begin with, we first present two important notions that will be used in our algorithm: indicative cost and purchase threshold.

\paragraph{Indicative Cost and Purchase Threshold.}
We define the indicative cost $\psi_k(t)$ 
for item $k$ in time-slot $t$ as the total demand for item $k$ during the interval $[1, t]$ that is not covered by any purchase. Then, indicative cost $\psi_k(t)$ evolves as 
\begin{equation} \label{psik}
    \psi_{k}(t) \coloneqq 
    \begin{cases}
    \psi_k(t-1) + a(t), & \begin{aligned}
        &i(t)=k \ \text{and} \\
        &t < \min\{t_k, t_c\},
    \end{aligned} \\
    \psi_k(t-1),  & \text{otherwise},
    \end{cases} 
\end{equation}
where $\psi_k(t)$ increases by $a(t)$ if the demand in time-slot $t$ is for item $k$ (i.e., $i(t) = k$) and is not covered by (either single or combo) purchases (i.e., $t < \min\{t_k, t_c\}$);
otherwise, it remains unchanged.
Note that $\psi_k(0) = 0$ for all $k$.
Let $\lambda_s \in (1, C_s]$ be the threshold for making single purchases. 
That is, if the single purchase indicator $\psi_k(t)$ reaches $\lambda_s$ (i.e., $\psi_k(t) \geq \lambda_s$), we will make a single purchase for item $k$.
We assume the same threshold $\lambda_s$ for all items because they have the same single purchase price $C_s$.
This is also an assumption made in~\citet{wu2021competitive}. While our algorithm can potentially be adapted to address item-specific purchase costs and thresholds, a more sophisticated competitive analysis is expected.

Similarly, let $\psi_c(t)$ and $\lambda_c \in (1, C_c]$ be the overall indicative cost $\psi_c(t)$ and the combo purchase threshold, respectively. We define the overall indicative cost $\psi_c(t)$ as 
\begin{equation}
   \psi_c(t) \coloneqq \sum_{k=1}^K \min \{\psi_k(t), \lambda_{s} \}. \label{psic}
\end{equation}
When $\psi_c(t)$ reaches the combo purchase threshold $\lambda_c$ (i.e., $\psi_c(t) \geq \lambda_c$), a combo purchase will be made. We assume that thresholds $\lambda_s$ and $\lambda_c$ are integers due to the integral demand we consider.

Intuitively, when indicative cost $\psi_k(t)$ reaches the single purchase threshold $\lambda_s$, it reflects a substantial demand for item $k$, leading to a single purchase of item $k$. Similarly, when the overall indicative cost $\psi_c(t)$ reaches $\lambda_c$, it indicates substantial demand for multiple items, leading to a combo purchase rather than single purchases or rent.

\begin{algorithm}[!t]
\SetAlgoLined
\SetKwInOut{Input}{Input}
\SetKwInOut{Output}{Output}
\SetKwInOut{Init}{Init.}
\Input{$K$, $C_s$, $C_c$, $\lambda_s$, $\lambda_c$, $\mathbf{D}$ (revealed in an online manner)}
\Output{$r(t), t_k, t_c$}
\Init{$\psi_c(t),\psi_k(t),r(t) \gets 0,  t_k, t_c \gets \infty$}

\While {new demand $d(t)$ arrives} {  \label{line:begin}
    Update $\psi_k(t)$ according to Eq.~\eqref{psik}\; \label{line:psik}
    Update $\psi_c(t)$ according to Eq.~\eqref{psic}\; \label{line:psic}
    
    \eIf{$d(t)$ is covered by a previous purchase}{ \label{line:cover}
    \label{line:update_finish}
    Do nothing and wait till next time-slot\; 
    }{

    
        \uIf{$\psi_c(t) \geq \lambda_c$}{  \label{line:combo_begin}
        $t_c \gets t$; \tcp{combo purchase}
        } \label{line:combo_complete}
        \uElseIf{$\psi_{i(t)}(t) \geq \lambda_{s}$} { \label{line:single_begin}
        $t_{i(t)}\gets t$; \tcp{single purchase}
        }\label{line:single_complete}
        \Else{
        $r(t) \gets 1$; \tcp{rent} \label{line:rental}
        }
    }
}
\caption{Robust Deterministic Two-level Ski-rental (RDTSR) Algorithm 
}
\label{alg:online}
\end{algorithm}

\paragraph{Algorithm Description.}
We now present our \emph{Robust Deterministic Two-level Ski-rental (RDTSR)} Algorithm in Algorithm~\ref{alg:online} and explain how it works in the following. 
When a new demand $d(t)$ arrives in time-slot $t$, it first updates the indicative costs $\psi_{k}(t)$ and $\psi_c(t)$ according to Eqs.~\eqref{psik} and~\eqref{psic}, respectively (Lines~\ref{line:psik}-\ref{line:psic}).
If the demand is already covered by a previous (single or combo) purchase (Line~\ref{line:cover}), it does nothing.
Otherwise, if the indicative cost exceeds the threshold, it will purchase by prioritizing the combo purchase (Lines~\ref{line:combo_begin}-\ref{line:single_complete}); if not, it will cover the demand via renting (Line~\ref{line:rental}).

\paragraph{Competitive Analysis.} We derive the CR of 
RDTSR as a function of thresholds $\lambda_s$ and $\lambda_c$. By choosing proper values of $\lambda_s$ and $\lambda_c$, we show that the CR is upper bounded by $3$.
\begin{theorem} \label{thm:cr}
The CR of RDTSR is upper bounded by
\begin{equation}
    3 - \frac{1}{C_s} - \frac{1}{C_c}(2 - \frac{1}{C_s}).
\end{equation}
In particular, this upper bound is achieved by choosing thresholds $\lambda_s = C_s$ and $\lambda_c = C_c$.
\end{theorem}
\begin{proof sketch}
The detailed proof is provided in Appendix~\ref{app:them:cr} and we give a proof sketch in the following.
The proof has two key steps: 1) given any $\lambda_s$ and $\lambda_c$, we derive an upper bound of the CR, which is a function of $\lambda_s$ and $\lambda_c$; 2) we show that the CR is at most $3 - \frac{1}{C_s} - \frac{1}{C_c}(2 - \frac{1}{C_s})$ when $\lambda_s = C_s$ and $\lambda_c = C_c$.

Note that it is typically challenging to directly analyze the performance for an arbitrary demand sequence. Therefore, we break Step 1) into three substeps: 1a)
We derive an upper bound of CR over all the demand sequences with the same total demand. The total demand is a vector consisting of $K$ elements where each element represents the total demand for the associated item. 1b) Following a similar line of analysis in~\citet{wu2021competitive}, we define a standard total demand to simplify the upper bound we obtained in Step 1a). 1c) We obtain an upper bound of CR over all the standard total demand, which is $\max \{\frac{\lambda_s - 1 + C_s}{\lambda_s}, \frac{\lambda_c - 1 + C_c + (\lambda_c - 1)\lambda_s^{-1}(C_s - 1)}{\min\{C_c,  \lambda_c\}}\}$. In Step 2), we minimize the function of $\lambda_s$ and $\lambda_c$ derived in Step 1c).
\end{proof sketch}

\begin{remark} \label{remark:comparison}
\citet{wu2021competitive} 
proposed a similar deterministic online algorithm and showed that their algorithm is $(3 - 1/C_s)$-competitive. However, we construct a simple example to show that the CR of their algorithm can be unbounded when allowing multiple units of demand to arrive in a time-slot (i.e., $a(t) > 1$). Specifically, consider a demand sequence $\Tilde{\mathbf{D}} := \{(1, C_s), (2, C_s), \dots, (K, C_s)\}$, i.e., $C_s$ units of demand for item $k$ arrive in time-slot $k$ for all $k$.
In their algorithm, single purchases will be made for all the items, while the combo purchase will never be made as their algorithm skips the updates of the overall indicative cost 
$\psi_c(t)$ after a single purchase. With the fact that the cost of the optimal offline algorithm is at most $C_c$, their algorithm's competitive ratio is at least $KC_s/C_c$, which can be arbitrarily large when $KC_s \gg C_c$. We address this issue by redefining the indicative cost and assigning the highest priority to the combo purchase. It is also noteworthy that even for the setting with unit demand arrival (i.e., $a(t)=1$), our algorithm still has a slightly improved upper bound ($3 - \frac{1}{C_s} - \frac{1}{C_c}(2 - \frac{1}{C_s})$ vs. $(3 - \frac{1}{C_s})$). 
\end{remark}

%% file: files/LA_Algorithm.tex
\section{Learning-augmented Online Algorithm}
In the previous section, we introduce a robust online algorithm, RDTSR, with a worst-case guarantee. However, online algorithms are often overly conservative and may have a poor average performance in real-world scenarios. In contrast, ML algorithms have the advantage of utilizing extensive historical data and building well-trained models, which enables them to achieve a promising average performance. Nonetheless, ML algorithms lack a worst-case guarantee when facing adversaries, outliers, distribution shifts, etc. 
To harness the benefits of both methods, we design a learning-augmented online algorithm that achieves both consistency and robustness, two commonly used notions in the literature.

\subsection{Machine Learning Predictions}
We consider the case where an ML algorithm provides a prediction $y \coloneqq (y_1, \dots, y_K)$ 
that represents the predicted total demand for each item. Predictions of the total demand are commonly used in real-world applications. For example, in cloud computing, various resources such as CPU, memory, and network bandwidth can be fed into a predictor to forecast future workload~\citep{masdari2020survey}.

We assume that the ML prediction $y$ is available in the very beginning (i.e., $t=0$) and is provided in full (i.e., all $y_k$'s are available). 
The actual total demand can be represented by a vector $z \coloneqq (z_1, \dots, z_K)$, where $z_k$ denotes the total demand for item $k$, i.e.,
 $z_k = \sum_{t=1}^T \mathds{1}_{\{i(t)=k\}} a(t)$.
The prediction error is defined as the $\ell 1$-norm distance from the actual total demand to the predicted ones, 
and we use $\eta_k \coloneqq | y_k - z_k|$ and $\eta \coloneqq \sum_{k=1}^K \eta_k$ to denote the prediction error for item $k$ and the total prediction error, respectively. 

Consider a learning-augmented online algorithm $\pi^{\dag}$ that generates a solution with cost $\mathcal{C}(\mathbf{D}, \pi^{\dag})$ using information available in time-slot $t$ and a prediction $y$. Algorithm $\pi^{\dag}$ is called \emph{$\alpha$-consistent} if $\mathcal{C}(\mathbf{D}, \pi^{\dag}) \le \alpha \cdot \mathit{OPT}(\mathbf{D})$ for any sequence $\mathbf{D}$ when the prediction is perfect (i.e., $\eta = 0$), and it is called \emph{$\beta$-robust} if $\mathcal{C}(\mathbf{D}, \pi^{\dag}) \le \beta \cdot \mathit{OPT}(\mathbf{D})$ for any sequence $\mathbf{D}$, regardless of the prediction error $\eta$.
Here, consistency measures the performance in the best-case scenario with perfect predictions, while robustness measures the performance of a learning-augmented algorithm in the worst-case scenario regardless of the prediction quality.

\subsection{A Simple Algorithm: Follow the Prediction} \label{sec::follow}
First, we show that a simple algorithm called \emph{Follow the Prediction (FTP)}, cannot guarantee robustness. 
FTP operates in the following way: it compares the combo purchase cost $C_c$ with the total cost without the combo purchase option (i.e., $\sum_{k=1}^K \min \{C_s, y_k\}$), and makes a combo purchase if it is lower (i.e., $\sum_{k=1}^K \min \{C_s, y_k\} \geq C_c$); otherwise, it makes a single purchase for item $k$ if $y_k \geq C_s$ or rents for item $k$ when needed.
Lemma~\ref{lemma:FTP} gives an upper bound on the total cost under FTP.

\begin{lemma} \label{lemma:FTP}
FTP satisfies $\mathcal{C}(\mathbf{D}, \mathrm{FTP}) \leq \mathit{OPT}(\mathbf{D}) + \eta$.
\end{lemma}

\begin{proof sketch}
The detailed proof is provided in Appendix~\ref{app:lemma-4.1} and we give a proof sketch in the following. 
We first show that the cost of both FTP and the optimal offline algorithm only depends on total demand. Then, we only need to analyze the CR based on the total demand. We use $\mathit{ALG}(z)$ and $\mathit{OPT}(z)$ to denote the cost of FTP and the optimal offline algorithm under a total demand $z$, respectively. 
Depending on whether FTP and the optimal offline algorithm each make a combo purchase or not, we consider four cases.
Finally, we show $\mathit{ALG}(z) \le \mathit{OPT}(z) + \eta$ for each case.
\end{proof sketch}

\begin{remark}
When the prediction is perfect (i.e., $\eta = 0$), FTP achieves the optimal cost. However, the performance of FTP can be very poor when the prediction error $\eta$ is large. For example, when $\sum_{k=1}^K \min \{C_s, y_k\} < C_c$ and $y_k  < C_s$, 
but there is a very large total demand for item $k$, FTP would choose to rent for item $k$ at all times, resulting in an unbounded cost and thus an unbounded CR as the optimal cost is bounded by $C_c$. Therefore, although FTP can achieve consistency, it lacks a robustness guarantee.
\end{remark}

\subsection{Design of Learning-augmented Algorithm} \label{sec::la}
As following the prediction in a straightforward way (e.g., FTP) cannot guarantee the worst-case performance, we aim to 
wisely integrate the ML predictions with RDTSR and design a learning-augmented online algorithm that can achieve both consistency and robustness.

The main idea of the design is to adjust the threshold of RDTSR based on the prediction: if a prediction suggests a (single or combo) purchase, then we should decrease the associated threshold to encourage the purchase (recall that the threshold determines when to make a purchase); if a prediction suggests not making a purchase, we should increase the associated threshold.
The extent to which we adjust the threshold relies on our confidence in the prediction. 

Based on the above intuition, we design the \emph{Learning-augmented Deterministic Two-level Ski-rental (LADTSR)} Algorithm, presented in Algorithm~\ref{ALG}.
It is almost the same as RDTSR, except for the choice of thresholds.
First, we specify a trust parameter $\theta \in [0, 1]$, reflecting the confidence in the prediction: 
a smaller $\theta$ indicates higher confidence.
Then, we adjust the thresholds as follows: for each single purchase threshold $\lambda_{s, k}$ for item $k$, we reduce it to $\theta C_s$ if the prediction suggests a single purchase for this item (i.e., $y_k \ge C_s$); otherwise, we increase $\lambda_{s, k}$ to $C_s/\theta$. 
For the combo purchase threshold $\lambda_c$, we reduce it to $\theta^2 C_c$ (see Remark~\ref{remark:lambdac} for an explanation of this choice instead of $\theta C_c$)
if the prediction suggests a combo purchase (i.e., $\sum_{k=1}^K \min \{C_s, y_k\} \ge C_c$); otherwise, we increase $\lambda_c$ to $C_c/\theta$. 
We set the threshold to infinity in case of division by 0.
Note that when running RDTSR (Line~\ref{line:run_rdtsr}), the definition of $\psi_c(t)$ (i.e., Eq.~\eqref{psic}) needs to be modified slightly by replacing $\lambda_s$ with $\lambda_{s, k}$.

\begin{algorithm}[!tb]
\SetAlgoLined
\SetKwInOut{Input}{Input}
\SetKwInOut{Output}{Output}
\SetKwInOut{Init}{Init.}
\Input{$y$, $\theta$, $K$, $C_s$, $C_c$, $\mathbf{D}$ (revealed in an online manner)}
\Output{$r(t), t_k, t_c$}
\Init{$\psi_c(t),\psi_k(t),r(t) \gets 0, t_k, t_c \gets \infty$}

\For{item $k=\{1,2,\dots,K\}$} {\label{line:setthreshold_begin}
\eIf{$y_k \ge C_s$
}{
 $\lambda_{s,k} \gets \theta C_s$\;
}{
  $\lambda_{s,k} \gets C_s / \theta$\;
}
}
\label{line:setthreshold_single}
\eIf{$\sum_{k=1}^K \min \{C_s, y_k\} \ge C_c$ 
}{ \label{line:setthreshold_combo}
$\lambda_c \gets \theta^2 C_c$\;
}{
    $\lambda_c \gets C_c / \theta$\;
}
\label{line:setthreshold_end}
Run RDTSR with thresholds $\lambda_{s, k}$ and $\lambda_c$; 
\label{line:run_rdtsr}

\caption{Learning-augmented Deterministic Two-level Ski-rental (LADTSR) Algorithm 
}
\label{ALG}
\end{algorithm}

\subsection{Analysis of Learning-augmented Algorithm}
In this subsection, we focus on the consistency and robustness analysis of LADTSR with $\theta \in (0, 1)$. The special cases of LADTSR with $\theta=0$ and $\theta=1$ correspond to FTP and RDTSR, respectively.
It is noteworthy that by choosing different values of $\theta$, LADTSR exhibits a crucial trade-off between consistency and robustness.
\begin{theorem} \label{thm:la}
The CR of LADTSR is upper bounded by
\begin{equation}
    \min \left \{1+\theta+\theta^2 + \frac{1 + 2\theta}{1 - \theta} \cdot \frac{\eta}{\mathit{OPT}(\mathbf{D})}, 1+\theta^{-1} + \theta^{-3} \right \},
\end{equation}
for $\theta \in (0, 1)$.
Futhermore, LADTSR is $(1 + \theta + \theta^2)$-consistent and $(1 + \theta^{-1} + \theta^{-3}) $-robust.
\end{theorem}

\begin{proof sketch}
The detailed proof is provided in Appendices~\ref{sec:consistency} and~\ref{sec::robust},
and we give a proof sketch in the following.

We first prove the first bound $1+\theta+\theta^2 + \frac{1 + 2\theta}{1 - \theta} \cdot \frac{\eta}{\mathit{OPT}(\mathbf{D})}$. Our goal is to show that this upper bound holds for any demand sequence and any prediction. First, we consider two cases for the choice of $\lambda_c$ (i.e., $\lambda_c = \theta^2 C_c$ and $\lambda_c = C_c / \theta$).
In each case, we consider four subcases, depending on whether LADTSR and the optimal offline algorithm each make a combo purchase or not. In each subcase, we derive an upper bound of CR, which is a function of $\eta$ and $\mathit{OPT}(\mathbf{D})$. Combining all eight subcases yields the first bound.

Next, we prove the second bound $1 + \theta^{-1} + \theta^{-3}$. The analysis is similar to Step 1) in the proof of Theorem~\ref{thm:cr} and also has three steps: 1a) Given any prediction, we derive an upper bound of CR over all the demand sequences with the same total demand. 1b) We simplify the upper bound obtained in Step 1a) through standard total demand. 1c) We show that the upper bound derived in 1b) is at most $1 + \theta^{-1} + \theta^{-3}$. Note that here we need 
a generalized standard total demand compared to the proof of Theorem~\ref{thm:cr} because LADTSR has two different values of single purchase threshold (i.e., $\theta C_s$ and $C_s/\theta$), while RDTSR only has one threshold value (i.e., $\lambda_s$). 
\end{proof sketch}

\begin{remark} \label{remark:ALG2}
With any $\theta \in (0,1)$, the CR of LADTSR is at most $1 + \theta^{-1} + \theta^{-3}$, regardless of the prediction error $\eta$, implying the worst-case performance guarantees (i.e., robustness). In particular, if LADTSR does not trust the prediction at all (i.e., $\theta \rightarrow 1$), then we have $1 + \theta^{-1} + \theta^{-3} \rightarrow 3$, achieving a similar performance guarantee to that of RDTSR.
\end{remark}

\begin{remark} \label{remark:ALG1}
When the prediction is perfect (i.e., $\eta = 0$), the CR of LADTSR is at most $1+\theta+\theta^2$ as $1 + \theta + \theta^2 < 1 + \theta^{-1} + \theta^{-3}$ for $\theta \in (0,1)$, implying ($1 + \theta + \theta^2$)-consistency.
In particular, if LADTSR fully trusts the prediction (i.e.,  $\theta \rightarrow 0$), then we have $1 + \theta + \theta^2 \rightarrow 1$, achieving a similar performance to the optimal offline algorithm.
\end{remark}

\begin{remark} \label{remark:lambdac}
We now explain the reason why we set the combo purchase threshold $\lambda_c$ as $\theta^2 C_c$ instead of $\theta C_c$. Specifically, we present an example to show that LADTSR fails to achieve $1$-consistency when $\theta \rightarrow 0$ if we choose $\lambda_c=\theta C_c$. Consider $C_c = (K-1)C_s + 1$ and a demand sequence $\Tilde{\mathbf{D}} := \{(1, C_s),(2, C_s), \dots, (K, C_s)\}$. If the prediction is perfect (i.e., $y = (C_s, \dots, C_s)$), then LADTSR sets $\lambda_{s, k} = \theta C_s$ for all the items since $y_k \ge C_s$ (in fact, $y_k = C_s$) and $\lambda_c = \theta C_c$
since $\sum_{k=1}^K \min \{C_s, y_k\} = KC_s > C_c$. 
When the demand $(k, C_s)$ arrives, according to Eqs.~\eqref{psik} and~\eqref{psic}, $\psi_k(k)$ increases by $C_s$, and $\psi_c(k)$ increases by $\min\{\psi_k(k), \lambda_{s, k}\} = \theta C_s$ (recall that we modify Eq.~\eqref{psic} slightly for LADTSR).
So in each time-slot $t$, the overall indicative cost is $\psi_c(t) = \sum_{k=1}^K \min\{\psi_k(t), \lambda_{s, k}\} = t \theta C_s$,
and it reaches $\lambda_c$ only when $t=K$ because $\psi_c(K-1) = (K-1)\theta C_s < \theta (K-1)C_s + \theta = \theta C_c = \lambda_c < K\theta C_s = \psi_c(K)$.
That is, LADTSR will make a combo purchase when the last demand arrives; for the previous $K-1$ demands, LADTSR will make single purchases for these items as their demands exceed the single purchase threshold. 
As a result, the total cost is $(K-1)C_s + C_c$. Since the optimal offline cost is $C_c = (K-1)C_s + 1$ (by assumption), the CR of LADTSR is at least $\frac{(K-1)C_s + C_c}{C_c} = \frac{C_c - 1 + C_c}{C_c} = 2 - \frac{1}{C_c}$, which is larger than 1 for any $\theta$.
\end{remark}


%% file: files/Experiment.tex
\section{Experimental Results}
In this section, we conduct numerical experiments using both synthetic and real-world trace data to evaluate the performance of our proposed algorithms (RDTSR and LADTSR).\footnote{Source code: \url{https://github.com/nauyek/LADTSR}.} 
The results demonstrate the robustness of RDTSR and corroborate that LADTSR offers the desired trade-off between consistency and robustness.

\subsection{Datasets}


\paragraph{Synthetic Dataset.} We consider the same settings adopted in~\citet{wu2021competitive}. Specifically, the demand sequence $\mathbf{D}$ is generated in the following way: First, the time horizon $T$ follows a uniform distribution ranging from 1 to 60. The demand in each time-slot is randomly assigned to one item, following two standard distributions: i) \textbf{Uniform:} the demand is uniformly distributed among all the items; ii) \textbf{Long-tailed:} the demand distribution follows the Pareto law, i.e., $80\%$ of demands are assigned to $20\%$ of items.
The theoretical results in~\citet{wu2021competitive} hold only for unit-demands (i.e., $a(t)=1$ in all time-slots). In our experiments, we also consider the case where multiple units of demand may arrive in each time-slot (i.e., $a(t) \ge 1$). We set the number of items $K = 6$.

\paragraph{Cloud Cost Management Dataset.} We consider the application of cloud service cost management, where a user needs to fulfill the demands from multiple servers via either rental or purchase (single or combo) options. To simulate the demands from multiple servers, we use the traces of Virtual Machine (VM) workloads from Microsoft Azure~\citep{cortez2017resource}. These traces include the average CPU workload every 5 minutes for 30 consecutive days. We randomly shuffle the workloads from multiple servers to generate the demand sequence $\mathbf{D}$. The number of VMs is 9 (i.e., $K=9$). 

\paragraph{Mobile App Usage Dataset.} We consider the application where a user needs to pay for the data usage from multiple mobile Apps. The data usage traces are provided by~\citet{yu2018smartphone}, which records the user ID, App ID, timestamp, and data traffic. We use the trace for multiple Apps from one user in one day as the demand sequence for one problem instance. The number of Apps is 4 (i.e., $K=4$). 

Unless otherwise specified, we set the single purchase price $C_s=9$ and the combo purchase price $C_c=30$ for all the three datasets we consider.


\subsection{Robust Online Algorithms}
In this subsection, we compare RDTSR with the \emph{Deterministic Two-level Ski-Rental (DTSR)} algorithm proposed in~\citet{wu2021competitive}.

\paragraph{Experimental Setup.}
We vary the value of the combo purchase price $C_c$ from 15 to 40. For each value of $C_c$, we run our experiments for $10^4$ independent demand sequences, with $40\%$ being the uniform sequences and $60\%$ being the long-tailed sequences.
All the experiments are implemented in Python and are conducted on a laptop with a 12th Gen Intel(R) i5-12500H processor and 16GB memory.

\begin{figure*}[!t]
\begin{subfigure}{0.24\linewidth}
\centering
\includegraphics[width=0.72\textwidth]{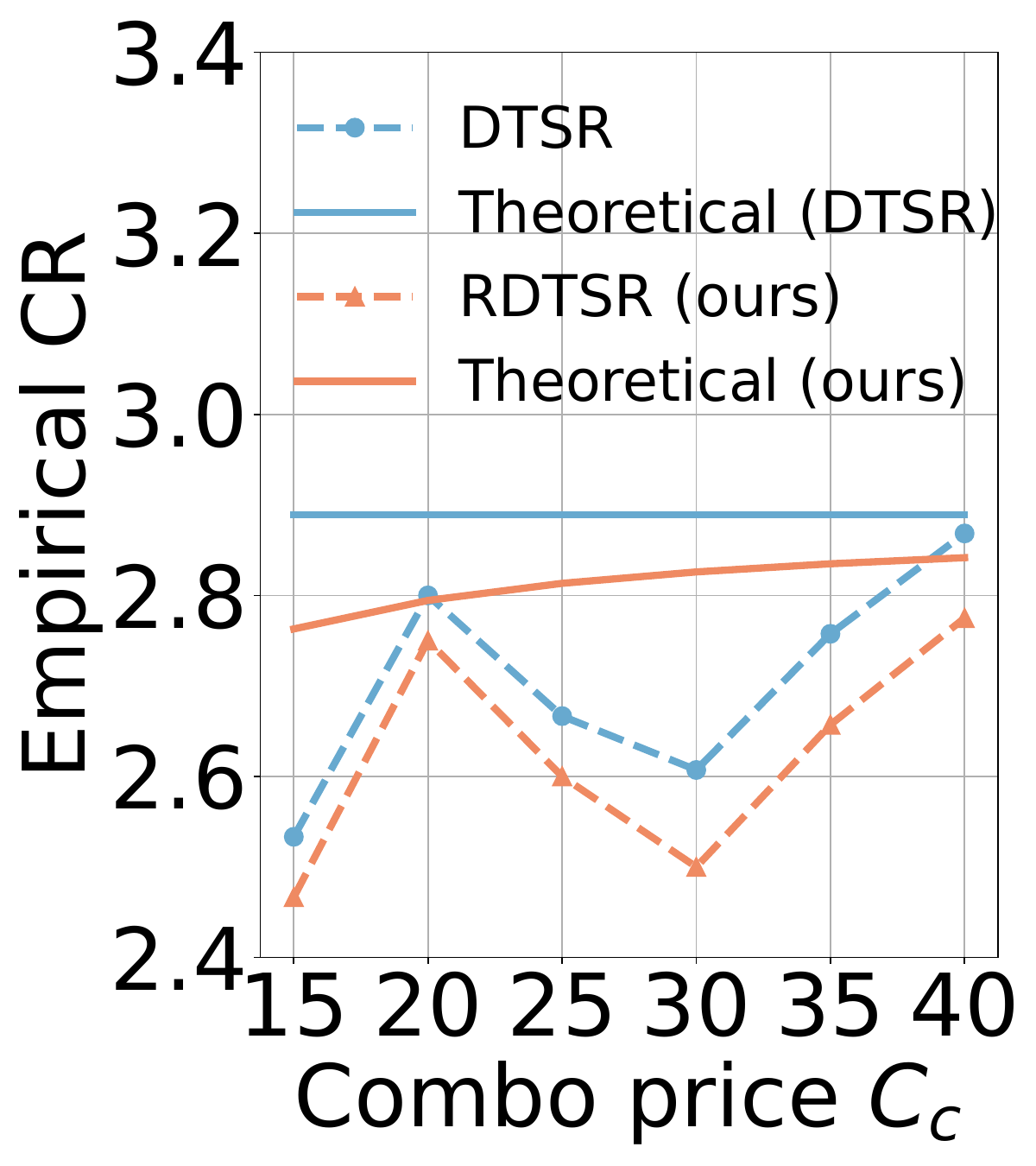}
 \caption{Unit demands}
\label{fig:CR_mix}
\end{subfigure}
\begin{subfigure}{0.24\linewidth}
\centering
\includegraphics[width=0.72\textwidth]{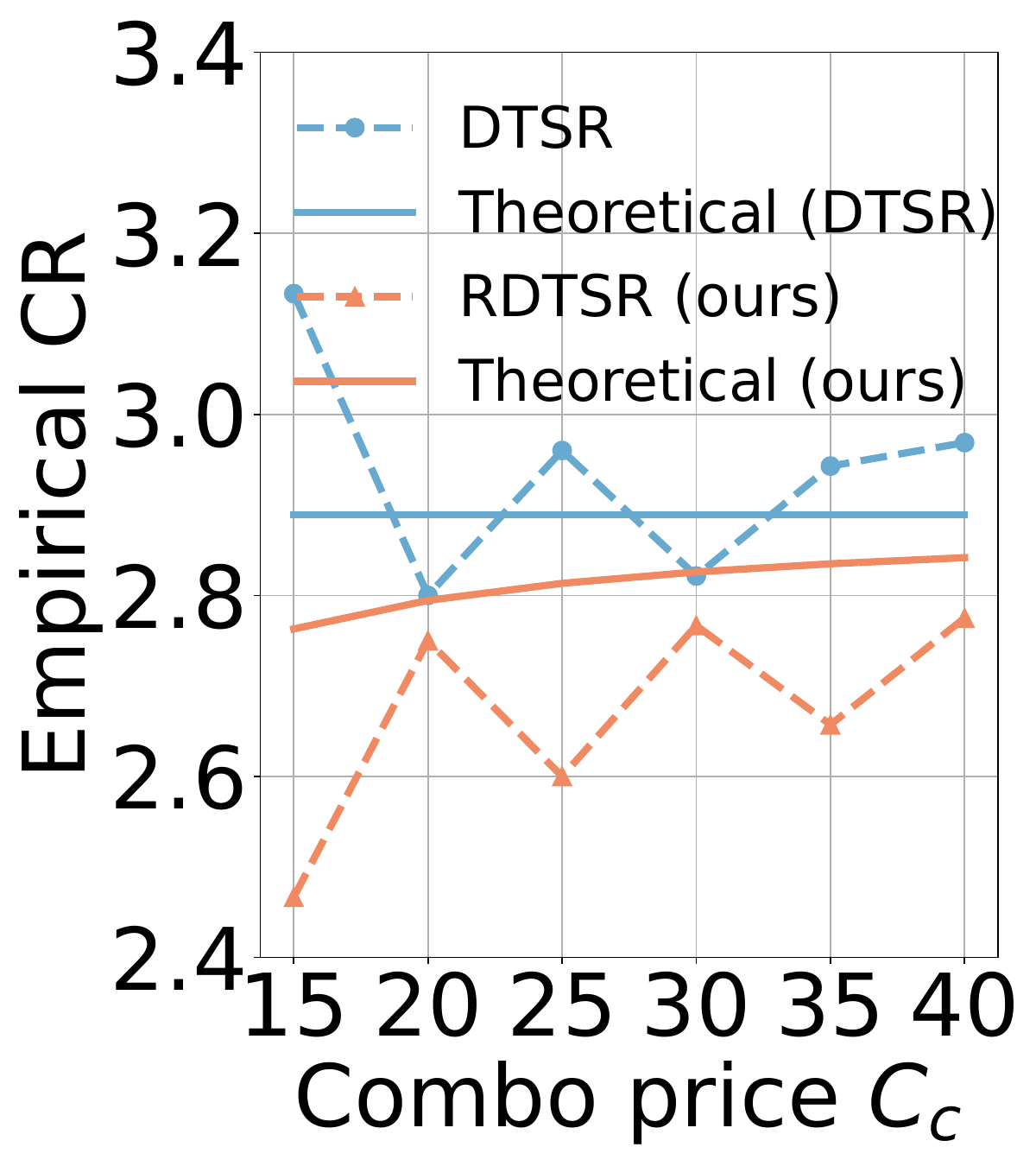}
 \caption{Multi-unit demands}
 \label{fig:CR_mix_multiple}
\end{subfigure}
\begin{subfigure}{0.24\linewidth}
\centering
\includegraphics[width=0.72\textwidth]{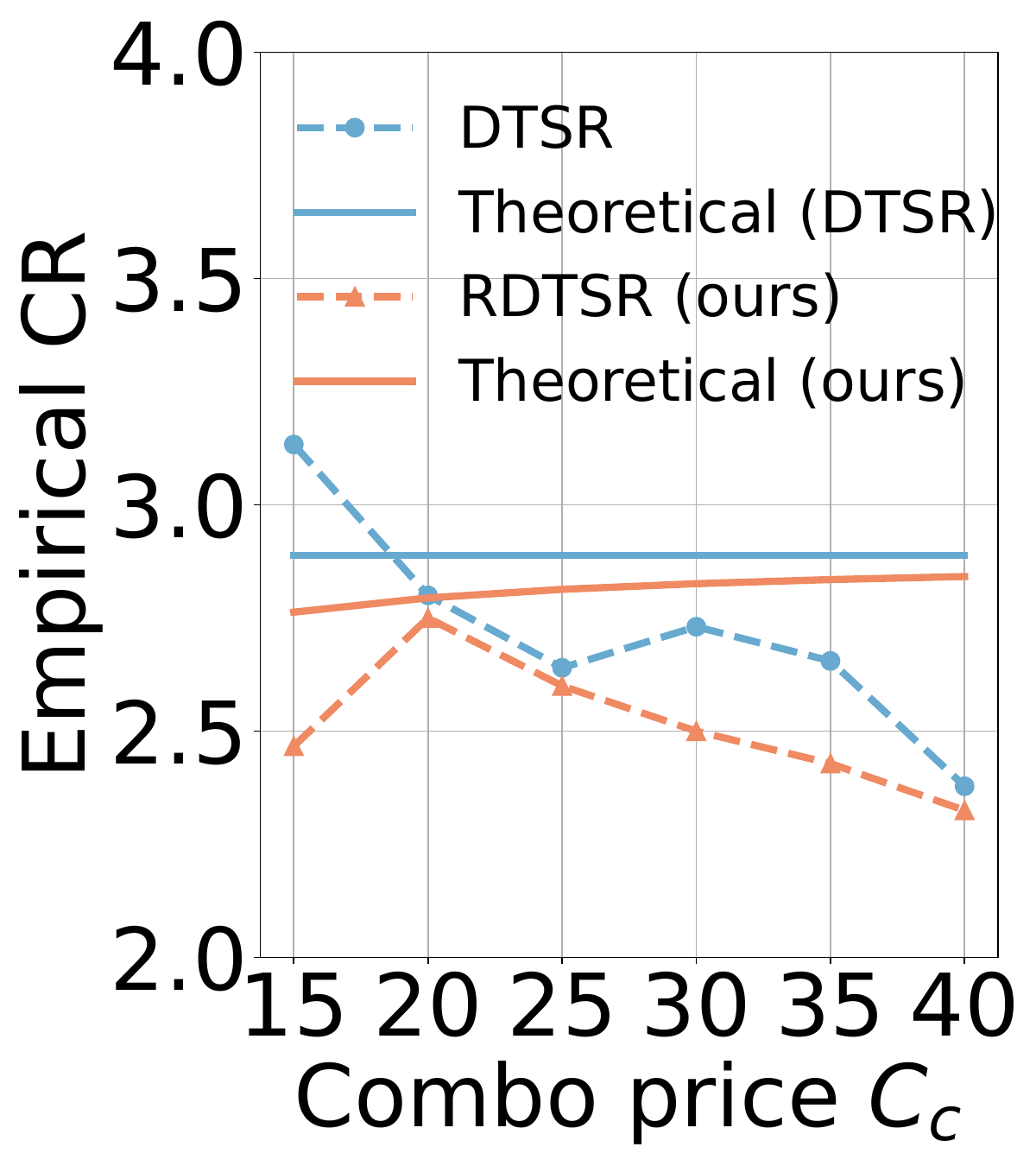}
 \caption{Azure dataset}
 \label{fig:CR_azure}
\end{subfigure}
\begin{subfigure}{0.24\linewidth}
\centering
\includegraphics[width=0.72\textwidth]{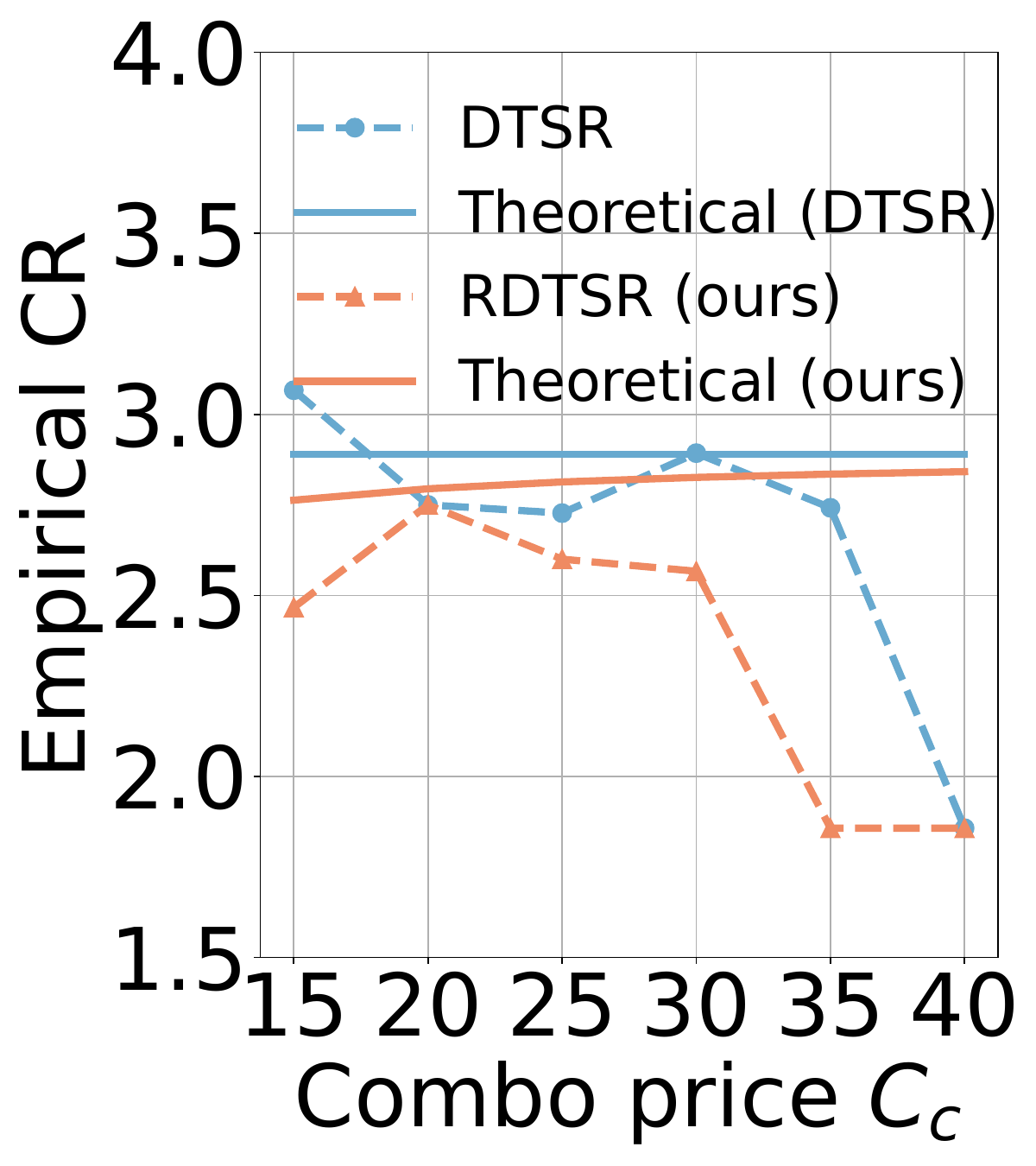}
 \caption{AppUsage dataset}
 \label{fig:CR_AppUsage}
\end{subfigure}
\caption{Empirical CR on the synthetic dataset.}
\label{fig:CR_result}
\end{figure*}

\paragraph{Results.}
Fig.~\ref{fig:CR_result} shows the empirical and theoretical CR of RDTSR and DTSR. Here, the empirical CR represents the worst cost ratio (under the online algorithm and the optimal offline algorithm) over $10^4$ independent demand sequences. We can observe that our algorithm (RDTSR) exhibits a lower empirical 
CR compared to the existing algorithm (DTSR). Furthermore, the empirical CR of RDTSR consistently remains below its theoretical bound, verifying our theoretical results. 
In the setting of multi-unit demands, the empirical CR of DTSR exceeds the theoretical bound (e.g., when $C_c = 15, 25, 35, 40$ in Fig.~\ref{fig:CR_mix_multiple}). This suggests that the CR of DTSR obtained in~\citet{wu2021competitive} does not hold when allowing multiple units of demand to arrive in a time-slot (see Remark~\ref{remark:comparison} for a detailed discussion). 
Similar results can be observed from the two real-world datasets in Fig.~\ref{fig:CR_azure} and Fig.~\ref{fig:CR_AppUsage}. Note that for the AppUsage dataset, when $C_c=40$, the combo purchase option is redundant and the problem boils down to the ski-rental problem (because $K \cdot C_s = 36 < 40$). In this case, the two algorithms are equivalent.

Furthermore, the simulation results also show that RDTSR and DTSR exhibit similar average performance, which is provided in Appendix~\ref{sec:average}.




\subsection{Learning-augmented Online Algorithms}
In this subsection, we study the performance of LADTSR under varying prediction errors. We first explain how to generate ML predictions in our experiments.

\paragraph{ML Prediction Generation.}
Recall that LADTSR receives the total demand for each item as a prediction. We construct an ML task that uses historical total demand information to predict the next total demand for one item. Suppose that there are $N$ rounds of two-level ski-rental problem instances, and for each round $n$, we have a total demand $z^n_k$ for item $k$. Then, we can use a sequence of historical total demand with length $l-1$ (i.e., $z^{n-1}_k, z^{n-2}_k, \dots, z^{n-l+1}_k$) to predict $z^n_k$. To generate a sequence of total demand for each item, for the Synthetic dataset, we generate multiple problem instances (40\% uniform sequences followed by 60\% long-tailed sequences) and then concatenate them together;
such a sequence of total demands for each item can be directly obtained in the two real-world datasets. Then, we use a sliding window to generate multiple sequences of total demand for training and testing. 
The window size $l$ is 10, 513, and 7 for the Synthetic, Azure, and AppUsage datasets, respectively.
The last total demand is the ground truth while the previous sequence of total demand is the historical data. 
The window size is specifically chosen for each dataset due to the performance of the models.
Finally, we use 80\% of the sequences for training. 
The Mean Average Percentage Error (MAPE) of our ML models during the test is at most 5\%. 

To generate predictions with varying qualities, we first train the ML models without any perturbation. Then during the test, we perturb the input of the network by adding a bias $\mu$ that is uniform across the whole sequence of input.
This represents the scenario where the prediction is biased due to the distribution shift between training and testing data.

\paragraph{Experimental Setup.}
We train a \emph{Long Short-term Memory (LSTM)} network to predict the total demand for each item. The network has two LSTM layers followed by one fully connected layer. The hidden states of the two LSTM layers are both 10, 256, and 10 for the Synthetic, Azure, and AppUsage datasets, respectively. We use the mean absolute error as the loss function and employ the Adam optimizer to train the weights. The model is implemented in TensorFlow. For each model, the training process takes about 2~minutes on an NVIDIA
GeForce RTX 3060 Laptop GPU.

To study the impact of trust parameter $\theta$, we evaluate the performance of LADTSR with different values of $\theta \in \{0, 0.25, 0.5, 0.75, 1\}$. Recall that LADTSR with $\theta=0$ and $\theta=1$ corresponds to FTP and RDTSR, respectively. 

\paragraph{Results.}
Fig.~\ref{fig:simu_result} shows the average cost ratio of LADTSR with different values of $\theta$ under varying prediction errors (corresponding to varying values of bias $\mu$). Here the average cost ratio represents the average ratio of cost (under the online algorithm and the optimal offline algorithm) over all the problem instances. Fig.~\ref{fig:simu_syn} shows that while FTP (i.e., LADTSR with $\theta=0$) performs quite well when the prediction is accurate
(i.e., $\mu=0$), its performance degrades significantly when the bias increases.
For LADTSR with $\theta \in \{0.25, 0.5, 0.75\}$, when the predictions are accurate, they all outperform RDTSR (i.e., LADTSR with $\theta=1$), achieving better average performance compared to the robust online algorithm without using any predictions. 
When the prediction is perturbed by a large bias $\mu$ ($\mu \leq -40$ or $\mu \geq 10$), the performance degrades, but not significantly compared to FTP.
Furthermore, with different values of $\theta$, LADTSR provides different trade-off curves for consistency and robustness.
It is noteworthy that the performance of our learning-augmented algorithms degrades abruptly when the actual demand is overestimated ($\mu > 0$), while at a slower rate when it is underestimated. 
This happens because overestimation leads to unnecessary (single or combo) purchases, resulting in an immediate rise in the cost ratio. On the other hand, underestimation incurs additional rental costs before the algorithm makes a purchase, rising gradually as the bias level $|\mu|$ increases.

\begin{figure}[!t] 
\centering 
\begin{subfigure}{0.325\linewidth}
\centering
\includegraphics[width=\textwidth]{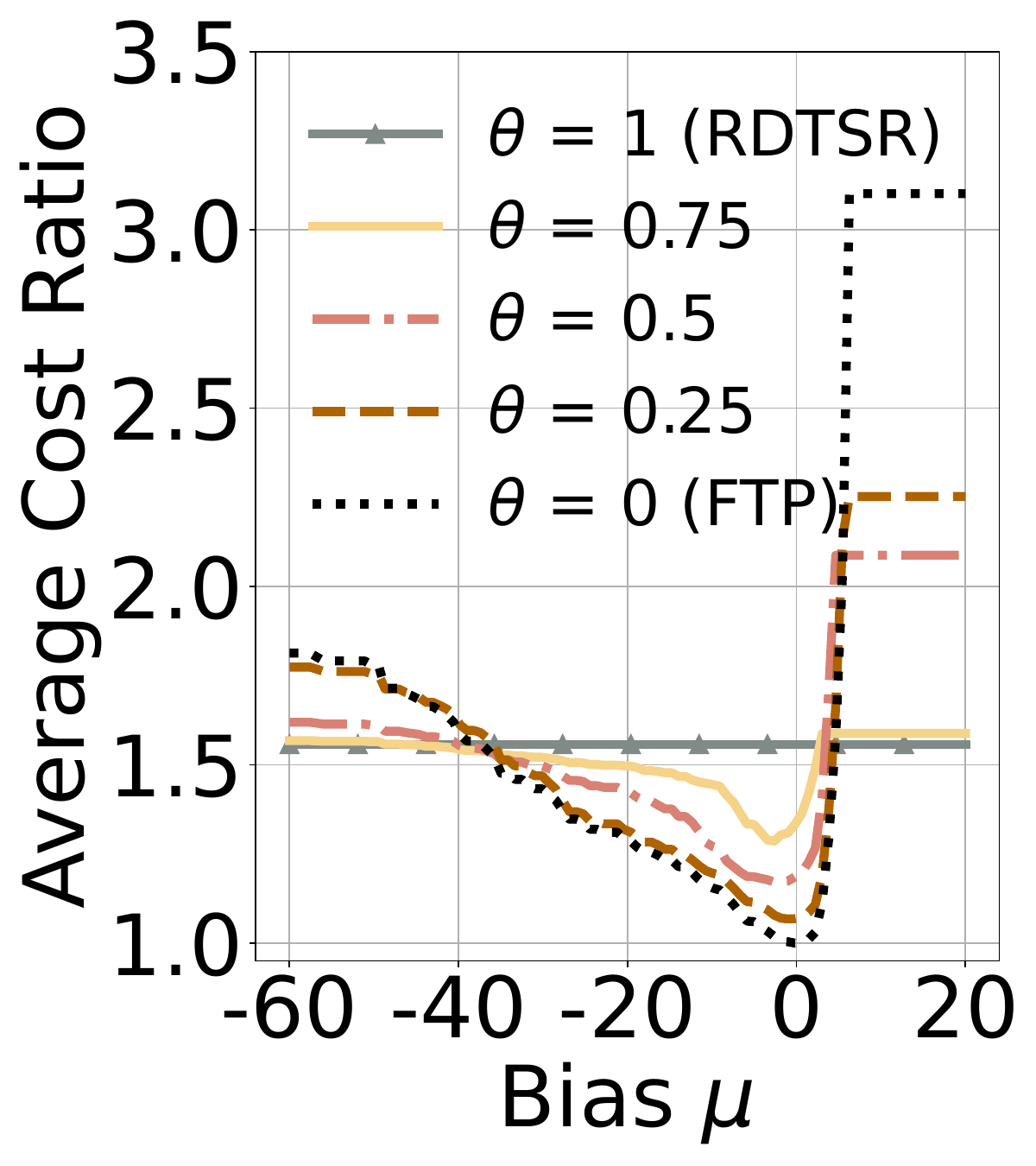}
\caption{Synthetic dataset}
\label{fig:simu_syn}
\end{subfigure}
\begin{subfigure}{0.325\linewidth}
\centering
\includegraphics[width=\textwidth]{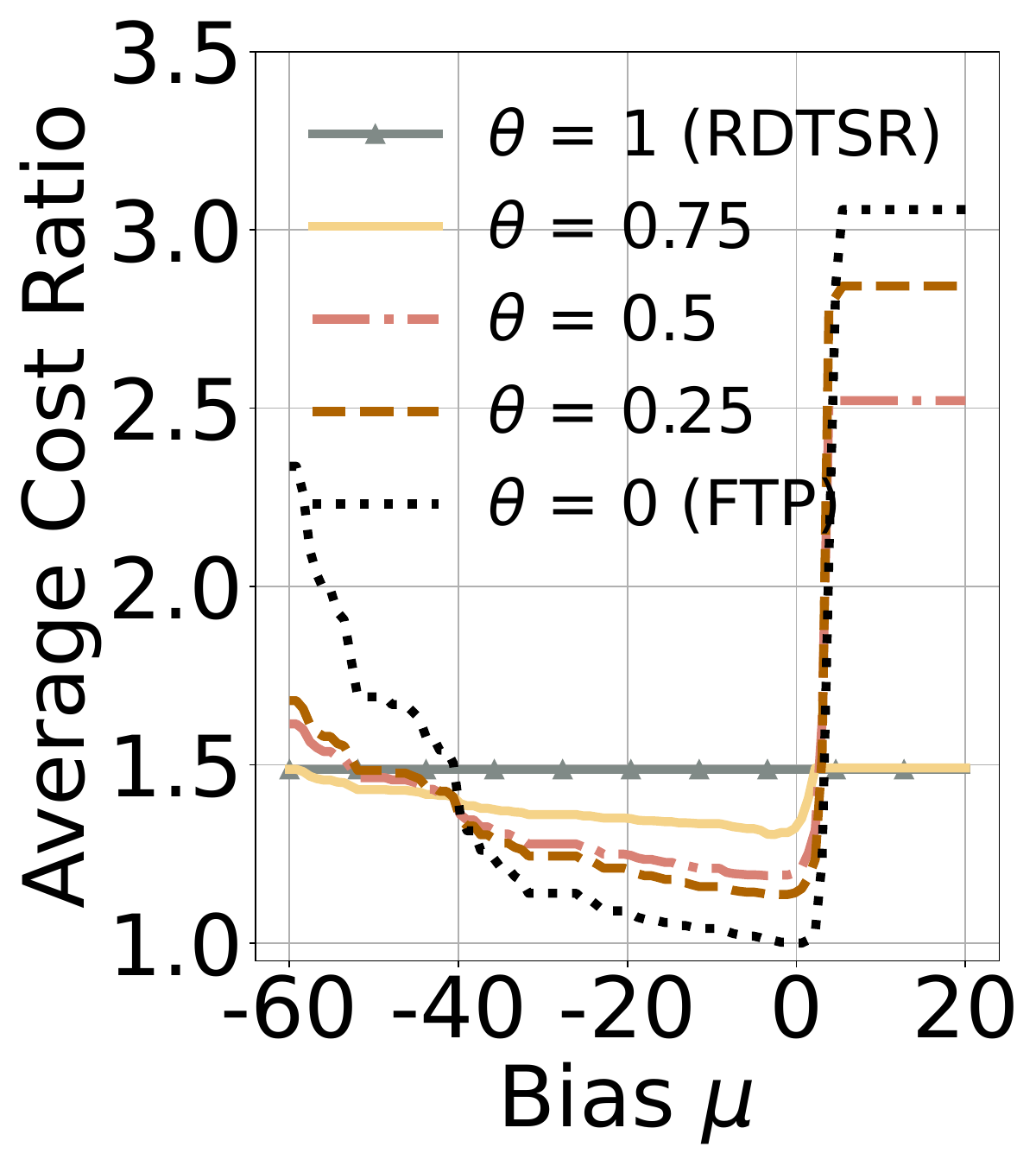}
 \caption{Azure dataset}
 \label{fig:simu_azure}
\end{subfigure}
\begin{subfigure}{0.325\linewidth}
\centering
\includegraphics[width=\textwidth]{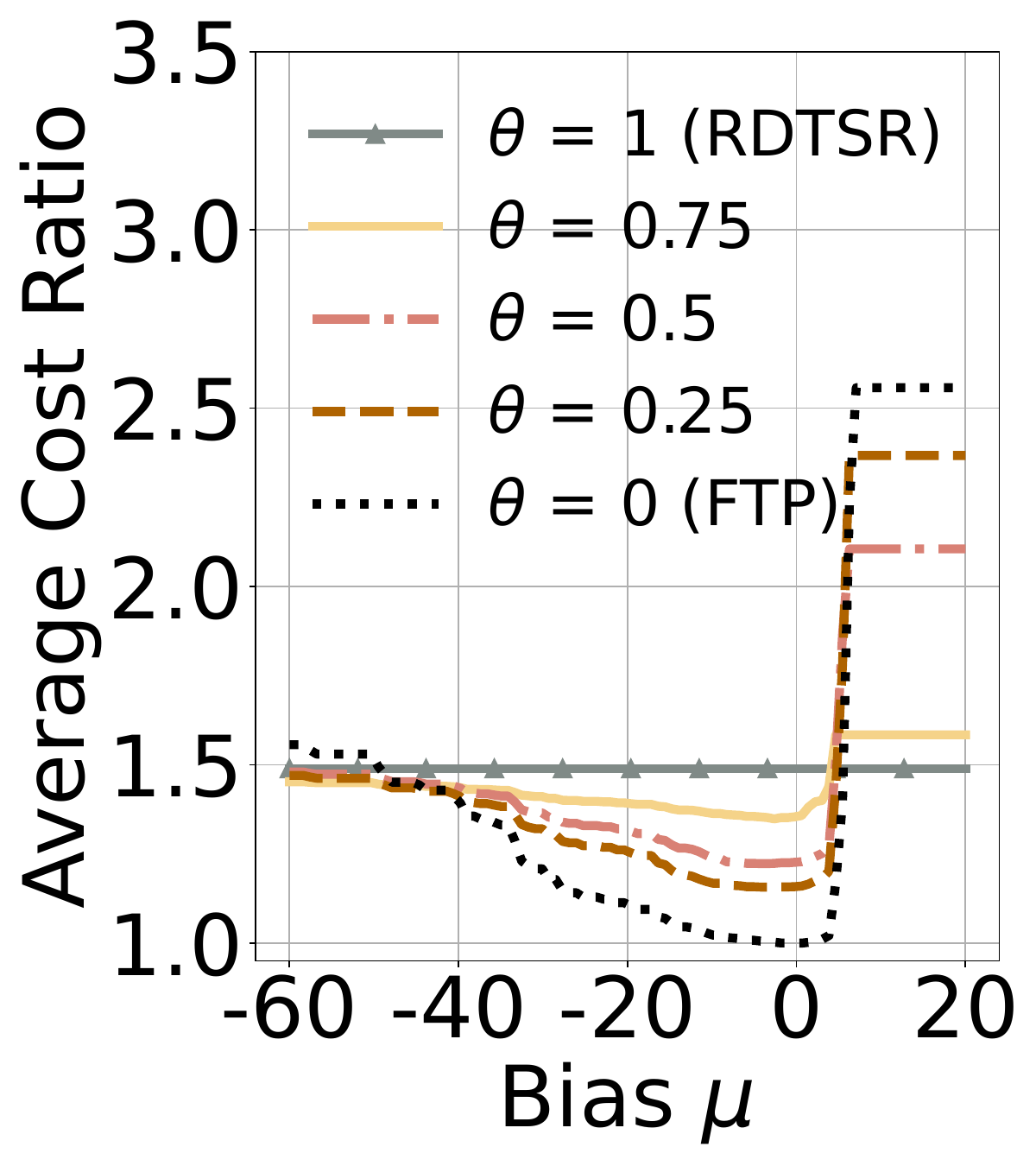}
 \caption{AppUsage dataset}
 \label{fig:simu_app}
\end{subfigure}
\caption{Average cost ratio under varying prediction errors.}
\label{fig:simu_result}
\end{figure}

%% file: files/Conclusion.tex
\section{Conclusion}
Motivated by real-world applications, we investigated the two-level ski-rental problem with multiple payment options and designed a $3$-competitive deterministic online algorithm called RDTSR. By integrating useful ML predictions with RDTSR, we proposed a learning-augmented online algorithm called LADTSR, which achieves both consistency and robustness. 
As for future work, one important direction is to obtain a non-trivial competitive ratio lower bound for the studied problem as well as the optimal trade-off for the learning-augmented algorithm design; another interesting direction is to design learning-augmented algorithms that can achieve better performance by adaptively choosing the trust parameter $\theta$;
it would also be nice to design algorithms that account for the asymmetric impact of overestimation and underestimation in predictions.

\section*{Acknowledgments}
This work is supported in part by the Commonwealth Cyber Initiative (CCI).

%% file: files/supplementary.tex
\appendix
\section{Preliminary Results}
Before we delve into the proof of our main results, we begin by presenting several useful lemmas. These lemmas will be crucial in our analysis of the competitive ratio (CR) of our robust online algorithm (RDTSR) and learning-augmented online algorithm (LADTSR). 

First, recall that in Eq.~\eqref{psik}, the indicative cost 
 $\psi_k(t)$ is defined as 
 \begin{equation} \label{eq:psi_k1}
    \psi_{k}(t) = 
    \begin{cases}
    \psi_k(t-1) + a(t), & \begin{aligned}
        &i(t)=k \ \text{and} \\
        &t < \min\{t_k, t_c\},
    \end{aligned} \\
    \psi_k(t-1),  & \text{otherwise},
    \end{cases} 
\end{equation}
where $\psi_k(t)$ increases by $a(t)$ if the demand in time-slot $t$ is for item $k$ (i.e., $i(t) = k$) and is not covered by (either single or combo) purchases (i.e., $t < \min\{t_k, t_c\}$);
otherwise, it remains unchanged.
We initialize $\psi_k(t)$ as $0$ at the beginning, i.e., $\psi_k(0)=0$.
Instead of using recursion, Lemma~\ref{lemma:psi_k} shows that the indicative cost $\psi_k(t)$ can also be defined iteratively: intuitively, it is the sum of demands for item $k$ during the interval $[1,t]$ when neither a single purchase for item $k$ nor a combo purchase is made. 



\begin{lemma} \label{lemma:psi_k}
The indicative cost $\psi_k(t)$ in Eq.~\eqref{eq:psi_k1} can be equivalently written as 
\begin{equation} \label{eq:psi_k2}
\psi_k(t) \coloneqq \sum_{\tau=1}^t \mathds{1}_{\{i(\tau)=k\}} \cdot \mathds{1}_{\{\tau < t_k\}} \cdot \mathds{1}_{(\tau < t_c)} \cdot a(\tau).
\end{equation}
\end{lemma} 
\begin{proof}


The proof is straightforward. Let $\psi'_k(t)$ evolves as Eq.~\eqref{eq:psi_k2}, we have
\begin{equation}
\psi'_k(t) - \psi'_k(t-1) = \mathds{1}_{\{i(t)=k\}} \cdot \mathds{1}_{\{t < t_k\}} \cdot \mathds{1}_{\{t < t_c\}} \cdot a(t).    
\end{equation}
That is, we have $\psi'_k(t) - \psi'_k(t-1) = a(t)$ if three conditions are met: i) $i(t)=k$, ii) $t<t_k$, iii) $t<t_c$; otherwise $\psi'_k(t) - \psi'_k(t-1) = 0$. Clearly, this is the same as Eq.~\eqref{eq:psi_k1}.

\end{proof}
The following two lemmas
state that the combo purchase decision made by our algorithms only depends on the total demand, rather than the specific order of demand arrival. Recall that the total demand of each demand sequence is denoted by $z = \{z_1, \dots, z_K\}$, where $z_k$ denotes the total demand of item $k$, i.e.,
\begin{equation} \label{eq:zk}
z_k = \sum_{t=1}^T \mathds{1}_{\{i(t)=k\}} \cdot a(t).
\end{equation}

\begin{lemma} \label{lemma::combo}
Given any total demand $z$, LADTSR will make the combo purchase iff \ $\sum_{k=1}^K \min \{ z_k, \lambda_{s, k} \} \ge \lambda_c$. 
\end{lemma}
\begin{proof}
First, we prove the sufficient condition. If LADTSR makes the combo purchase, clearly we have $t_c \le T$. Then, it turns out that we have
\begin{equation}
    \sum_{k=1}^K \min \{z_k, \lambda_{s,k}\} \overset{(a)}{\ge} \psi_c(t_c) \overset{(b)}{\ge} \lambda_c.
\end{equation}
To see this, for inequality $(a)$, we have
\begin{equation}
\begin{aligned}
z_k &\overset{(c)}{=} \sum_{t=1}^T \mathds{1}_{\{i(t)=k\}} \cdot a(t) \\
&\ge  \sum_{t=1}^{t_c} \mathds{1}_{\{i(t)=k\}} \cdot a(t) \\
&\ge \sum_{t=1}^{t_c} \mathds{1}_{\{i(t)=k\}} \cdot  \mathds{1}_{\{t < t_k\}} \cdot \mathds{1}_{\{t < t_c\}}  \cdot a(t) \\
&\overset{(d)}{=} \psi_k(t_c),
\end{aligned}
\end{equation} 
where $(c)$ is due to the definition of $z_k$ in Eq.~\eqref{eq:zk} and $(d)$ is from Lemma~\ref{lemma:psi_k}.
Furthermore, recall that the definition of the overall indicative cost $\psi_c(t)$ for LADTSR is
\begin{equation}
   \psi_c(t) = \sum_{k=1}^K \min \{\psi_k(t), \lambda_{s, k} \}. \label{eq:psic}
\end{equation}
Therefore, we have
\begin{equation*}
    \sum_{k=1}^K \min \{z_k, \lambda_{s,k}\} \ge \sum_{k=1}^K \min \{\psi_k(t_c), \lambda_{s,k}\} = \psi_c(t_c).
\end{equation*}
The first inequality $(a)$ has been proved.

The inequality $(b)$ holds because if LADTSR makes the combo purchase, we have $\psi_c(t_c) \ge \lambda_c$. Combining the inequalities $(a)$ and $(b)$, we conclude the proof for the sufficient condition. 

Next, we prove the necessary condition. When
$\sum_{k=1}^K \min \{ z_k, \lambda_{s,k} \} \ge \lambda_c$, we prove that LADTSR will make the combo purchase by contradiction. Suppose LADTSR does not make the combo purchase (i.e., $t_c>T$), which implies that we have $\psi_c(T) < \lambda_c$ at the end of the demand sequence, according to Line~\ref{line:combo_begin} of RDTSR (holds for LADTSR as well). However, we can show that $\psi_c(T) \ge \lambda_c$, which leads to the contradiction.

To see this, we first show that for each item $k$, we have
\begin{equation}
\min \{\psi_k(T),\lambda_{s,k}\} \ge \min \{z_k, \lambda_{s,k}\}.  
\end{equation}
We consider two cases: whether LADTSR makes the single purchase for item $k$ or not.
\begin{enumerate}[i)]
    \item If LADTSR makes the single purchase for $k$, we have $\psi_k(T) \ge \lambda_{s,k}$ according to Line~\ref{line:single_begin} of RDTSR (holds for LADTSR as well).
Therefore, we have $\min \{\psi_k(T), \lambda_{s,k} \} = \lambda_{s,k} \ge \min \{z_k, \lambda_{s,k}\}$. 
\item Otherwise, if it does not make the single purchase for $k$ (i.e., $t_k>T$), we have
\begin{equation} \label{eq:psik}
\begin{aligned}
    \psi_k(T) &\overset{(a)}{=} \sum_{t=1}^T \mathds{1}_{\{i(t)=k\}} \cdot \mathds{1}_{\{t<t_k\}} \cdot \mathds{1}_{\{t<t_c\}} \cdot a(t) \\
    &\overset{(b)}{=} \sum_{t=1}^T \mathds{1}_{\{i(t)=k\}}\cdot a(t) \overset{(c)}{=} z_k,
\end{aligned}
\end{equation}
where $(a)$ is due to the definition of $\psi_k(t)$ in Eq.~\eqref{eq:psi_k2}, $(b)$ is due to the combo purchase and single purchase for $k$ is never made according to the assumption (i.e., $t_k, t_c > T$), and $(c)$ is due to the definition of $z_k$ in Eq.~\eqref{eq:zk}. In addition, we have $\psi_k(T) \le \lambda_{s,k}$ because the algorithm does not make the combo purchase and the single purchase for any item $k$. Therefore, we have $\min \{\psi_k(T), \lambda_{s,k} \} = \psi_k(T) = z_k \ge \min \{z_k, \lambda_{s,k}\}$.
\end{enumerate}
Now, we have $\min \{\psi_k(T),\lambda_{s,k}\} \ge \min \{z_k, \lambda_{s,k}\}$ for any item $k$ by combining above two cases. 
Finally, we have
\begin{equation} \label{eq:contra}
\begin{aligned}
\psi_c(T) &\overset{(a)}{=} \sum_{k=1}^K \min \{\psi_k(T), \lambda_{s,k}\} \\
&\overset{(b)}{\ge} \sum_{k=1}^K \min \{z_k, \lambda_{s,k}\} \overset{(c)}{\ge} \lambda_c,
\end{aligned}   
\end{equation}
where $(a)$ is due to the definition of $\psi_c(t)$ in Eq.~\eqref{eq:psic}, $(b)$ is due to $\min \{\psi_k(T),\lambda_{s,k}\} \ge \min \{z_k, \lambda_{s,k}\}$ for each item $k$, and $(c)$ is due to the necessary condition.
Eq.~\eqref{eq:contra} leads to the contradiction that $\psi_c(T) < \lambda_c$.
\end{proof}

\begin{lemma} \label{lemma::combo2}
Given any total demand $z$, RDTSR will make the combo purchase iff \ $\sum_{k=1}^K \min \{ z_k, \lambda_{s} \} \ge \lambda_c$.
\end{lemma}
\begin{proof}
Recall that in RDTSR, the single purchase for each item is the same (i.e., $\lambda_s$), which is a specific instance of that of LADTSR by letting $\lambda_{s, 1} = \dots = \lambda_{s, K}=\lambda_s$. Thus, the proof of Lemma~\ref{lemma::combo} also holds in this case.
\end{proof}

The last lemma in this section is to analyze the cost of the optimal offline algorithm. In Lemma~\ref{lemma:opt}, we will show that the cost of the optimal offline algorithm remains the same over all the demand sequences with the same total demand.
\begin{lemma} \label{lemma:opt}
Given any total demand $z$,  the cost of the optimal offline algorithm over all the demand sequences with the same total demand $z$ is the same. Let $\mathit{OPT}(z)$ denote the cost of the optimal offline algorithm on total demand $z$ and we have
\begin{equation} \label{eq:OPTcost}
\mathit{OPT}(z) = \min \{C_c, \sum_{k=1}^K \min \{z_k, C_s\} \}.
\end{equation}
\end{lemma}
\begin{proof}
First, we consider the offline optimal cost without the combo purchase option. In this case, there are two options to fulfill the demand for each item: renting or purchasing. Clearly, the optimal choice for each item $k$ is as follows: when $z_k < C_s$, always rent; otherwise, when $z_k\ge C_s$, purchase at the beginning. Therefore, the total optimal cost for all the items is $\sum_{k=1}^K \min \{z_k, C_s\}$. 

Next, we consider the offline optimal cost with the combo purchase option. As the combo purchase can fulfill demands for all the items, the optimal choice can only be one of the two cases: either making the combo purchase at the beginning or never making the combo purchase at all. In the first case, the cost is $C_c$. In the second case, as we previously analyzed, the optimal cost without the combo purchase is $\sum_{k=1}^K \min \{z_k, C_s\}$. Therefore, the cost of the optimal offline algorithm is one of the minimum values among the two cases, i.e., $\min \{C_c, \sum_{k=1}^K \min \{z_k, C_s\} \}$.  
\end{proof}
\section{Proof of Theorem~\ref{thm:cr}}
\label{app:them:cr}


\begin{proof}

Our goal is to show that by choosing the thresholds $\lambda_s$ and $\lambda_c$ properly, the CR of RDTSR is at most $3 - \frac{1}{C_s} - \frac{1}{C_c}(2-\frac{1}{C_s})$.
In particular, when $\lambda_s = C_s$ and $\lambda_c = C_c$, RDTSR can achieve a CR of $3 - \frac{1}{C_s} - \frac{1}{C_c}(2-\frac{1}{C_s})$. 

The outline of the proof is as follows. 
Generally, our analysis has two steps:
\begin{enumerate}[1)]
    \item We derive the upper bound of CR over all the demand sequences, which is a function of two thresholds: $\lambda_s$ and $\lambda_c$ (see Eq.~\eqref{eq:up}); \\
    \item To minimze the CR by optimizing $\lambda_s$ and $\lambda_c$, we show that when $\lambda_s = C_s$ and $\lambda_c = C_c$, the CR is at most $3 - \frac{1}{C_s} - \frac{1}{C_c}(2-\frac{1}{C_s})$.
\end{enumerate}


\textbf{Step 1):} Since directly analyzing the CR on demand sequences is hard, we break it down into three substeps:
\begin{enumerate}[{1}a)]
    \item We derive the upper bound of CR over all the demand sequences with the same total demand $z$, which is a function of total demand $z$ and two thresholds $\lambda_s, \lambda_c$ (see Lemma~\ref{lemma:ALG}); \\
    \item We define a standard total demand (see Definition~\ref{definition}) to simplify the upper bound we obtained in Step 1a) (see Lemma~\ref{lemma:zstd}). \\
    \item We analyze the upper bound of CR that holds for any standard total demand, which is also the upper bound of CR over all the possible demand sequences. This upper bound is a function of $\lambda_s$ and $\lambda_c$ (i.e., Eq.~\eqref{eq:up}).
\end{enumerate}


\textbf{Step 1a):} We present the definition of CR on a total demand $z$ as follows. Let $\mathcal{D}_z$ denote the set of all the demand sequences that have the same total demand $z$, i.e., $\mathcal{D}_z = \{\mathbf{D} \ | \ \text{$\mathbf{D}$ has a total demand $z$} \}$. Then the CR on a total demand $z$, denoted by $\mathit{CR}(z)$, can be defined as the worst ratio of the cost of the online algorithm to that of the optimal offline algorithm over all the demand sequences with the same total demand $z$, i.e.,
\begin{equation} \label{eq:crz}
    CR(z) \coloneqq \mathop {\max }\limits_{\mathbf{D} \in {\mathcal{D}_z}} \frac{\mathit{ALG}(\mathbf{D})}{\mathit{OPT}(\mathbf{D})},
\end{equation}
where $\mathit{ALG}(\mathbf{D})$ represents the cost of RDTSR and $\mathit{OPT}(\mathbf{D})$ represents the cost of the optimal offline algorithm, both on the demand sequence $\mathbf{D}$.

Furthermore, according to Lemma~\ref{lemma::combo}, the cost of the optimal offline algorithm is the same over all the demand sequences in $\mathcal{D}_z$ (denoted as $\mathit{OPT}(z)$ in Eq.~\eqref{eq:OPTcost}). Therefore, we can rewrite $\mathit{CR}(z)$ as
\begin{equation}
    \mathit{CR}(z) =  \frac{\max_{\mathbf{D} \in \mathcal{D}_z}\mathit{ALG}(\mathbf{D})}{\mathit{OPT}(z)}.
\end{equation}
To find the upper bound for $\mathit{CR}(z)$, we can first focus on the upper bound for the cost of RDTSR on the total demand $z$, denoted by $U_{\mathrm{ALG}}(z)$, i.e.,
\begin{equation}
    \max_{\mathbf{D} \in \mathcal{D}_z} \mathit{ALG}(\mathbf{D}) \le U_{\mathrm{ALG}}(z).
\end{equation}
With $U_{\mathrm{ALG}}(z)$, we can define the upper bound for the CR on total demand $z$ (denoted by $U_{\mathrm{CR}}(z)$) as 
\begin{equation} \label{eq:UCR}
U_{\mathrm{CR}}(z) \coloneqq \frac{U_{\mathrm{ALG}}(z)}{\mathit{OPT}(z)},
\end{equation}
which is greater or equal to $\mathit{CR}(z)$.
To derive the value of $U_{CR}(z)$, we only need to analyze $U_{\mathrm{ALG}}(z)$ because $\mathit{OPT}(z)$ is already obtained in Lemma~\ref{lemma:opt}. 
Given a total demand $z$, although the cost of RDTSR can vary across different demand sequences, from Lemma~\ref{lemma::combo2}, we know that the combo purchase decision only depends on the total demand $z$. Therefore, we can consider two cases: if RDTSR makes the combo purchase or not. We present the result of $U_{\mathrm{ALG}}(z)$ in Lemma~\ref{lemma:ALG}.

\begin{lemma} \label{lemma:ALG}
Given a total demand $z$, the upper bound for RDTSR's cost $U_{\mathrm{ALG}}(z)$ is as follows:
\begin{enumerate}[i)]
    \item If RDTSR does not make the combo purchase, we have
\begin{equation} \label{eq:ALGcase1}
U_{\mathrm{ALG}}(z) = \sum_{k=1}^K U_{\mathrm{ALG}}^k(z),    
\end{equation}
where $U_{\mathrm{ALG}}^k(z)$ denotes the upper bound of the cost of RDTSR for item $k$ on total demand $z$ and we have
\begin{equation}
U_{\mathrm{ALG}}^k(z) = 
\begin{cases} \label{eq:case1}
    z_k, & z_k < \lambda_s,\\
    \lambda_s - 1 + C_s, & \text{otherwise};
\end{cases}     
\end{equation}
\item If RDTSR makes the combo purchase, we have
\begin{equation} \label{eq:case2}
\begin{aligned} 
U_{\mathrm{ALG}}(z) = &\lambda_c - 1 + C_c \\
&+ \min \{\sum_{k=1}^K \mathds{1}_{\{z_k \ge \lambda_s\}}, \frac{\lambda_c - 1}{\lambda_s}\} (C_s - 1).    
\end{aligned}
\end{equation}
\end{enumerate}
\end{lemma}
We prove Lemma~\ref{lemma:ALG} in Appendix~\ref{app:ALG}.

\textbf{Step 1b):} From Eqs.~\eqref{eq:case1} and~\eqref{eq:case2}, we can see that $U_\mathrm{CR}(z)$ depends on the total demand $z$, which has $K$  variables (i.e., $z_1, \dots, z_K$).
When $K$ is large, it is still challenging for us to analyze the CR on the total demand $z$.
Hence, in the following, we establish a special total demand called \emph{standard total demand} $z_{\mathrm{std}}$ to analyze $U_{\mathrm{CR}}(z)$. The standard total demand has two properties: 1) all the standard total demand can be expressed by three variables (see Definition~\ref{definition}), and 2) for any total demand $z$, there exists a standard total demand $z_{\mathrm{std}}$, which satisfies $U_{\mathrm{CR}}(z) \le U_{\mathrm{CR}}(z_{\mathrm{std}})$.

Next, we give the definition of the standard total demand.



\begin{definition} \label{definition}
A total demand $z$ is said to be the standard total demand $z_{\mathrm{std}}(m,n,x)$ if it satisfies:
\begin{enumerate}[i)]
    \item $m$ items have a total demand of at least $C_s$;
    \item $n$ items have a total demand equal to $\lambda_s$;
    \item among all the remaining items, each item has a total demand within the interval $[0, \lambda_s)$. Let $x$ represent the sum of their total demands;
\end{enumerate}
here $m, n, x$ can be any non-negative integers.
\end{definition}


Next, we show a specific way to construct a standard total demand given any total demand $z$. The construction approach is presented in Algorithm~\ref{alg::construct}. First, we start with a total demand $z' = \{z'_1, \dots, z'_K\}$, which equals to $z$ (i.e., $z'_k=z_k$ for item $k=1, \dots, K$). Then, for each item $k$, we modify $z'$ as follows:
\begin{enumerate}[i)]
    \item If $\lambda_s \le z'_k <  C_s$, set $z'_k$ as $\lambda_s$. 
    \item All the remaining items remain the same.
\end{enumerate}
\setcounter{algocf}{2}
\begin{algorithm}[!tb]
\SetAlgoLined
\SetKwInOut{Input}{Input}\SetKwInOut{Output}{Output}
\Input{$z$}
\Output{$z'$}
$z' \leftarrow z$\;

\While{there exists item $k$ s.t. $\lambda_s \le z'_k < C_s$}{
\label{line:loop2start}
$z'_k \gets \lambda_s$\;
\label{line:loop2end}
}
\caption{Construct $z_{\mathrm{std}}(m,n,x)$ from $z$}
\label{alg::construct}
\end{algorithm}

By Lemma~\ref{lemma:zstd}, we claim that for any total demand $z$, following the above construction approach, we can find a standard total demand to bound $U_{\mathrm{CR}}(z)$.

\begin{lemma} \label{lemma:zstd}
Given any total demand $z$, the total demand $z'$ constructed by Algorithm~\ref{alg::construct} is a standard total demand. Furthermore, we have  
\begin{equation}
  U_{\mathrm{CR}}(z) \le U_{\mathrm{CR}}(z').
\end{equation}
\end{lemma}

We prove Lemma~\ref{lemma:zstd} in Appendix~\ref{app:zstd}. 

\textbf{Step 1c):} Given any total demand $z$, we use $z_{\mathrm{std}}(m, n, x)$ to denote the constructed standard total demand. To obtain $U_{\mathrm{CR}}(z_{\mathrm{std}}(m, n, x))$, we first derive $\mathit{OPT}(z_{\mathrm{std}}(m,n,x))$ and $U_{\mathrm{ALG}}(z_{\mathrm{std}}(m,n,x))$, respectively.
For $\mathit{OPT}(z_{\mathrm{std}}(m,n,x))$, 
we can rewrite Eq.~\eqref{eq:OPTcost} as
\begin{equation}\label{eq:OPT}
\begin{aligned}
\mathit{OPT}(z_{\mathrm{std}}(m,n,x)) &= \min \{C_c, \sum_{k=1}^K \min \{z_k, C_s\} \} \\
&=\min \{C_c, mC_s + n\lambda_s +x\},    
\end{aligned}
\end{equation}
where the second term $mC_s + n\lambda_s +x$ is because there are $m$ items with a total demand of at least $C_s$ and the remaining items have a total demand no greater than $C_s$ and the sum of their total demands is $n\lambda_s + x$ (in fact, $n$ items with a total demand equal to $\lambda_s \le C_s$ and the other items with a total demand less than $\lambda_s$).


For $U_{\mathrm{ALG}}(z_{\mathrm{std}}(m,n,x))$, if RDTSR does not make the combo purchase,  
according to Lemma~\ref{lemma:ALG}, 
we have
\begin{equation} \label{eq:ALGcaseA}
U_{\mathrm{ALG}}(z_{\mathrm{std}}(m,n,x)) = (m+n)(\lambda_s - 1 + C_s) + x,
\end{equation}
where the first item is because there are $m+n$ items with a total demand of at least $\lambda_s$, and the second item is due to for the remaining items, the total demand for each item is less than $\lambda_s$ ($\lambda_s \le C_s$) and the sum of their total demands is $x$.
On the other hand, if RDTSR makes the combo purchase, 
according to Lemma~\ref{lemma:ALG},
we have
\begin{equation} \label{eq:ALGcaseB}
\begin{aligned}
U_{\mathrm{ALG}}(z_{\mathrm{std}}(m,n,x))  &= \lambda_c - 1 + C_c + \\
&\quad \min \{m+n, \frac{\lambda_c - 1}{\lambda_s}\} (C_s - 1),    
\end{aligned}
\end{equation}
where the term $m+n$ is due to there are $m+n$ items having a total demand greater or equal to $\lambda_s$.

Next, we show how to find an upper bound of CR over all the standard total demand. For any standard total demand $z_{\mathrm{std}}(m,n,x)$,
we consider two cases: A) RDTSR makes the combo purchase on it, and B) RDTSR does not make the combo purchase on it. 

\textbf{Case A):} In this case, according to Eqs.~\eqref{eq:OPT} and~\eqref{eq:ALGcaseB}, $U_{\mathrm{CR}}$ becomes
\begin{equation}
\begin{aligned}
&U_{\mathrm{CR}}(z_{\mathrm{std}}(m,n,x)) \\
&\quad =  \frac{\lambda_c - 1 + C_c + \min \Big \{m+n,\frac{\lambda_c - 1}{\lambda_s}\Big \} (C_s - 1)}{\min \{C_c, mC_s + n\lambda_s + x\}}.
\end{aligned}   \label{ineq:crB1}
\end{equation}
Because the algorithm makes the combo purchase, according to Lemma~\ref{lemma::combo}, we have $\sum_{k=1}^K \min \{z_k, \lambda_s\} \ge \lambda_c$. Then for the standard total demand, we have
\begin{equation} \label{con:combo}
   \sum_{k=1}^K \min \{z_k, \lambda_s\} = (m + n)\lambda_s + x \ge \lambda_c.
\end{equation}

Observe that we can set $m=0$ without reducing $U_{\mathrm{CR}}$. To see this, for any $z_{\mathrm{std}}(m,n,x)$ with $m\neq0$, there exists another standard total demand $z'_{\mathrm{std}}(m',n',x')$ such that $m'=0$, $n'=m+n$, and $x'=x$. First, we have $(m' + n')\lambda_s + x' = (m + n)\lambda_s + x \ge \lambda_c$, which implies that RDTSR also makes the combo purchase on $z'_{\mathrm{std}}(m',n',x')$ according to Lemma~\ref{lemma::combo}. Hence, we have
\begin{equation}
\begin{aligned}
& U_{\mathrm{CR}}(z'_{\mathrm{std}}(m',n',x')) \\
&\quad = \frac{\lambda_c - 1 + C_c + \min \Big \{m' + n',  \frac{\lambda_c - 1}{\lambda_s} \Big \} (C_s - 1)}{\min \{C_c, m'C_s + n'\lambda_s + x\}} \\
&\quad = \frac{\lambda_c - 1 + C_c + \min \Big \{m+n, \frac{\lambda_c - 1}{\lambda_s}\Big \} (C_s - 1)}{\min \{C_c, (m+n)\lambda_s + x\}} \\
&\quad \overset{(a)}{\ge}
\frac{\lambda_c - 1 + C_c + \min \Big \{m+n, \frac{\lambda_c - 1}{\lambda_s} \Big \} (C_s - 1)}{\min \{C_c, mC_s + n\lambda_s + x\}} \\
&\quad = U_{\mathrm{CR}}(z_{\mathrm{std}}(m,n,x)),
\end{aligned}
\end{equation}
where $(a)$ is due to $(m+n)\lambda_s + x \le mC_s + n\lambda_s + x$ because we assume $\lambda_s \le C_s$.
Therefore, for any $z_{\mathrm{std}}(m,n,x)$ in Case A with $m > 0$, there exists another standard total demand $z'_{\mathrm{std}}(m',n',x')$ in Case A with $m'=0$ and $U_{\mathrm{CR}}(z'_{\mathrm{std}}(m',n',x')) \ge U_{\mathrm{CR}}(z_{\mathrm{std}}(m,n,x))$. 

Hence, Eq.~\eqref{ineq:crB1} is upper bounded by
\begin{equation}
\begin{aligned}
&U_{\mathrm{CR}}(z_{\mathrm{std}}(m,n,x)) \le U_{\mathrm{CR}}(z'_{\mathrm{std}}(m',n',x'))  \\
&\quad =  \frac{\lambda_c - 1 + C_c + \min \Big \{n',\frac{\lambda_c - 1}{\lambda_s} \Big \} (C_s - 1)}{\min \{C_c, n'\lambda_s + x\}}, \label{ineq:crB2}     
\end{aligned}
\end{equation}
and Eq.~\eqref{con:combo} becomes
\begin{equation} \label{con:combo2}
    n'\lambda_s + x \ge \lambda_c.
\end{equation}
According to Eq.~\eqref{con:combo2}, the denominator of Eq.~\eqref{ineq:crB2} is at least $\min \{C_c, \lambda_c\}$. 
We can rewrite Eq.~\eqref{ineq:crB2} as
\begin{equation}
\begin{aligned}
 &U_{\mathrm{CR}}(z_{\mathrm{std}}(m,n,x)) \\
 &\quad \le \frac{\lambda_c - 1 + C_c + \min \Big \{n',\frac{\lambda_c - 1}{\lambda_s} \Big \} (C_s - 1)}{\min \{C_c, \lambda_c\}}.    
\end{aligned}
\label{ineq:crB4}
\end{equation}
As $\min \{n',\frac{\lambda_c - 1}{\lambda_s} \} \le \frac{\lambda_c - 1}{\lambda_s}$, the numerator is at most $\lambda_c - 1 + C_c + \frac{\lambda_c - 1}{\lambda_s} (C_s - 1)$. Therefore, we have
\begin{equation}
U_{\mathrm{CR}}(z_{\mathrm{std}}(m,n,x)) \le  \frac{\lambda_c - 1 + C_c + \frac{\lambda_c - 1}{\lambda_s}(C_s - 1)}{\min \{C_c, \lambda_c\}}. \label{ineq:crB5}
\end{equation}

\textbf{Case B):} In this case, we consider two subcases: I) The optimal offline algorithm does not make the combo purchase. II) It makes the combo purchase.

\textbf{Case B-I):} According to Eqs.~\eqref{eq:OPT} and~\eqref{eq:ALGcaseA}, the upper bound $U_{\mathrm{CR}}(z_{\mathrm{std}})$ becomes
\begin{equation}
\begin{aligned}
U_{\mathrm{CR}}(z_{\mathrm{std}}(m,n,x)) &= \frac{(m+n)(\lambda_s-1+C_s)+x}{mC_s + n\lambda_s + x} \\
&\overset{(a)}{\le} \frac{(m+n)(\lambda_s-1+C_s)+x}{(m+n)\lambda_s + x} \\
&\le \frac{(m+n)(\lambda_s-1+C_s)}{(m+n)\lambda_s}   \\
&= \frac{\lambda_s-1+C_s}{\lambda_s}, \\
\end{aligned} \label{eq:A-I}
\end{equation}
where $(a)$ is due to $\lambda_s \le C_s$.

\textbf{Case B-II):} According to Eqs.~\eqref{eq:OPT} and~\eqref{eq:ALGcaseA}, the upper bound $U_{\mathrm{CR}}(z_{\mathrm{std}})$ becomes
\begin{equation} \label{eq:A-II}
    U_{\mathrm{CR}}(z_{\mathrm{std}}(m,n,x)) = \frac{(m+n)(\lambda_s - 1 + C_s) + x}{C_c}.
\end{equation}
Because RDTSR does not make the combo purchase, according to Lemma~\ref{lemma::combo}, we have
\begin{equation} \label{A-II-condition}
\begin{aligned}
\sum_{k=1}^K \min \{z_k, \lambda_s \} = (m+n)\lambda_s + x \le \lambda_c - 1.
\end{aligned}    
\end{equation}
Substitute Eq.~\eqref{A-II-condition} into~\eqref{eq:A-II}, we have
\begin{equation}
   U_{\mathrm{CR}}(z_{\mathrm{std}}(m,n,x)) \le \frac{\lambda_c - 1 + (m+n)(C_s - 1)}{C_c}. 
\end{equation}
It turns out that we can increase $m$ and $n$ to find another standard total demand $z'_{\mathrm{std}}$ in Case A and have $U_{\mathrm{CR}}(z'_{\mathrm{std}}) \ge U_{\mathrm{CR}}(z_{\mathrm{std}})$. Suppose that $z'_{\mathrm{std}}$ can be expressed by $m'$, $n'$, and $x'$, where $m'\ge m$, $n'\ge n$ and $x'=x$. The upper bound $U_{\mathrm{CR}}(z'_{\mathrm{std}}(m',n',x'))$ can be expressed as
\begin{equation}
\begin{aligned}
&U_{\mathrm{CR}}(z'_{\mathrm{std}}(m,n,x)) \\
&\quad = \frac{\lambda_c - 1 + C_c + \min \{m'+n', \frac{\lambda_c - 1}{\lambda_s}\}(C_s - 1)}{C_c}.    
\end{aligned}
\end{equation}

Clearly, we have $m + n \le m' + n'$; according to Eq.~\eqref{A-II-condition}, we have $m+n \le \frac{\lambda_c - 1}{\lambda_s}$. Therefore, we have $m + n \le \min \{m' + n', \frac{\lambda_c - 1}{\lambda_s}\}$. It directly implies that $U_{\mathrm{CR}}(z'_{\mathrm{std}}(m',n',x')) \ge U_{\mathrm{CR}}(z_{\mathrm{std}}(m,n,x))$.  Furthermore, we can conclude that the upper bound of $U_{\mathrm{CR}}(z_{\mathrm{std}})$ we find for Case A is also the upper bound for Case B-II.

In summary, in Case A and Case B-II the CR is upper bounded by Eq.~\eqref{ineq:crB5}, and in Case B-I it is upper bounded by Eq.~\eqref{eq:A-I}.
Combining two upper bounds, we have the CR which is upper bounded by
\begin{equation} \label{eq:up}
    \max \{\frac{\lambda_s - 1 + C_s}{\lambda_s}, \frac{\lambda_c - 1 + C_c + \frac{\lambda_c - 1}{\lambda_s}(C_s - 1)}{\min \{C_c, \lambda_c\}}\}.
\end{equation}




\textbf{Step 2):} The next step is to minimize Eq.~\eqref{eq:up} by optimizing the thresholds $\lambda_s$ and $\lambda_c$. Firstly, observe that when $\lambda_s$ increases, both terms in the Eq.~\eqref{eq:up} decrease. Therefore, the function takes the minimum when $\lambda_s=C_s$ since we assume $\lambda_s \le C_s$. Therefore, Eq.~\eqref{eq:up} becomes
\begin{equation}
\label{eq:up2}
    \max \{2-\frac{1}{C_s}, \frac{\lambda_c - 1 + C_c + \frac{\lambda_c - 1}{C_s}(C_s - 1)}{\min \{C_c, \lambda_c\}}\}.    
\end{equation}
Next, we analyze the second item in Eq.~\eqref{eq:up2}. We have
\begin{equation}
\label{eq:up3}
\begin{aligned}
    &\frac{\lambda_c - 1 + C_c + \frac{\lambda_c - 1}{C_s}(C_s - 1)}{\min \{C_c, \lambda_c\}} \\
&\quad \overset{(a)}{=} \frac{\lambda_c - 1 + C_c + \frac{\lambda_c - 1}{C_s}(C_s - 1)}{\lambda_c} \\
&\quad = 1 + \frac{C_c-1}{\lambda_c} + \frac{\lambda_c - 1}{\lambda_c}\frac{C_s-1}{C_s} \\
&\quad = 2 - \frac{1}{C_s} + \frac{1}{\lambda_c}(C_c +\frac{1}{C_s} - 2),
\end{aligned}
\end{equation}
where $(a)$ is due to we assume $\lambda_c \le C_c$. 
Because $C_c + C_s^{-1} \ge C_s + C_s^{-1} \ge 2$, Eq.~\eqref{eq:up3} decreases monotonically as $\lambda_c$ increases. When $\lambda_c$ takes the largest (i.e., $\lambda_c=C_c$), Eq.~\eqref{eq:up3} takes the minimum. Therefore, we have
\begin{equation}
    2 - \frac{1}{C_s} + \frac{1}{\lambda_c}(C_c +\frac{1}{C_s} - 2) = 3 - \frac{1}{C_s} - \frac{1}{C_c}(2-\frac{1}{C_s}).
\end{equation} 
Comparing $3 - \frac{1}{C_s} - \frac{1}{C_c}(2-\frac{1}{C_s})$ with $2-\frac{1}{C_s}$, we have
\begin{equation}
\begin{aligned}
    &3 - \frac{1}{C_s} - \frac{1}{C_c}(2-\frac{1}{C_s}) - (2-\frac{1}{C_s}) \\
    &\quad= 1 - \frac{1}{C_c}(2-\frac{1}{C_s}) \\
    &\quad= \frac{1}{C_c}(C_c +\frac{1}{C_s} - 2) \ge 0,
\end{aligned}
\end{equation}
where the last inequality is due to $C_c + C_s^{-1} \ge 2$. Finally, we conclude that when $\lambda_s = C_s$ and $\lambda_c = C_c$, the CR of RDTSR is at most 
\begin{equation}
3 - \frac{1}{C_s} - \frac{1}{C_c}(2-\frac{1}{C_s}).
\end{equation}
\end{proof}

\section{Proof of Lemma~\ref{lemma:ALG}.} \label{app:ALG}
\begin{proof}
We analyze $U_{\mathrm{ALG}}(z)$ in two cases: A) RDTSR does not make the combo purchase and B) RDTSR makes the combo purchase.

\textbf{Case A):} In this case, our goal is  show that if RDTSR does not make the combo purchase (i.e., $t_c > T$), then for any item $k$, we have
\begin{equation*}
U_{\mathrm{ALG}}^k(z) =
\begin{cases}
    z_k, & z_k < \lambda_s,\\
    \lambda_s - 1 + C_s, & \text{otherwise}.
\end{cases}
\end{equation*} 
We consider two subcases: I) the total demand of item $k$ is less than the single purchase threshold (i.e., $z_k < \lambda_s$) and II) the total demand of item $k$ is no less than the single purchase threshold  (i.e., $z_k \ge \lambda_s$).

\textbf{Case A-I):} If $z_k < \lambda_s$, RDTSR will never make the single purchase for item $k$. This is because RDTSR makes the single purchase for item $k$ in some time-slot $t$ only when $\psi_k(t)  \ge \lambda_s$ ($t \le T$). 
However, we have
\begin{equation}
\begin{aligned}
\psi_k(T) &\overset{(a)}{=} \sum_{t=1}^T \mathds{1}_{\{i(t)=k\}} \cdot \mathds{1}_{\{t < t_k\}} \cdot \mathds{1}_{\{t < t_c\}} \cdot a(t) \\
&\le \sum_{t=1}^T \mathds{1}_{\{i(t)=k\}}\cdot a(t) \\
&\overset{(b)}{=}z_k < \lambda_s,
\end{aligned}    
\end{equation}
where $(a)$ is due to Lemma~\ref{lemma:psi_k} and $(b)$ is due to the definition of $z_k$ (i.e., Eq.~\eqref{eq:zk}). 
Consequently, RDTSR will not make the single purchase for item $k$ and covers all the demand for item $k$ by rental with a cost of $z_k$. Moreover, we have $\psi_k(T)=z_k$ because neither the single purchase nor the combo purchase is made (i.e., $t_k > T$ and $t_c > T$).


\textbf{Case A-II):} If $z_k \ge \lambda_s$, then we can prove that RDTSR will make the single purchase for item $k$ by contradiction. 
Suppose that RDTSR does not make the single purchase for $k$, then we have $\mathds{1}_{\{t < t_k\}} = 1$ for any time-slot $t$.
Similarly, as we assumed that RDTSR does not make the combo purchase, we have $\mathds{1}_{\{t < t_c\}} = 1$ for any time-slot $t$.
Therefore, according to Lemma~\ref{lemma:psi_k}, 
we have $\psi_k(T) = \sum_{t=1}^T \mathds{1}_{\{i(t)=k\}}\cdot a(t) = z_k$. Since $z_k \ge \lambda_s$, we know that RDTSR will make the single purchase for item $k$, which is constracting with our assumption. Therefore, we conclude that RDTSR will make the single purchase for item $k$,  and the corresponding single purchase cost for item $k$ is $C_s$. 
Next, we analyze the rental cost of RDTSR for item $k$ before it makes the single purchase. Note that the rental cost for item $k$ cannot exceed $\psi_k(t)$ because any amount of rental cost for item $k$ will be first added to $\psi_k(t)$ according to the algorithm. Therefore, the rental cost is at most $\lambda_s - 1$ because if it is greater than $\lambda_s - 1$ (i.e., greater or equal to $\lambda_s$) in some time-slot $t$, then 
we have $\psi_k(t) \ge \lambda_s$, which means that RDTSR will make the single purchase for item $k$ in that time-slot. 
In summary, the total cost for item $k$ is $\lambda_s - 1 + C_s$.

Now we derive $U_{\mathrm{ALG}}^k(z)$ as
\begin{equation} \label{cost:fraction}
U_{\mathrm{ALG}}^k(z) = 
\begin{cases}
    z_k, & z_k < \lambda_s,\\
    \lambda_s - 1 + C_s, & \text{otherwise},
\end{cases}
\end{equation}
and the corresponding $U_{\mathrm{ALG}}(z)$ is
\begin{equation} 
U_{\mathrm{ALG}}(z) = \sum_{k=1}^K U_{\mathrm{ALG}}^k(z).    
\end{equation}

\textbf{Case B):} We want to show that if RDTSR makes the combo purchase,  we have
\begin{equation} \label{eq:case2result}
\begin{aligned}
U_{\mathrm{ALG}}(z) = &\lambda_c - 1 + C_c \\
&+ \min \{\sum_{k=1}^K \mathds{1}_{\{z_k \ge \lambda_s\}}, \frac{\lambda_c - 1}{\lambda_s}\} (C_s - 1).    
\end{aligned}    
\end{equation}
In time-slot $t_c$, we know that RDTSR makes the combo purchase with a cost of $C_c$. No cost is incurred after the time-slot $t_c$ because all the following demands are covered by the combo purchase. 
Next, we analyze the single purchase cost and the rental cost before $t_c$.

For the single purchase cost, let $n_s$ denote the number of single purchases made by RDTSR, and the total single purchase cost is $n_s C_s$. 

For the total rental cost, it turns out that the rental cost of RDTSR during the interval $[1,t_c-1]$ is at most $\lambda_c - 1 - n_s$. To see this, in the following, we first show that the total rental cost until time-slot $(t_c - 1)$ is at most $\psi_c(t_c - 1) - n_s$, and then we show that $\psi_c(t_c - 1) \le \lambda_c - 1$.



First, we show that the total rental cost until time-slot $(t_c - 1)$ is at most $\psi_c(t_c - 1) - n_s$. As RDTSR does not make the combo purchase during the interval $[1,t_c-1]$, we can directly apply the result in Case A) for the demand sequence during the interval $[1, t_c-1]$. Specifically, for any demand sequence $\mathbf{D} = \{d(t)\}_{t=1}^T \in \mathcal{D}_z$ with the total demand $z$, let $\mathbf{D}'$ denote the subsequence during $[1, t_c-1]$, i.e, $\mathbf{D}' = \{d(t)\}_{t=1}^{t_c-1}$. Let $z'$ denote the total demand of $\mathbf{D}'$. For each item $k$, we consider two cases: $z'_k < \lambda_s$ and $z'_k \ge \lambda_s$.
\begin{enumerate}[i)]
    \item If $z'_k < \lambda_s$, according to the analysis in Case A-I), we have the rental cost is $z'_k$ and $\psi_k(t_c-1)=z'_k$. Since $z'_k < \lambda_s$, the rental cost for item $k$ is $z'_k = \psi_k(t_c-1) < \min \{\psi_k(t_c - 1), \lambda_s\}$.
\item If $z'_k \ge \lambda_s$, according to the analysis in Case A-II), RDTSR will make the single purchase for item  $k$ (i.e., $\psi_k(t_c - 1) \ge \lambda_s$) and the rental cost is at most $\lambda_s - 1$. Therefore, the rental cost is at most $\lambda_s - 1 = \min \{\psi_k(t_c - 1), \lambda_s\} - 1$.
\end{enumerate}
Because we assume that the number of single purchases is $n_s$, therefore the total rental cost is at most $\sum_{k=1}^K \min \{\psi_k(t_c - 1), \lambda_s\} - n_s$.
According to the definition of $\psi_c(t)$ (i.e., Eq.~\eqref{eq:psic}), we have $\psi_c(t_c - 1)=\sum_{k=1}^K \min \{\psi_k(t_c - 1), \lambda_s \}$. Therefore, the total rental cost for all the items is at most $\psi_c(t_c - 1) - n_s$.



Then we show that 
\begin{equation} \label{eq:case2temp}
  \psi_c(t_c - 1) \le  \lambda_c - 1. 
\end{equation}
This is because if $\psi_c(t_c - 1) > \lambda_c - 1$ (i.e., $\psi_c(t_c - 1) \ge \lambda_c$), then according to Line~\ref{line:combo_begin} of RDTSR, it will make the combo purchase in time-slot $t_c - 1$ rather than $t_c$. This leads to the contradiction.
Therefore, we conclude that the total rental cost is at most $\lambda_c - 1 - n_s$.

In summary, if RDTSR makes the combo purchase and $n_s$ single purchases, its total cost is at most
\begin{equation}
\begin{aligned} \label{eq:result}
\lambda_c - 1 -n_s + C_c + n_s C_s \\
 = \lambda_c - 1 + C_s + n_s(C_s - 1).
\end{aligned}
\end{equation}
Next, we identify the value of $n_s$. Clearly, $U_{\mathrm{ALG}}(z)$ increases monotonically as $n_s$ increases. Therefore, we need to identify the possible maximum value of $n_s$. First, the count of single purchases cannot exceed the number of items for which its total demand exceeds $\lambda_s$. Then, we have
\begin{equation}
    n_s \le \sum_{k=1}^K \mathds{1}_{\{z_k \ge \lambda_s\}}.
\end{equation}
On the other hand, we can show that $n_s \le (\lambda_c - 1)/\lambda_s$. 
To see this, we only need to show that $n_s\lambda_s \le \psi_c(t_c - 1) \le \lambda_c - 1$.
The first inequality holds because if RDTSR makes the single purchase for item $k$, we have $\psi_k(t_c - 1) \ge \lambda_s$, which means $\lambda_s = \min \{\psi_k(t_c - 1), \lambda_s\}$. Therefore, for these $n_s$ items, we have $n_s\lambda_s \le \sum_{k=1}^K \min \{\psi_k(t_c - 1), \lambda_s \} = \psi_c(t_c - 1)$.
We have shown that $\psi_c(t_c - 1) \le \lambda_c - 1$ in Eq.~\eqref{eq:case2temp}, so the second inequality holds.

Now, we obtain the upper bound of $n_s$, which is  $n_s\le \min \{\sum_{k=1}^K \mathds{1}_{\{z_k \ge \lambda_s\}}, (\lambda_c - 1)/\lambda_s\}$. The upper bound of $\mathit{ALG}(z)$ is then
\begin{equation} \label{eq:cost_fraction_combo}
\begin{aligned}
U_{\mathrm{ALG}}(z) = &\lambda_c - 1 + C_c \\
&+ \min \{\sum_{k=1}^K \mathds{1}_{\{z_k \ge \lambda_s\}}, \frac{\lambda_c - 1}{\lambda_s}\} (C_s - 1).    
\end{aligned}    
\end{equation}
   
\end{proof}


\section{Proof of Lemma~\ref{lemma:zstd}} \label{app:zstd}
\begin{proof}
First, we show that the output total demand $z'$ is a standard total demand. According to Algorithm~\ref{alg::construct}, there is one loop, i.e., Lines~\ref{line:loop2start}-\ref{line:loop2end}. After this loop, no item has a total demand within the interval $(\lambda_s, C_s)$. 
In other words, the total demand can fall into one of three categories: i) at least $C_s$, ii) equal to $\lambda_s$, or iii) within the interval $[0, \lambda_s)$. Clearly, the constructed total demand $z'$ satisfies the three conditions in Definition~\ref{definition} at the same time.

Second, we show that $U_{\mathrm{CR}}(z') \ge U_{\mathrm{CR}}(z)$. Initially, we have $z'=z$. Then we only need to show that after the loop, $U_{\mathrm{CR}}(z')$ does not decrease. As we will see in Lemma~\ref{std2}, for each iteration in this loop, $U_{\mathrm{CR}}(z')$ does not decrease.
\end{proof}

\begin{lemma} \label{std2}
Given any total demand $z$, if there exists an item $j$ such that $\lambda_{s} \le z_j < C_s$, then there exists another total demand $z'=\{z'_1, \dots, z'_K\}$ such that
\begin{equation} \label{std2:z'}
z'_k = 
\begin{cases}
\lambda_s, &k=j,\\
z_k, &k\neq j,
\end{cases}
\end{equation}
and we have $U_{\mathrm{CR}}(z) \le U_{\mathrm{CR}}(z')$.
\end{lemma}
\begin{proof}
By the definition of $U_{\mathrm{CR}}(z)$,  we can proceed the proof in the following two steps: 1) we show that $\mathit{OPT}(z) \ge \mathit{OPT}(z')$, and 2) we show that $U_{\mathrm{ALG}}(z) = U_{\mathrm{ALG}}(z')$.

\textbf{Step 1):} For $\mathit{OPT}(z)$, we have
\begin{equation}
\begin{aligned}
    \mathit{OPT}(z) &\overset{(a)}{=} \min \{C_c,\sum_{k=1}^K \min \{ C_s, z_k\} \} \\
    &\overset{(b)}{=} \min \{C_c, z_j + \sum_{\substack{k=1, k \neq j}}^K \min \{ C_s, z_k\} \} \\
    &\overset{(c)}{\ge} \min \{C_c, z'_j + \sum_{\substack{k=1, k \neq j}}^K \min \{ C_s, z'_k\} \} \\
    &\overset{(d)}{=} \min \{C_c, \sum_{k=1}^K \min \{ C_s, z'_k\} \}  \overset{(e)}{=} \mathit{OPT}(z'),
\end{aligned}
\end{equation}
where $(a)$ comes from the definition of $\mathit{OPT}(z)$ in Eq.~\eqref{eq:OPTcost}, $(b)$ is due to $\min \{C_s, z_j\} = z_j$ because $z_j \le C_s$, $(c)$ is due to $z'_k \le z_k$ for any item $k$ according to Eq.~\eqref{std2:z'}, 
$(d)$ is due to $\min \{C_s, z'_j\} = z'_j$ because $z'_j = \lambda_s \le C_s$, and $(e)$ comes from Lemma~\ref{lemma:opt}.

\textbf{Step 2):} Then, we show that $U_{\mathrm{ALG}}(z) = U_{\mathrm{ALG}}(z')$.
To see this, we first show that $\sum_{k=1}^K \min \{z_k, \lambda_s\} = \sum_{k=1}^K \min \{z'_k, \lambda_s\}$, which implies that RDTSR's combo purchase decisions are the same for $z$ and $z'$ according to Lemma~\ref{lemma::combo}. This is due to two facts: 1) For item $j$, we have $\min \{z_j, \lambda_s \} = \min \{ z'_j, \lambda_s \} = \lambda_s$ because $z_j, z'_j \ge \lambda_s$. 2) For item $k$ that $k\neq j$, we have $z_k = z'_k$ according to Eq.~\eqref{std2:z'}. Therefore, we conclude that $\sum_{k=1}^K \min \{z_k, \lambda_s\} = \sum_{k=1}^K \min \{z'_k, \lambda_s\}$ because $\min \{z_k, \lambda_s\} = \min \{z'_k, \lambda_s\}$ for any item $k$.
Next, to compare $U_{\mathrm{ALG}}(z)$ and $U_{\mathrm{ALG}}(z')$, we consider two cases: whether RDTSR makes the combo purchase on $z$ and $z'$ or not.

If RDTSR makes the combo purchase, according to Lemma~\ref{lemma:ALG} we have $U_{\mathrm{ALG}}(z) = \lambda_c - 1 + C_c + \min \{\sum_{k=1}^K \mathds{1}_{\{z_k \ge \lambda_s\}}, \frac{\lambda_c - 1}{\lambda_s}\} (C_s - 1)$.
We only need to compare the term related to $z$ and $z'$ (i.e., $\sum_{k=1}^K \mathds{1}_{\{z_k \ge \lambda_s\}}$ and $\sum_{k=1}^K \mathds{1}_{\{z'_k \ge \lambda_s\}}$).
It turns out that $\sum_{k=1}^K \mathds{1}_{\{z_k \ge \lambda_s\}} = \sum_{k=1}^K \mathds{1}_{\{z'_k \ge \lambda_s\}}$ because for item $j$, we have $z_j, z'_j \ge \lambda_s$ (i.e., $\mathds{1}_{\{z_j \ge \lambda_s\}} = \mathds{1}_{\{z'_j \ge \lambda_s\}} = 1$), and for item $k$ that $k\neq j$, we have $z_k = z'_k$ according to Eq.~\eqref{std2:z'} (i.e., $\mathds{1}_{\{z_k \ge \lambda_s\}} = \mathds{1}_{\{z'_k \ge \lambda_s\}}$).
Therefore, in this case, we have $U_{\mathrm{ALG}}(z) = U_{\mathrm{ALG}}(z')$.

If RDTSR does not make the combo purchase, according to Lemma~\ref{lemma:ALG}, we have
\begin{equation*}
U_{\mathrm{ALG}}^k(z) = 
\begin{cases}
    z_k, & z_k < \lambda_s,\\
    \lambda_s - 1 + C_s, & \text{otherwise}.
\end{cases}
\end{equation*}
We claim that $U_{\mathrm{ALG}}^k(z) = U_{\mathrm{ALG}}^k(z')$ for $k=1,\dots,K$. For item $j$, we have $z_j, z'_j \ge \lambda_s$ and $U_{\mathrm{ALG}}^j(z) = U_{\mathrm{ALG}}^j(z') = \lambda_s - 1 + C_s$; for item $k$ such that $k\neq j$, we have $z_k = z'_k$ according to Eq.~\eqref{std2:z'} and  hence $U_{\mathrm{ALG}}^k(z) = U_{\mathrm{ALG}}^k(z')$. Therefore, we have $U_{\mathrm{ALG}}(z) = \sum_{k=1}^K U_{\mathrm{ALG}}^k(z) =  U_{\mathrm{ALG}}^k(z')$.

Combining two cases, we have $U_{\mathrm{ALG}}(z) = U_{\mathrm{ALG}}(z')$.
\end{proof}

\section{Proof of Lemma~4.1} \label{app:lemma-4.1}
\begin{proof}
First, we claim that given any total demand $z$ and a prediction $y$, the cost of FTP over all the demand sequences $\mathbf{D} \in {\mathcal{D}_z}$ is the same. We use $\mathit{ALG}(z)$ to denote the cost of FTP on a total demand $z$. There are two cases for FTP: If the prediction suggests the combo purchase, i.e., $\sum_{k=1}^K \min \{C_s, y_k \} \ge C_c$, FTP will make the combo purchase in $t=1$.
Otherwise, FTP will never make the combo purchase. 
For each item $k$, FTP will make the single purchase for $k$ if the prediction suggests so (i.e., $y_k \ge C_s$), otherwise FTP always pays the demand for $k$ by rental. 
Therefore, we have
\begin{align} \label{eq:lemmaALG}
    \mathit{ALG}(z) = 
    \begin{cases}
    C_c, & \sum_{k=1}^K \min(C_s, y_k) \ge C_c ,\\
    \sum_{k=1}^K{\mathit{ALG}_k(z)}, & \text{otherwise},
    \end{cases}
\end{align}
where $\mathit{ALG}_k(z)$ denotes the cost of FTP for item $k$ when it does not make the combo purchase, i.e.,
\begin{align}
    \mathit{ALG}_k(z) = 
    \begin{cases}
    C_s, & y_k \ge C_s ,\\
    z_k, & \text{otherwise}.
    \end{cases}
\end{align}
We can see that the cost of FTP only depends on the actual total demand $z$ and the predicted total demand $y$. For the optimal offline algorithm, according to Lemma~\ref{lemma:opt}, we have
\begin{equation} \label{eq:lemmaOPT}
\mathit{OPT}(z) = \min \{C_c, \sum_{k=1}^K \min \{z_k, C_s\} \}.
\end{equation}

Therefore, we can analyze the cost of FTP and the cost of the optimal offline algorithm based on the total demand. Specifically, given a total demand of $z$, we consider two cases: A) FTP makes the combo purchase on $z$, and B) FTP does not make the combo purchase on $z$. 
Similarly, we consider two cases
for the optimal offline algorithm: I)  the optimal offline algorithm makes the combo purchase, and II) the optimal offline algorithm does not make the combo purchase. Therefore,
we consider $2 \times 2 = 4$ cases and we show that the
result holds for each case.


\textbf{Case A-I):} In this case, we have $\mathit{ALG}(z)=C_c$ and $\mathit{OPT}(z)=C_c$ from Eqs.~\eqref{eq:lemmaALG} and~\eqref{eq:lemmaOPT}, respectively. As $\eta \ge 0$, we have $\mathit{ALG}(z) \le \mathit{OPT}(z) + \eta$.

\textbf{Case A-II):} In this case, we have $\mathit{ALG}(z)=C_c$ and $\mathit{OPT}(z)=\sum_{k=1}^K \min \{z_k, C_s\}$ from Eqs.~\eqref{eq:lemmaALG} and~\eqref{eq:lemmaOPT}, respectively.
Next, we show that for each item $k$, we have
\begin{align}
    \min \{z_k, C_s\} + \eta_k \ge \min \{y_k, C_s \}. \label{ineq:A-II}
\end{align}
We show this for $z_k \ge C_s$ and $z_k < C_s$, respectively: 
\begin{enumerate}[i)]
    \item If $z_k \ge C_s$, we have $\min \{z_k, C_s\} + \eta_k = C_s + \eta_k \ge C_s \ge \min \{y_k, C_s \}$.
    \item If $z_k < C_s$, we have $\min \{z_k, C_s\} + \eta_k = z_k + \eta_k \ge y_k \ge \min \{y_k, C_s\}$.
\end{enumerate}
Then we have
\begin{equation}
    \begin{aligned}
        \mathit{ALG}(z) &= C_c \\
        &\overset{(a)}{\le}  \sum_{k=1}^K \min \{y_k, C_s\} \\
        &\overset{(b)}{\le} \sum_{k=1}^K  (\min \{z_k, C_s\} +\eta_k) \\
        &\le \mathit{OPT}(z) + \eta,
    \end{aligned}
\end{equation}
where $(a)$ is due to the condition that the prediction suggests the combo purchase (i.e., Eq.~\eqref{eq:lemmaALG}) and $(b)$ is due to Eq.~\eqref{ineq:A-II}.


\textbf{Case B-I):} In this case, we have $\mathit{ALG}(z) = \sum_{k=1}^K{\mathit{ALG}_k(z)}$ and $\mathit{OPT}(z)=C_c$ from Eqs.~\eqref{eq:lemmaALG} and~\eqref{eq:lemmaOPT}, respectively. We show that for each item $k$, we have
\begin{equation}
    ALG_k(z) \le \min \{C_s, y_k\} + \eta_k. \label{ineq:B-I}
\end{equation}
We show this for $y_k \ge C_s$ and $y_k < C_s$, respectively:
\begin{enumerate}[i)]
    \item If $y_k \ge C_s$, we have $\mathit{ALG}_k(z) = C_s = \min \{C_s, y_k\}\le \min\{C_s, y_k\} + \eta_k$,
    \item If $y_k < C_s$, we have $ALG_k(z) = z_k \le y_k + \eta_k = \min \{C_s, y_k\} + \eta_k $.   
\end{enumerate}
Then we have
\begin{equation}
\begin{aligned}
    \mathit{ALG}(z) &= \sum_{k=1}^K \mathit{ALG}_k(z) \\
    & \overset{(a)}{\le} \sum_{k=1}^K(\min \{C_s, y_k\} +  \eta_k)\\
    & \overset{(b)}{<} C_c + \eta = OPT + \eta,
\end{aligned}
\end{equation}
where $(a)$ comes from the Eq.~\eqref{ineq:B-I} and $(b)$ comes from the condition that the prediction does not suggest the combo purchase (i.e., $\sum_{k=1}\min\{C_s, y_k\} < C_c$).





\textbf{Case B-II):} In this case, we have $\mathit{ALG}(z)=\sum_{k=1}^K\mathit{ALG}_k(z)$ and $\mathit{OPT}(z)=\sum_{k=1}^K\min \{z_k, C_s\}$ from Eqs.~\eqref{eq:lemmaALG} and~\eqref{eq:lemmaOPT}, respectively. Without the combo purchase, for each item $k$, this problem degenerates into a ski-rental problem (i.e., $z_k$ is the ski days and $C_s$ is the purchase price). 
In this case, FTP will make the single purchase for item $k$ if the predicted total demand $y_k \ge C_s$ or keep renting for all the demand for item $k$ if the predicted total demand $y_k < C_s$. In fact, FTP is exactly Algorithm~1 in~\citet{purohit2018improving}. Therefore, we can apply the Lemma~2.1 proved in~\citet{purohit2018improving}, which states that the cost of the algorithm is less than the sum of the offline optimal and the prediction error, i.e., 
\begin{equation}
    \mathit{ALG}_k(z) \le \min\{z_k, C_s\} + \eta_k.
\end{equation}
Finally, by summing up all the items, we obtain
\begin{equation}
\begin{aligned}
 \mathit{ALG}(z) = \sum_{k=1}^K\mathit{ALG}_k(z) &\le \sum_{k=1}^K\min\{z_k, C_s\} + \eta_k \\
 &= \mathit{OPT}(z) + \eta.    
\end{aligned}
\end{equation} 

Combining all the four subcases, we complete the proof.

\end{proof}


\section{Proof of Consistency} \label{sec:consistency}
\begin{proof}
In this section, we prove the first upper bound $(1 + \theta + \theta^2) + \frac{1+2\theta}{1-\theta}\frac{\eta}{\mathit{OPT}(\mathbf{D})}$. The outline of the proof is as follows. Our goal is to show that this upper bound holds for any demand sequence under any predictions. First, given any predictions, we consider two cases: if $\lambda_c = \theta^2 C_c$ and if $\lambda_c = C_c/\theta$. For each case, given any demand sequence, similar to the proof of Lemma~4.1, we consider four subcases. 
That is, we consider two subcases for LADTSR: A) LADTSR makes the combo purchase. B) LADTSR does not make the combo purchase. 
For each subcase A and B, we consider two subcases for the optimal offline algorithm: I) The optimal offline algorithm makes the combo purchase. II) The optimal offline algorithm does not make the combo purchase. Therefore, we consider $8$ subcases. 
For each subcase, we will derive the upper bound of CR, which is a function of the trust parameter $\theta$ and the prediction error $\eta$. 


Before we delve into each subcase, similar to the proof for Theorem~\ref{thm:cr} and Lemma~\ref{lemma:FTP}, we need to analyze the cost of LADTSR on a total demand $z$. We still use $\mathit{CR}(z)$ to denote the CR of LADTSR, $U_{\mathrm{ALG}}(z)$ to denote the upper bound of the cost of LADTSR, and $U_{\mathrm{CR}}(z)$ to denote the upper bound of CR of LADTSR, respectively. Therefore, we have
\begin{equation}
   \mathit{CR}(z) \le U_{\mathrm{CR}}(z) = \frac{U_{\mathrm{ALG}}(z)}{\mathit{OPT}(z)}.
\end{equation}

Next, we analyze $U_{\mathrm{ALG}}(z)$. Similar to Lemma~\ref{lemma:ALG}, we present the result of $U_{\mathrm{ALG}}(z)$ for LADTSR in Lemma~\ref{lemma:LA-ALG}.

\begin{lemma} \label{lemma:LA-ALG}
Given a total demand $z$, the upper bound for LADTSR's cost $U_{\mathrm{ALG}}(z)$ is as follows:
\begin{enumerate}[i)]
    \item If LADTSR does not make the combo purchase,
\begin{equation} \label{eq:laALGcase1}
U_{\mathrm{ALG}}(z) = \sum_{k=1}^K U_{\mathrm{ALG}}^k(z),    
\end{equation}
where $U_{\mathrm{ALG}}^k(z)$ denotes the upper bound for LADTSR's cost on item $k$ and we have
\begin{equation} \label{eq:lacase1}
U_{\mathrm{ALG}}^k(z) = 
\begin{cases}
    z_k, & \text{if} \ z_k < \lambda_{s,k},\\
    \lambda_{s,k} + C_s, & \text{otherwise}.
\end{cases}     
\end{equation}
\item If LADTSR makes the combo purchase, let $s_1$ denote the number of items whose single purchase is $C_s/\theta$ and its total demand for this item is no less than the single purchase $C_s/\theta$ (i.e., those items $k$ such that $\lambda_{s, k}=C_s/\theta$ and $z_k \ge \lambda_{s, k}=C_s/\theta$), and let $s_2$ denote the number of items whose single purchase is $\theta C_s$ and its total demand for this item is no less than the single purchase $\theta C_s$ (i.e., those items $k$ such that $\lambda_{s, k}=\theta C_s$ and $z_k \ge \lambda_{s, k}=\theta C_s$),
then we have
\begin{equation} \label{eq:LAcase2}
U_{\mathrm{ALG}}(z) = \lambda_c + C_c +f(s_1, s_2)C_s,
\end{equation}
where $f(s_1, s_2)$ is a function that calculates the maximum number of single purchases and is computed by
\begin{equation}
f(s_1, s_2) = \min \{\frac{\lambda_c}{\theta C_s}, s_2 + \frac{\lambda_c -\theta C_s s_2}{C_s/\theta}, s_1+s_2 \}.    
\end{equation}
\end{enumerate}
\end{lemma}
We prove Lemma~\ref{lemma:LA-ALG} in Appendix~\ref{app:LA-ALG}.

Next, we consider the case when $\lambda_c = \theta^2 C_c$, which means the prediction suggests a combo purchase. According to Line~\ref{line:setthreshold_combo} in LADTSR, we have 
\begin{equation} \label{eq:combo-cond}
    \sum_{k=1}^K \min \{y_k, C_s\} \ge C_c.
\end{equation}

\textbf{Case A-I):} In this case, we have $\mathit{OPT}(z)=C_c$. For LADTSR, according to Lemma~\ref{lemma:LA-ALG}, we have 
\begin{equation}\label{eq:ALGtheta^2} 
\begin{aligned}
U_{\mathrm{ALG}}(z) &= \lambda_c + C_c + f(s_1, s_2)C_s \\
&\overset{(a)}{\le }\lambda_c + C_c + \frac{\lambda_c}{\theta C_s}C_s \\
&= \theta^2C_c + C_c + \theta C_c \\
&= (1 + \theta + \theta^2) C_c \\
&\le  (1 + \theta + \theta^2) (\mathit{OPT}(z) + \eta),
\end{aligned}    
\end{equation}
where $(a)$ is due to $f(s_1, s_2)\le \frac{\lambda_c}{\theta C_s}$. Therefore, we have
\begin{equation}
    \mathit{CR}(z) \le \frac{U_{\mathrm{ALG}}(z)}{\mathit{OPT}(z)} \le (1 + \theta + \theta^2) (1 + \frac{\eta}{\mathit{OPT}(z)}).
\end{equation}

\textbf{Case A-II):} In this case, from Lemma~\ref{lemma:opt} we have $\mathit{OPT}(z)=\sum_{k=1}^K \min\{C_s, z_k\}$. For LADTSR, from  Eq.~\eqref{eq:ALGtheta^2}, we have
\begin{equation}
 U_{\mathrm{ALG}}(z) \le (1 + \theta + \theta^2) C_c.   
\end{equation}
Next, we will show that $C_c \leq \mathit{OPT}(z) + \eta$. This is because
\begin{equation}
\begin{aligned}
C_c &\overset{(a)}{\leq} \sum_{k=1}^K \min \{ y_k, C_s \} \\
    & \overset{(b)}{\leq} \sum_{k=1}^K (\min \{z_k, C_s \} + \eta_k)  \\
    & = \mathit{OPT} + \eta,
\end{aligned}
\end{equation}
where $(a)$ is due to Eq.~\eqref{eq:combo-cond} and $(b)$ is shown as follows:
\begin{itemize}
    \item If $z_k \ge C_s$, we have $\min \{z_k, C_s\} + \eta_k = C_s + \eta_k \ge C_s \ge \min \{y_k, C_s \}$.
    \item If $z_k < C_s$, we have $\min \{z_k, C_s\} + \eta_k = z_k + \eta_k \ge y_k \ge \min \{y_k, C_s\}$.
\end{itemize}
Therefore we have,
\begin{equation}
    \mathit{CR}(z) \le \frac{U_{\mathrm{ALG}}(z)}{\mathit{OPT}(z)} = (1 + \theta + \theta^2) (1 + \frac{\eta}{\mathit{OPT}(z)}).
\end{equation}

\textbf{Case B-I):} In this case, we will show that such a total demand does not exist. First, we have
\begin{equation} \label{eq:notexist}
\begin{aligned}
\sum_{k=1}^K \min \{z_k, \lambda_{s,k} \} 
&\overset{(a)}{\geq}  \sum_{k=1}^K \min \{z_k,\theta C_s \} \\
&\geq \sum_{k=1}^K \min \{ \theta z_k ,\theta C_s \} \\
&= \theta \sum_{k=1}^K\min \{ z_k, C_s \} \\
&\overset{(b)}{\geq} \theta C_c \geq \theta^2 C_c = \lambda_c,
\end{aligned}
\end{equation}
where $(a)$ is due to $\lambda_{s,k} \geq \theta C_s$ for any item $k$ and $(b)$ is due to the optimal offline algorithm making the combo purchase. Based on Lemma~\ref{lemma::combo}, Eq.~\eqref{eq:notexist} implies that LADTSR will make the combo purchase, leading to the contradiction that it does not make the combo purchase in this case.

\textbf{Case B-II):} In this subcase, the problem degenerates into the ski-rental problem. 
Furthermore, our LADTSR is equivalent to the LADTSR proposed in~\citet{purohit2018improving}. That is, buy on day $\ceil{\theta C_s}$ if $y_k \ge C_s$, otherwise buy on day $\ceil{C_s / \theta}$.
Therefore, we can directly apply the Theorem~2.2 in their work, i.e., for each item $k$, the cost of LADTSR for item $k$ (denoted by $\mathit{ALG}_k(z)$) is at most $(1 + \theta) \mathit{OPT}_k(z) + \frac{1}{1-\theta} \eta_k$, where $\mathit{OPT}_k(z)$ is the cost of the optimal offline algorithm for item $k$. As the total cost of both algorithms in this case is the sum of their cost for each item, we have
\begin{equation}
\begin{aligned}
\mathit{CR}(z) &= \frac{\sum_{k=1}^K \mathit{ALG}_k(z)}{\sum_{k=1}^K \mathit{OPT}_k(z)} \\
& \le \frac{(1+\theta)\sum_{k=1}^K \mathit{OPT}_k(z) + \frac{1}{1-\theta} \sum_{k=1}^K \eta_k}{\sum_{k=1}^K \mathit{OPT}_k(z)} \\
&= 1+\theta + \frac{1}{1-\theta}\frac{\eta}{\mathit{OPT}(z)}.
\end{aligned}
\end{equation}

Combining all the four subcases, we have
\begin{equation}
\begin{aligned}
\mathit{CR}(z) \le &\max \{1 + \theta + \theta^2, 1 + \theta\} \\
&\quad + \max \{1 + \theta + \theta^2, \frac{1}{1-\theta}\} \frac{\eta}{\mathit{OPT}(z)}.
\end{aligned}
\end{equation}
Clearly, we have $1+\theta \le 1 + \theta + \theta^2$. Comparing $1 + \theta + \theta^2$ and $\frac{1}{1-\theta}$, we have
\begin{equation}
\begin{aligned}
&\frac{1}{1-\theta} - (1 + \theta + \theta^2) \\
&\quad = \frac{1-(1+\theta+\theta^2)(1-\theta)}{1-\theta} \\
&\quad = \frac{1-(1-\theta^3)}{1-\theta} \\
&\quad = \frac{\theta^3}{1-\theta} > 0.
\end{aligned}    
\end{equation}
Therefore, for the case when $\lambda_c = \theta^2 C_c$, we have
\begin{equation}  \label{eq:result-theta^2}
\mathit{CR}(z) \le 1 + \theta + \theta^2 + \frac{1}{1-\theta} \frac{\eta}{\mathit{OPT}(z)}.
\end{equation}

Now we consider another case when $\lambda_c = C_c/\theta$. According to Lines~{\ref{line:setthreshold_combo}-\ref{line:setthreshold_end}} in LADTSR, we have
\begin{equation} \label{eq:combo-cond2}
    \sum_{k=1}^K \min \{C_s, y_k\} < C_c.
\end{equation}

We follow the above proof and consider four subcases. 

\textbf{Case A-I):} In this case, we have $\mathit{OPT}(z) = C_c$ and LADTSR makes the combo purchase. First, we need to find the connection between the $\mathit{OPT}(z)$ and the prediction error $\eta$. It turns out that we can show that
\begin{align}
    C_c \leq \frac{\theta}{1 - \theta} \eta. \label{ineq::eta}
\end{align}
Suppose we have $n$ items (index from $1$ to $n$) with single purchase threshold $\theta C_s$ and $m$ items (index from $n+1$ to $n+m$) with threshold $C_s / \theta$. Clearly, we have $n+m=K$. 

Because LADTSR makes the combo purchase, according to Lemma~\ref{lemma::combo}, we have $\sum_{k=1}^K \min \{z_k, \lambda_{s,k} \} \geq \lambda_c$. Therefore, we have
\begin{equation}
    \sum_{k=1}^n \min \{z_k, \theta C_s\} + \sum_{k=n+1}^{n+m} \min \{z_k, \frac{C_s}{\theta}  \} \geq   \frac{C_c}{\theta}.
\end{equation}
Move $\sum_{k=1}^n \min \{z_k, \theta C_s\}$ to the RHS and we have,
\begin{equation}
\begin{aligned}
\sum_{k=n+1}^{n+m} \min \{z_k, \frac{C_s}{\theta}  \} & \geq   \frac{C_c}{\theta} - \sum_{k=1}^n \min \{z_k,  \theta C_s \} \\
&\overset{(a)}{\geq}   \frac{C_c}{\theta}   - n\theta C_s\\
& \geq \frac{1-\theta}{\theta}C_c + (C_c - n \theta C_s  )\\
& \geq \frac{1-\theta}{\theta}C_c + (C_c - n C_s), \label{ineq:yk0}   
\end{aligned}
\end{equation}
where $(a)$ is due to $\min \{z_k, \theta C_s\} \leq \theta C_s$. 

On the other hand, according to Lines~\ref{line:setthreshold_begin}-\ref{line:setthreshold_single} in LADTSR, for any item $k$, if $\lambda_{s, k}=C_s/\theta$, we have $y_k < C_s$; if $\lambda_{s, k} = \theta C_s$, we have $y_k \ge C_s$. Then, Eq.~\eqref{eq:combo-cond2} becomes
\begin{equation} \label{eq:nC_s}
\sum_{k=1}^K \min \{C_s, y_k\} = nC_s + \sum_{k=n+1}^{n+m}y_k < C_c.
\end{equation}
Move $nC_s$ to the RHS, we have
\begin{equation}
\sum_{k=n+1}^{n+m}y_k < C_c - nC_s. \label{ineq::yk}
\end{equation}
Substitute Eq.~\eqref{ineq::yk} into~\eqref{ineq:yk0} we have,
\begin{equation}
\sum_{k=n+1}^{n+m} \min \{z_k, \frac{C_s}{\theta} \} \geq \frac{1-\theta}{\theta}C_c + \sum_{k=n+1}^{n+m} y_k. \label{ineq::yk1}
\end{equation}
On the other hand, we have
\begin{equation}
    \sum_{k=n+1}^{n+m} \min \{z_k, \frac{C_s}{\theta}\} \leq \sum_{k=n+1}^{n+m} z_k \leq \sum_{k=n+1}^{n+m}  (y_k + \eta_k),   \label{ineq::yk2}
\end{equation}
Combining Eqs.~\eqref{ineq::yk1} and~\eqref{ineq::yk2} we obtain,
\begin{equation}
\begin{aligned}
\frac{1-\theta}{\theta}C_c + \sum_{k=n+1}^{n+m} y_k &\leq \sum_{k=n+1}^{n+m}  (y_k + \eta_k), \\
    \frac{1-\theta}{\theta}C_c &\leq \sum_{k=n+1}^{n+m} \eta_k \leq \eta. \label{ineq::2-b1}    
\end{aligned}
\end{equation}
Therefore, we show that $C_c \leq \frac{\theta}{1-\theta}\eta$. 

Next, we analyze $U_{\mathrm{ALG}}(z)$. According to Lemma~\ref{lemma:LA-ALG}, as $f(s_1, s_2) \le s_2 + \frac{\lambda_c - \theta C_s s_2}{C_s/\theta}$, we have
\begin{equation}
\begin{aligned}
U_{\mathrm{ALG}}(z) &\le \lambda_c + C_c + (s_2 + \frac{\lambda_c - \theta C_s s_2}{C_s/\theta})C_s \\
&= \frac{C_c}{\theta} + C_c + (s_2 + \frac{C_c}{C_s} - \theta^2 s_2)C_s \\
&\leq \frac{C_c}{\theta} + C_c + (s_2 + \frac{C_c}{C_s})C_s. \\
\end{aligned}
\end{equation}

Next, we show that $s_2 \le C_c/C_s$. From Eq.~\eqref{eq:nC_s}, as $\sum_{k=n+1}^{n+m}y_k \ge 0$, we have $n < \frac{C_c}{C_s}$. Clearly, we have $s_2 \le n < C_c/C_s$ because $n$ is the total number of items with a threshold equal to $\theta C_s$. Therefore, we have
\begin{equation} \label{eq:UALG-AI}
\begin{aligned}
U_{\mathrm{ALG}}(z) &\leq \frac{C_c}{\theta} + C_c + \frac{2C_c}{C_s}C_s \\
&= C_c + (2+\frac{1}{\theta})C_c.
\end{aligned}
\end{equation}
Substitute Eq.~\eqref{ineq::eta} into Eq.~\eqref{eq:UALG-AI}, we have
\begin{equation}
\begin{aligned}
U_{\mathrm{ALG}}(z) &\leq C_c +  \frac{\theta}{1-\theta}\eta(2+\frac{1}{\theta}) \\
&= \mathit{OPT}(z) + \frac{1+2\theta}{1-\theta}\eta.
\end{aligned}
\end{equation}
Therefore, we have
\begin{equation}
    \mathit{CR}(z) \le \frac{U_{\mathrm{ALG}}(z)}{\mathit{OPT}(z)} = 1 + \frac{1+2\theta}{1-\theta}\frac{\eta}{\mathit{OPT}(z)}.
\end{equation}

\textbf{Case A-II):} In this case, we show that such a total demand does not exist. First, we have
\begin{equation}
\begin{aligned}
\sum_{k=1}^K \min \{z_k, \lambda_{s,k} \} &\overset{(a)}{\leq} \sum_{k=1}^K \min \{z_k,  \frac{C_s}{\theta}  \} \\
    &\leq \sum_{k=1}^K \min \{  \frac{z_k}{\theta}, \frac{C_s}{\theta} \} \\
    & = \frac{1}{\theta}\sum_{k=1}^K \min \{z_k, C_s\} \overset{(b)}{<} \frac{C_c}{\theta}, \\ 
\end{aligned}
\end{equation}
where $(a)$ is due to $\lambda_{s, k} \le C_s/\theta$ for any item $k$ and $(b)$ is due to the condition that the optimal offline algorithm does not make the combo purchase. According to Lemma~\ref{lemma::combo}, it implies that LADTSR does not make the combo purchase, leading to the contradiction.

\textbf{Case B-I):} In this case, LADTSR does not make the combo purchase. According to Lemma~\ref{lemma:LA-ALG}, we have
\begin{equation}
U_{\mathrm{ALG}}^k(z) = 
\begin{cases}
    z_k, & z_k < \lambda_{s,k},\\
    \lambda_{s,k} + C_s, & \text{otherwise}.
\end{cases}     
\end{equation}
Next, we show that for each item $k$, we have
\begin{equation}
    U_{\mathrm{ALG}}^k(z) \leq (1+\theta) (\min \{C_s, y_k\} + \eta_k). \label{ineq:ALG/theta}
\end{equation}
We consider two cases: $y_k \ge C_s$ and $y_k < C_s$.
\begin{itemize}
    \item If $y_k \geq C_s$, we have $\lambda_{s, k} = \theta C_s $ and $\min \{y_k, C_s\} = C_s$. We consider two subcases: $z_k < \theta C_s$ and $z_k \ge \theta C_s$.
    \begin{itemize}
        \item If $z_k < \theta C_s$, we have $U_{\mathrm{ALG}}^k(z) = z_k < \theta C_s \leq C_s$. As we have $C_s = \min \{y_k, C_s\} \le \min \{y_k, C_s\} + \eta_k \le (1+\theta)(\min \{y_k, C_s\} + \eta_k)$, we obtain $U_{\mathrm{ALG}}^k(z) \le (1+\theta) (\min \{C_s, y_k\} + \eta_k)$.
        \item If $z_k \geq \theta C_s$, we have $U_{\mathrm{ALG}}^k(z) =  \theta C_s + C_s \leq (1+\theta) C_s \leq (1+\theta)(\min \{y_k, C_s\} + \eta_k)$.
    \end{itemize}
    \item If $y_k < C_s$, we have $\lambda_{s, k} = C_s/\theta $ and $\min \{y_k, C_s\}= y_k$. Similarly, we consider two subcases: $z_k < C_s/\theta$ and $z_k \ge C_s/\theta$.
    \begin{itemize}
        \item If $z_k <  C_s / \theta$, we have $U_{\mathrm{ALG}}^k(z) = z_k$. As we have $z_k \leq y_k + \eta_k = \min \{y_k, C_s\} + \eta_k \le (1+\theta)(\min \{y_k, C_s\} + \eta_k)$, we obtain $U_{\mathrm{ALG}}^k(z) \le (1+\theta) (\min \{C_s, y_k\} + \eta_k)$.
        \item If $z_k \geq C_s / \theta $, we have $U_{\mathrm{ALG}}^k(z) = C_s/\theta + C_s = (1+\theta)  (C_s / \theta)  \leq (1+\theta) z_k \leq (1+\theta)(\min \{y_k, C_s\} + \eta_k)$. 
    \end{itemize}
\end{itemize}
Combining all the cases, we show that Eq.~\eqref{ineq:ALG/theta} holds. As $\mathit{OPT}(z) = C_c$, we have
\begin{equation}
\begin{aligned}
\mathit{CR}(z) &\le \frac{\sum_{k=1}^K U_{\mathrm{ALG}}^k(z)}{\mathit{OPT}(z)} \\
&\le \frac{(1+\theta)(\sum_{k=1}^K (\min \{y_k, C_s\} + \eta_k))}{C_c} \\
&\overset{(a)}{\leq} \frac{(1+\theta)(C_c + \eta)}{C_c} \\
&= 1 + \theta + (1+\theta)\frac{\eta}{\mathit{OPT}(z)}, 
\end{aligned}
\end{equation}
where $(a)$ is due to Eq.~\eqref{eq:combo-cond2}.

\textbf{Case B-II):} In this case, same as Case B-II when $\lambda_c = \theta^2 C_c$ we have
\begin{equation}
    \mathit{CR}(z) \le 1 + \theta + \frac{1}{1-\theta}\frac{\eta}{\mathit{OPT}(z)}.
\end{equation}

Combining all the four subcases, we derive the upper bound for CR when $\lambda_c = C_c/\theta$, i.e., 
\begin{equation}
\begin{aligned}
\mathit{CR}(z)\leq &\max\{1, 1+\theta\} \\
&\quad + \max\{\frac{1+2\theta}{1-\theta}, 1+\theta, \frac{1}{1-\theta}\} \frac{\eta}{\mathit{OPT}(z)}.
\end{aligned}
\end{equation}
Clearly, we have $\frac{1+2\theta}{1-\theta} \ge \frac{1}{1-\theta} \ge 1 + \theta$ for $\theta \in (0, 1)$. Therefore, we have
\begin{equation}
\mathit{CR}(z) \leq 1+\theta + \frac{1+2\theta}{1-\theta}\frac{\eta}{\mathit{OPT}(z)}.
\end{equation}
Finally, combing the upper bound we derive when $\lambda_c = \theta^2 C_c$ (i.e., Eq.~\eqref{eq:result-theta^2}), we have
\begin{equation}
\begin{aligned}
\mathit{CR}(z) \leq \max \{1 + \theta&, 1 + \theta + \theta^2 \} \\
&+ \max \{\frac{1}{1-\theta}, \frac{1+2\theta}{1-\theta} \} \frac{\eta}{\mathit{OPT}(z)}. 
\end{aligned}
\end{equation}
Clearly, we have $1 + \theta \le 1 + \theta + \theta^2$
and $\frac{1}{1-\theta} \le \frac{1+2\theta}{1-\theta}$.
Therefore, we obtain the upper bound as 
\begin{equation}
    \mathit{CR}(z) \le (1+\theta + \theta^2)+ \frac{1+2\theta}{1-\theta}\frac{\eta}{\mathit{OPT}(z)}.
\end{equation}
According to Lemma~\ref{lemma:opt}, the optimal offline cost only depends on the total demand vector regardless of specific demand sequence $\mathbf{D}$. Therefore, we have
\begin{equation}
    \mathit{CR}(z) \le (1+\theta + \theta^2)+ \frac{1+2\theta}{1-\theta}\frac{\eta}{\mathit{OPT}(\mathbf{D})}.
\end{equation}
\end{proof}

\section{Proof of Lemma~\ref{lemma:LA-ALG}} \label{app:LA-ALG}
\begin{proof}
We prove the result of $U_{\mathrm{ALG}}(z)$ in two cases respectively, i.e., A) LADTSR does not make the combo purchase and B) LADTSR makes the combo purchase.

\textbf{Case A):} In this case, we want to show that
\begin{equation}
U_{\mathrm{ALG}}^k(z) = 
\begin{cases}
    z_k, & z_k < \lambda_{s,k},\\
    \lambda_{s,k} + C_s, & \text{otherwise}.
\end{cases}     
\end{equation}

Because all the demands are integers, therefore $\psi_k(t)$ is always an integer. Therefore, $\psi_k(t) \ge \lambda_{s, k}$ is equivalent to $\psi_k(t) \ge \ceil{\lambda_{s, k}}$. Furthermore, we can directly apply the result in Lemma~\ref{lemma:ALG}. That is, the cost of LADTSR for item $k$ is at most $z_k$ if $z_k < \ceil{\lambda_{s,k}}$, otherwise it is at most $\ceil{\lambda_{s,k}} - 1 + C_s \le \lambda_{s,k} + C_s$. Because $z_k < \ceil{\lambda_{s,k}}$ is equivalent to $z_k < \lambda_{s,k}$ as $z_k$ are integers, we have
\begin{equation}
U_{\mathrm{ALG}}^k(z) = 
\begin{cases}
    z_k, & z_k < \lambda_{s,k},\\
    \lambda_s + C_s, & \text{otherwise}.
\end{cases}     
\end{equation}

\textbf{Case B):} In this case, we want to show that
\begin{equation}
U_{\mathrm{ALG}}(z) = \lambda_c + C_c +f(s_1, s_2)C_s,
\end{equation}
where
\begin{equation}
f(s_1, s_2) = \min \{\frac{\lambda_c}{\theta C_s}, s_2 + \frac{\lambda_c -\theta C_s s_2}{C_s/\theta}, s_1+s_2 \}.    
\end{equation}

Similar to the proof for Lemma~\ref{lemma:ALG}, suppose LADTSR makes $n_s$ single purchases before $t_c$ (recall that $t_c$ is the time-slot when the combo purchase is made), and we will discuss the value of $n_s$ later. Next, we will show that
\begin{equation}
U_{\mathrm{ALG}}(z) = \lambda_c + C_c + n_sC_s.
\end{equation}
To see this, as the combo purchase cost is $C_c$ and the single purchase cost is $n_sC_s$, we only need to show that the total rental cost is $\lambda_c$. To see this, we first show that the rental cost for each item $k$ is at most $\min \{\psi_k(t_c - 1), \lambda_{s, k}\}$.

On the one hand, in any time-slot $t$, if LADTSR chooses to cover the demand by rental (with a rental cost of $a(t)$), we have $\psi_{i(t)}(t) = \psi_{i(t)}(t-1) + a(t)$, which means the associated indicative cost will also increase by $a(t)$. It directly implies that the rental cost for each item $k$ until time-slot $t_c - 1$ cannot exceed $\psi_k(t_c - 1)$. 

On the other hand, the rental cost for each item $k$ cannot exceed $\lambda_{s, k}$ as well. 
We can prove this by contradiction. Soppuse in some time-slot $t$, LADTSR chooses to rent for item $i(t)$ and its rental cost exceeds $\lambda_{s, i(t)}$. As the rental cost is at most $\psi_{i(t)}(t)$, we have $\psi_{i(t)}(t) \ge \lambda_{s, i(t)}$, which directly implies that LADTSR will make the single purchase for $i(t)$, leading to the contradiction.

Therefore, the total rental cost is at most 
\begin{equation}
\sum_{k=1}^K \min \{\psi_k(t_c - 1), \lambda_{s, k}\} \overset{(a)}{=} \psi_c(t_c - 1) \overset{(b)}{\le} \lambda_c,  
\end{equation}
where $(a)$ is due to the definition of $\psi_c(t)$ in Eq.~\eqref{eq:psic} and $(b)$ is due to the combo purchase is not made before $t_c$.

Next, we analyze the value of $n_s$. Recall that $s_1$ is the number of all the items $k$ such that $\lambda_{s, k} = C_s / \theta$ and $z_k \ge C_s / \theta$, and $s_2$ is the number of all the items $k$ such that $\lambda_{s, k} =  \theta C_s$ and $z_k \ge \theta C_s $.
We only need to show that i) $n_s \le \frac{\lambda_c}{\theta C_s}$, ii) $n_s \le s_2 + \frac{\lambda_c- \theta C_s s_2}{C_s/\theta}$, and iii) $n_s\le s_1+s_2$.
\begin{enumerate}[i)]
    \item At the end of time-slot $t_c - 1$, for each item $k$ such that LADTSR makes the single purchase for it, we have $\psi_k(t_c - 1) \ge \lambda_{s, k} \ge \theta C_s$. Therefore, we have $\lambda_c \ge \psi_c(t_c - 1) \ge \sum_{k=1}^K \min \{\psi_k(t_c - 1), \lambda_{s, k}\} \ge n_s \theta C_s$.
    \item Suppose among $n_s$ items, there are $n_1$ items with a threshold equal to $C_s /\theta$ and $n_2$ items with a threshold equal to $\theta C_s$. Clearly, we have $n_1 \le s_1$ and $n_2 \le s_2$ because LADTSR will not make the single purchase for the item that has a total demand less than its single purchase threshold. Then similar to the proof for 1), we have $\lambda_c \ge n_1 C_s/ \theta + n_2 \theta C_s$. As $n_s = n_1 + n_2$, we have
    \begin{equation}
        n_s = n_1 + n_2 \le n_2 + \frac{\lambda_c - \theta C_s n_2}{C_s/\theta} = (1-\theta^2)n_2 + \frac{\theta \lambda_c}{C_s}.
    \end{equation}
    As $1-\theta^2 \ge 0$ and $n_2 \le s_2$, we have
    \begin{equation}
        n_2 \le (1-\theta^2)s_2 + \frac{\theta \lambda_c}{C_s} = s_2 + \frac{\lambda_c - \theta C_s s_2}{C_s/\theta}.
    \end{equation}
    \item Clearly, we have $n_s = n_1 + n_2 \le s_1 + s_2$ because $n_1 \le s_1$ and $n_2 \le s_2$.
\end{enumerate}

\end{proof}

\section{Proof of Robustness} \label{sec::robust}
\begin{proof}
In this section, we prove the other upper bound $1+\theta^{-1} + \theta^{-3}$.
Our goal is still to derive an upper bound of CR over all the demand sequences under any prediction quality. 

The analysis follows a similar pattern to that of Theorem~\ref{thm:cr}. Specifically, we proceed with the proof in two steps:
\begin{enumerate}[{1}a)]
    \item We establish the upper bound of CR for all demand sequences with the same total demand. \\
    \item We define a standard total demand (see Definition~\ref{definition}) to simplify the upper bound we obtained in Step 1a) (see Lemma~\ref{lemma:la-zstd}). \\
    \item We determine the upper bound of CR over all the standard total demand, regardless of the specific demand sequence, which is $1 + \theta^{-1} + \theta^{-3}$.
\end{enumerate}

\textbf{Step 1a):} We can directly apply the result of Lemma~\ref{lemma:LA-ALG} in the proof of consistency.

\textbf{Step 1b):} Given the complexity of $U_{\mathrm{ALG}}(z)$, we establish the standard total demand for LADTSR.
However, the difficulty here is that the single purchase threshold can be two values: $\theta C_s$ or $C_s/\theta$. Therefore, we need to extend our definition of standard total demand for RDTSR. First, we provide the extended definition as follows.
\begin{definition} \label{la:definition}
Given any prediction $y$ and its associated $\lambda_{s, k}$ (for item $k=1, \dots, K$), a total demand $z$ is said to be the standard total demand $z_{\mathrm{std}}^y(m_1, m_2, n_2, x)$ if it satisfies:
\begin{enumerate}[i)]
    \item  for items with a threshold equal to $C_s/\theta$:
    \begin{itemize}
        \item $m_1$ items have a total demand of at least $ \ceil{C_s / \theta}$; 
    \end{itemize}
    \item for items with a threshold equal to $\theta C_s$: 
    \begin{itemize}
        \item $m_2$ items have a total demand of at least $C_s$,
        \item $n_2$ items have a total demand equal to $\ceil{\theta C_s}$;
    \end{itemize}
    \item among all the remaining items, each item $k$ has a total demand within the interval $[0, \min \{C_s, \lambda_{s, k}\})$. Let $x$ represent the sum of their total demands;
\end{enumerate}
here $m_1, m_2, n_2, x$ can be any non-negative integers.
\end{definition}

Next, we present an approach to construct a standard total demand given any total demand $z$, which is presented in Algorithm~\ref{alg::la-construct}. Specifically, we still start with a total demand $z' = \{z'_1, \dots, z'_K\}$ equal to $z$. 
Then we modify $z'$ as follows.
\begin{enumerate}[i)]
    \item If there exists an item $k$ such that $\lambda_{s, k} = C_s/\theta$ and $C_s \le z'_k < \ceil{C_s/\theta}$, set $z'_k$ as $\ceil{C_s/\theta}$.
    \item If there exists an item $k$ such that $\lambda_{s, k} = \theta C_s$ and $\ceil{\theta C_s} \le z'_k < C_s$, set $z'_k$ as $\ceil{\theta C_s}$.
    \item All the remaining items remain the same.
\end{enumerate}

\begin{algorithm}[!tb]
\SetAlgoLined
\SetKwInOut{Input}{Input}\SetKwInOut{Output}{Output}
\Input{$z$, $y$}
\Output{$z'$}
$z' \leftarrow z$\;
Calculate $\lambda_{s, k}$ for each item $k$ according to $y$\;
\While{there exists item $k$ s.t. $\lambda_{s, k} = C_s/\theta$ and $C_s\le z'_k < \ceil{C_s/\theta}$}
{\label{line:la-loop2start}
$z'_k \leftarrow \ceil{C_s/\theta}$\;
\label{line:la-loop2end}
}


\While{there exists item $k$ s.t. $\lambda_{s, k} = \theta C_s$ and $\ceil{\theta C_s} \le z'_k < C_s$}{
\label{line:la-loop3start}
$z'_k \leftarrow \ceil{\theta C_s}$\;
\label{line:la-loop3end}
}
\caption{Construct $z^y_{\mathrm{std}}(m_1, m_2, n_2, x)$ from $z$ and $y$}
\label{alg::la-construct}
\end{algorithm}

We will show that the constructed total demand $z'$ follows Definition~\ref{la:definition} and can bound $U_{\mathrm{CR}}(z)$ in Lemma~\ref{lemma:la-zstd}.


\begin{lemma} \label{lemma:la-zstd}
Given any total demand $z$, the total demand $z'$ constructed by Algorithm~\ref{alg::la-construct} is a standard total demand. Furthermore, we have  
\begin{equation}
  U_{\mathrm{CR}}(z) \le U_{\mathrm{CR}}(z').
\end{equation}
\end{lemma}

We prove Lemma~\ref{lemma:la-zstd} in Appendix~\ref{app:la-zstd}.
Hence, with any prediction $y$ and any total demand $z$, Algorithm~\ref{alg::la-construct} allows us to find a standard total demand to bound $U_{\mathrm{CR}}(z)$.

\textbf{Step 1c):} Given any total demand $z$, we use $z_{\mathrm{std}}^y(m_1,m_2, n_2, x)$ to denote the constructed standard total demand.
To obtain $U_{\mathrm{CR}}(z_{\mathrm{std}}^y(m_1,m_2, n_2, x))$, we consider two cases: A) LADTSR makes the combo purchase on $z_{\mathrm{std}}^y(m_1,m_2, n_2, x)$; B) LADTSR does not make the combo purchase on $z_{\mathrm{std}}^y(m_1,m_2, n_2, x)$.



\textbf{Case A):}  We first analyze $\mathit{OPT}(z_{\mathrm{std}}^y(m_1, m_2, n_2, x))$. According to Definition~\ref{la:definition}, there are $m_1+m_2$ items that have a total demand of at least $C_s$; for the remaining items, their total demand is less than $C_s$ and the sum of their total demand is $n_2\ceil{\theta C_s} + x$. Same as $z_k$, we use $z_{\mathrm{std}}^y(m_1, m_2, n_2, x)_k$ to represent the total demand of item $k$ in $z_{\mathrm{std}}^y(m_1, m_2, n_2, x)$. Then, we can rewrite Eq.~\eqref{eq:OPTcost} as
\begin{equation}
\begin{aligned}
     &\mathit{OPT}(z_{\mathrm{std}}^y(m_1, m_2, n_2, x)) \\
     &\quad = \min \{C_c, \sum_{k=1}^K \min \{C_s, z_{\mathrm{std}}^y(m_1, m_2, n_2, x)_k\} \} \\
     &\quad = \min \{C_c,  (m_1 + m_2)C_s + n_2 \left \lceil \theta C_s \right \rceil + x\}.
\end{aligned}
\end{equation}

Next, we analyze $U_{\mathrm{ALG}}(z_{\mathrm{std}}^y(m_1, m_2, n_2, x))$. For items with a threshold $C_s/\theta$, there are $m_1$ items having a total demand of at least $\ceil{C_s/\theta}\ge C_s/\theta$ (i.e., $s_1 = m_1$). 
Similarly, for items with a threshold $\theta C_s$, there are $m_2$ items having a total demand of at least $C_s \ge \theta C_s$ and $n_2$ items having a total demand of at least $\ceil{\theta C_s} \ge \theta C_s$ (i.e., $s_2=m_2 + n_2$). According to Lemma~\ref{lemma:LA-ALG}, we can rewrite $U_{\mathrm{ALG}}(z_{\mathrm{std}}^y(m_1, m_2, n_2, x))$ as
\begin{equation}
U_{\mathrm{ALG}}(z_{\mathrm{std}}^y(m_1, m_2, n_2, x)) = \lambda_c + C_c + f(m_1, m_2+n_2) C_s.
\end{equation}
Therefore, we can express $U_{\mathrm{CR}}(z_{\mathrm{std}}^y(m_1, m_2, n_2, x))$ as
\begin{equation} \label{eq:laUCR1}
\begin{aligned}
&U_{\mathrm{CR}}(z_{\mathrm{std}}^y(m_1, m_2, n_2, x)) \\
&\quad = \frac{\lambda_c + C_c + f(m_1, m_2+n_2)C_s}{\min \{C_c,  (m_1 + m_2)C_s + n_2 \left \lceil \theta C_s \right \rceil + x\}}.
\end{aligned}
\end{equation}
As LADTSR makes the combo purchase in this case, according to Lemma~\ref{lemma::combo}, we have
\begin{equation} \label{eq:lacombo0}
    \sum_{k=1}^K \min \{z_{\mathrm{std}}^y(m_1, m_2, n_2, x)_k, \lambda_{s,k}\} \ge \lambda_c.
\end{equation}
According to Definition~\ref{la:definition}, there are $m_1$ items with a total demand exceeding its threshold $C_s/\theta$; there are $m_2+n_2$ items with a total demand exceeding its threshold $\theta C_s$; all the remaining items having a total demand less than its threshold and the sum of their total demands is $x$. Therefore, Eq.~\eqref{eq:lacombo0} becomes
\begin{equation} \label{eq:lacombo}
m_1 \frac{C_s}{\theta} + (m_2+n_2)\theta C_s + x  \ge \lambda_c.
\end{equation}

Next, we claim that we only need to consider the case when $m_2=0$. To see this, consider another standard total demand $z'^y_{\mathrm{std}}(m'_1, m'_2, n'_2, x')$ such that $m'_1 =m_1$, $m'_2=0$, $n'_2 = m_2 + n_2$, and $x'=x$. Note that the prediction for $z'^y_{\mathrm{std}}(m'_1, m'_2, n'_2, x')$ remains the same. First, we have
\begin{equation}
\label{eq:lacombo2}
\begin{aligned}
&\sum_{k=1}^K \min \{z'^y_{\mathrm{std}}(m'_1, m'_2, n'_2, x')_k, \lambda_{s, k} \} \\
&\quad \overset{(a)}{=} m'_1 \frac{C_s}{\theta} + (m'_2 + n'_2) \theta C_s + x' \\
&\quad = m_1 \frac{C_s}{\theta} +  (m_2 + n_2) \theta C_s + x \\
&\quad \ge \lambda_c,
\end{aligned}
\end{equation}
where $(a)$ is due to the same reason for Eq.~\eqref{eq:lacombo} (replace $m_1$, $m_2$, $n_2$, $x$ by $m'_1$, $m'_2$, $n'_2$, $x'$).

Eq.~\eqref{eq:lacombo2} implies that LADTSR will make the combo purchase on $z'^y_{\mathrm{std}}(m'_1, m'_2, n'_2, x')$ according to Lemma~\ref{lemma::combo}.
Therefore, according to Lemma~\ref{lemma:LA-ALG}, we have
\begin{equation}
\begin{aligned}
&U_{\mathrm{CR}}(z'^y_{\mathrm{std}}(m'_1, m'_2, n'_2, x')) \\
&\quad = \frac{\lambda_c + C_c +f(m'_1, m'_2+n'_2) C_s}{\min \{C_c,  (m'_1 + m'_2)C_s + n'_2 \left \lceil \theta C_s \right \rceil + x'\}},    
\end{aligned}
\end{equation}
Then, we can show that
\begin{equation} \label{eq:la-m2=0}
    U_{\mathrm{CR}}(z'^y_{\mathrm{std}}(m'_1, m'_2, n'_2, x')) \ge U_{\mathrm{CR}}(z_{\mathrm{std}}^y(m_1, m_2, n_2, x)).
\end{equation} 
To see this, by comparing the denominator, we have
\begin{equation}
\begin{aligned}
&\min \{C_c,  (m'_1 + m'_2)C_s + n'_2 \left \lceil \theta C_s \right \rceil + x'\} \\
&\quad = \min \{C_c,  (m_1+0)C_s + (m_2 + n_2) \left \lceil \theta C_s \right \rceil + x\} \\
&\quad \le \min\{C_c,  (m_1 + m_2) C_s + n_2 \left \lceil \theta C_s \right \rceil + x\}.
\end{aligned}
\end{equation}
For the numerator, as $m'_1 = m_1$ and $m_2'+n'_2 = m_2 + n_2$, we have $f(m'_1, m'_2+n'_2) = f(m_1, m_2+n_2)$.
Therefore, the numerator of $U_{\mathrm{CR}}(z'^y_{\mathrm{std}}(m'_1, m'_2, n'_2, x'))$ is identical to that of $U_{\mathrm{CR}}(z^y_{\mathrm{std}}(m_1, m_2, n_2, x))$, while $U_{\mathrm{CR}}(z'^y_{\mathrm{std}}(m'_1, m'_2, n'_2, x'))$ has a smaller denominator. So we conclude that Eq.~\eqref{eq:la-m2=0} holds. 

Therefore, given any $z^y_{\mathrm{std}}(m_1, m_2, n_2, x)$, there exists another standard total demand $z'^y_{\mathrm{std}}(m'_1, 0, n'_2, x')$ and we have
\begin{equation} \label{eq:laUCR2}
\begin{aligned}
&U_{\mathrm{CR}}(z^y_{\mathrm{std}}(m_1, m_2, n_2, x)) \\
&\quad \le U_{\mathrm{CR}}(z'^y_{\mathrm{std}}(m'_1, 0, n'_2, x')) \\
&\quad = \frac{\lambda_c + C_c +f(m'_1, n'_2) C_s}{\min \{C_c,  m'_1C_s + n'_2 \ceil{\theta C_s} + x'\}},
\end{aligned}
\end{equation}
and Eq.~\eqref{eq:lacombo} 
becomes
\begin{equation} \label{eq:gcon}
m'_1 \frac{C_s}{\theta} + n'_2 \theta C_s + x' \ge \lambda_c.  
\end{equation}
Next, relax $m'_1$, $n'_2$, $x'$ to be non-negative real numbers and let $g(m'_1, n'_2, x')$ denote the RHS of Eq.~\eqref{eq:laUCR2}, i.e.,
\begin{equation} \label{eq:g}
    g(m'_1, n'_2, x') = \frac{\lambda_c + C_c +f(s'_1, s'_2) C_s}{\min \{C_c,  m'_1C_s + n'_2 \ceil{\theta C_s} + x'\}}.
\end{equation}
If we can find an upper bound holds for $g(m'_1, n'_2, x')$ with the constraint Eq.~\eqref{eq:gcon}, it is also an upper bound for $ U_{\mathrm{CR}}(z'^y_{\mathrm{std}}(m'_1, 0, n'_2, x'))$. That is,
\begin{equation} \label{eq:la-com}
\begin{aligned}
&\forall \ m'_1, n'_2, x' \in \mathbb{N}, U_{\mathrm{CR}}(z'^y_{\mathrm{std}}(m'_1, 0, n'_2, x'))\\
&\quad \le \max_{m'_1, n'_2, x' \in \mathbb{N}} U_{\mathrm{CR}}(z'^y_{\mathrm{std}}(m'_1, 0, n'_2, x')) \\
&\quad \le \max_{m'_1, n'_2, x' \in \mathbb{R}^+\cup \{0\}} g(m'_1, n'_2, x').
\end{aligned}
\end{equation}


Then, we claim that to find the upper bound for the function $g(m'_1, n'_2, x')$, we only need to focus on the case when $m'_1C_s/\theta + n'_2 \theta C_s + x' = \lambda_c$. 
To see this, for all the possible $m'_1$, $n'_2$, and $x'$ such that $m'_1C_s/\theta + n'_2 \theta C_s + x' > \lambda_c$, we show that there exists another pair of $\Tilde{m}_1$, $\Tilde{n}_2$, and $\Tilde{x}$ such that $\Tilde{m}_1 C_s/\theta +\Tilde{n}_2 \theta C_s + \Tilde{x} = \lambda_c$ and $g(\Tilde{m}_1, \Tilde{n}_2, \Tilde{x}) \ge g(m'_1, n'_2, x')$.
Specifically, we consider two cases: if $x'=0$ and if $x' > 0$.
\begin{enumerate}[i)]
    \item If $x'=0$, we consider two subcases: if $m'_1=0$ and if $m'_1 > 0$.
    \begin{itemize}
        \item If $m'_1=0$, then Eq.~\eqref{eq:g} becomes
        \begin{equation} \label{eq:laUCR3}
        g(m'_1, n'_2, x') = \frac{\lambda_c + C_c +f(0, n'_2) C_s}{\min \{C_c,  n'_2\ceil{\theta C_s} \}},  
        \end{equation}
where 
\begin{equation}
f(0, n'_2) = \min \{\frac{\lambda_c}{\theta C_s}, n'_2 + \frac{\lambda_c -\theta C_s n'_2}{C_s/\theta}, n'_2\}.
\end{equation}
Also, we have $n'_2\theta C_s > \lambda_c$. Therefore, for $f(0, n'_2)$ we have

\begin{equation}
\begin{aligned}
n'_2 + \frac{\lambda_c -\theta C_s n'_2}{C_s/\theta} &= (1-\theta^2)n'_2 + \frac{\theta\lambda_c}{C_s} \\
& > (1 - \theta^2)\frac{\lambda_c}{\theta C_s} + \frac{\theta\lambda_c}{C_s} \\
&= \frac{\lambda_c}{\theta C_s}.
\end{aligned}
\end{equation}
As $n'_2 > \frac{\lambda_c}{\theta C_s}$, then we have $f(0, n'_2) = \frac{\lambda_c}{\theta C_s}$. 
Therefore, Eq.~\eqref{eq:laUCR3} becomes
\begin{equation}
g(m'_1, n'_2, x') = \frac{\lambda_c + C_c + \frac{\lambda_c}{\theta C_s} C_s}{\min \{C_c,  n'_2 \ceil{\theta C_s} \}}.  
\end{equation}
Let $\Tilde{n}_2 = \frac{\lambda_c}{\theta C_s} < n'_2$ and $\Tilde{m}_1 = \Tilde{x} = 0$, then we have
\begin{equation} 
\begin{aligned}
g(\Tilde{m}_1, \Tilde{n}_2, \Tilde{x}) &= \frac{\lambda_c + C_c + \frac{\lambda_c}{\theta C_s} C_s}{\min \{C_c,  \Tilde{n}_2 \ceil{\theta C_s} \}} \\
&\ge \frac{\lambda_c + C_c + \frac{\lambda_c}{\theta C_s} C_s}{\min \{C_c,  n'_2 \ceil{\theta C_s} \}} \\
&= g(m'_1, n'_2, x').
\end{aligned}
\end{equation}
 
\item If $m'_1 > 0$, because $m'_1C_s/\theta  + n'_2\theta C_s > \lambda_c$, then for $f(m'_1, n'_2)$ we have
\begin{equation}
n'_2 + \frac{\lambda_c - \theta C_s n'_2}{C_s/\theta} < n'_2 + \frac{m'_1 C_s/\theta}{C_s/\theta} = n'_2 + m'_1.
\end{equation}
Therefore, Eq.~\eqref{eq:g} becomes
\begin{equation} \label{eq:laUCR4}
g(m'_1, n'_2, x') = \frac{\lambda_c + C_c +f(m'_1, n'_2) C_s}{\min \{C_c, m'_1C_s + n'_2 \ceil{\theta C_s} \}}, 
\end{equation}
where 
\begin{equation}
f(m'_1, n'_2) = \min \{\frac{\lambda_c}{\theta C_s}, n'_2 + \frac{\lambda_c - \theta C_s n'_2}{C_s/\theta}\}.
\end{equation}
If $n'_2\theta C_s \le \lambda_c$, let $\Tilde{m}_1C_s/\theta + n'_2\theta C_s = \lambda_c$ ($\Tilde{m}_1 \le m'_1$), $\Tilde{n}_2=n'_2$, and $\Tilde{x}=0$, we have
\begin{equation} \label{eq:compare3}
\begin{aligned}
g(\Tilde{m}_1, \Tilde{n}_2, \Tilde{x}) &= \frac{\lambda_c + C_c +f(\Tilde{m}_1, \Tilde{n}_2) C_s}{\min \{C_c, \Tilde{m}_1C_s + \Tilde{n}_2 \ceil{\theta C_s}\}} \\
&= \frac{\lambda_c + C_c +f(m'_1, n'_2) C_s}{\min \{C_c, \Tilde{m}_1C_s + n'_2 \ceil{\theta C_s}\}} \\
&\ge \frac{\lambda_c + C_c +f(m'_1, n'_2) C_s}{\min \{C_c, m'_1C_s + n'_2 \ceil{\theta C_s}\}} \\
&= g(m'_1, n'_2, x').
\end{aligned}
\end{equation}
If $n'_2\theta C_s > \lambda_c$, then set $\Tilde{m}_1=0$, $\Tilde{n}_2=\frac{\lambda_c}{\theta C_s}$ and $\Tilde{x}=0$, we have
\begin{equation}
\begin{aligned}
g(\Tilde{m}_1, \Tilde{n}_2, \Tilde{x}) &= g(0, \Tilde{n}_2, 0) \\
&\overset{(a)}{\ge} g(0, n'_2, 0) \\
&\overset{(b)}{\ge} g(m'_1, n'_2, 0) \\
&= g(0, \Tilde{n}_2, x'),
\end{aligned}
\end{equation}
where $(a)$ is shown in the case when $m'_1=x'=0$ and $(b)$ is shown in Eq.~\eqref{eq:compare3}.

\end{itemize}
\item When $x'>0$, if $m'_1C_s/\theta + n'_2\theta C_s \le \lambda_c$, let $m'_1C_s/\theta + n'_2\theta C_s + \Tilde{x} = \lambda_c$ ($\Tilde{x} < x'$), $\Tilde{m}_1=m'_1$, and $\Tilde{n}_2=n'_2$. Then, from Eq.~\eqref{eq:g} we have
\begin{equation}
\begin{aligned}
g(\Tilde{m}_1, \Tilde{n}_2, \Tilde{x}) &= \frac{\lambda_c + C_c +f(\Tilde{m}_1, \Tilde{n}_2) C_s}{\min \{C_c, \Tilde{m}_1C_s+ \Tilde{n}_2\ceil{\theta C_s} + \Tilde{x} \}} \\
&= \frac{\lambda_c + C_c +f(m'_1, n'_2) C_s}{\min \{C_c, m'_1C_s+ n'_2\ceil{\theta C_s} + \Tilde{x} \}} \\
&\ge \frac{\lambda_c + C_c +f(m'_1, n'_2) C_s}{\min \{C_c, m'_1C_s+ n'_2\ceil{\theta C_s} + x' \}} \\
&= g(m'_1, n'_2, x').
\end{aligned}
\end{equation}

If $m'_1C_s/\theta + n'_2\theta C_s > \lambda_c$, let $\Tilde{x}=0$, similarly we have
\begin{equation}
g(m'_1, n'_2, \Tilde{x}) \le g(m'_1, n'_2, x').
\end{equation}
And we already show how to bound $g(m'_1, n'_2, 0)$ in the case when $x'=0$.

\end{enumerate}
Therefore, we can conclude that if $m'_1C_s/\theta + n'_2\theta C_s + x' > \lambda_c$, there exists non-negative real variables $\Tilde{m}_1, \Tilde{n}_2, \Tilde{x}$ such that $\Tilde{m}_1 C_s/\theta + \Tilde{n}_2 \theta C_s + \Tilde{x}=\lambda_c$ and $g(\Tilde{m}_1, \Tilde{n}_2, \Tilde{x}) \ge g(m'_1, n'_2, x')$.

Next, we focus on $g(\Tilde{m}_1, \Tilde{n}_2, \Tilde{x})$ with $\Tilde{m}_1 C_s/\theta + \Tilde{n}_2 \theta C_s + \Tilde{x}=\lambda_c$. Substitute $\Tilde{n}_2 \theta C_s + \Tilde{x} = \lambda_c -  \Tilde{m}_1C_s/\theta$ into $g(\Tilde{m}_1, \Tilde{n}_2, \Tilde{x})$, we have
\begin{equation} \label{eq:laUCR5}
\begin{aligned}
g(\Tilde{m}_1, \Tilde{n}_2, \Tilde{x})& \le \frac{\lambda_c + C_c + f(\Tilde{m}_1, \Tilde{n}_2) C_s}{\min \{C_c,  \Tilde{m}_1C_s + \Tilde{n}_2 \ceil{\theta C_s} + \Tilde{x}\}} \\
&\le \frac{\lambda_c + C_c + f(\Tilde{m}_1, \Tilde{n}_2) C_s}{\min \{C_c,  \Tilde{m}_1C_s + \Tilde{n}_2 \theta C_s + \Tilde{x}\}} \\
&= \frac{\lambda_c + C_c + f(\Tilde{m}_1, \Tilde{n}_2) C_s}{\min \{C_c,  \lambda_c - \Tilde{m}_1 (\frac{C_s}{\theta} - C_s)\}},
\end{aligned}  
\end{equation}
where
\begin{equation}
f(\Tilde{m}_1, \Tilde{n}_2) = \min \{\frac{\lambda_c}{\theta C_s}, \Tilde{n}_2 + \frac{\lambda_c - \theta C_s \Tilde{n}_2}{C_s/\theta}, \Tilde{m}_1+\Tilde{n}_2 \}.
\end{equation}
As $f(\Tilde{m}_1, \Tilde{n}_2) \le \Tilde{m}_1+\Tilde{n}_2$, we have
\begin{equation} \label{eq:laUCR6}
g(\Tilde{m}_1, \Tilde{n}_2, \Tilde{x}) \le  \frac{\lambda_c + C_c + (\Tilde{m}_1+\Tilde{n}_2) C_s}{\min \{C_c,  \lambda_c - \Tilde{m}_1(\frac{C_s}{\theta} - C_s)\}}.
\end{equation}
Also, because $\Tilde{m}_1 \frac{C_s}{\theta} + \Tilde{n}_2 \theta C_s + \Tilde{x} = \lambda_c$ and $\Tilde{x} \ge 0$, we have
\begin{equation}
 \Tilde{n}_2 = \frac{\lambda_c - \Tilde{x} - \Tilde{m}_1C_s/\theta}{\theta C_s} \le \frac{\lambda_c - \Tilde{m}_1C_s/\theta}{\theta C_s}.
\end{equation}
Therefore, Eq.~\eqref{eq:laUCR6} becomes
\begin{equation} \label{eq:laUCR7}
g(\Tilde{m}_1, \Tilde{n}_2, \Tilde{x}) \le \frac{\lambda_c + C_c +(\Tilde{m}_1 + \frac{\lambda_c - \Tilde{m}_1C_s/\theta}{\theta C_s})C_s}{\min \{C_c,  \lambda_c - \Tilde{m}_1 (\frac{C_s}{\theta} - C_s)\}}.  
\end{equation}

Next, to analyze the upper bound of CR holds for any $\Tilde{m}_1$, we consider two cases: if $\lambda_c = \theta^2 C_c$ and if $\lambda_c = C_c/\theta$. 

If $\lambda_c = \theta^2 C_c$, then Eq.~\eqref{eq:laUCR7} becomes
\begin{equation} \label{eq:laUCR7-1}
\begin{aligned}
& g(\Tilde{m}_1, \Tilde{n}_2, \Tilde{x}) \\
&\quad \le  \frac{\theta^2 C_c + C_c +(\Tilde{m}_1 + \frac{\theta^2 C_c - \Tilde{m}_1C_s/\theta}{\theta C_s}) C_s}{\min \{C_c,  \theta^2 C_c - \Tilde{m}_1 (\frac{C_s}{\theta} - C_s)\}} \\
&\quad = \frac{\theta^2 C_c + C_c +(\Tilde{m}_1 + \frac{\theta^2 C_c - \Tilde{m}_1C_s/\theta}{\theta C_s}) C_s}{\min \{C_c,  \theta^2 C_c - \Tilde{m}_1C_s (\frac{1}{\theta} - 1)\}}.
\end{aligned}
\end{equation}
As $1/\theta - 1 \ge 0$, for the denominator we have
\begin{equation}
\theta^2 C_c - \Tilde{m}_1C_s(\frac{1}{\theta} - 1) \le \theta^2 C_c \le C_c.
\end{equation} 
Therefore, Eq.~\eqref{eq:laUCR7-1} becomes
\begin{equation} \label{eq:laUCR7-2}
\begin{aligned}
&g(\Tilde{m}_1, \Tilde{n}_2, \Tilde{x}) \\
&\quad \le \frac{\theta^2 C_c + C_c +(\Tilde{m}_1 + \frac{\theta^2 C_c - \Tilde{m}_1C_s/\theta}{\theta C_s}) C_s}{\theta^2 C_c - \Tilde{m}_1(\frac{C_s}{\theta} - C_s)} \\
&\quad = \frac{(1 + \theta + \theta^2) C_c + \Tilde{m}_1C_s (1 - \frac{1}{\theta^2})}{\theta^2 C_c + \Tilde{m}_1C_s(1 - \frac{1}{\theta})}. 
\end{aligned} 
\end{equation}
As $\Tilde{m}_1$, $\Tilde{n}_2$ ,and $\Tilde{x}$ are non-negative and $\Tilde{m}_1 C_s/\theta + \Tilde{n}_2\theta C_s + \Tilde{x} = \theta^2 C_c$, we have
\begin{equation}
    0 \le \Tilde{m}_1 = \frac{\theta^2 C_c - \Tilde{x} - \theta C_s \Tilde{n}_2}{C_s/\theta} \le \frac{\theta^3 C_c}{C_s}.
\end{equation}
Substitute $0 \le \Tilde{m}_1 \le \frac{\theta^3 C_c}{C_s}$ into the denominator of Eq.~\eqref{eq:laUCR7-2} we have
\begin{equation}
    0 < \theta^3 C_c \le \theta^2 C_c + \Tilde{m}_1 C_s (1-\frac{1}{\theta}) \le \theta^2 C_c.
\end{equation}
Therefore, The RHS function in Eq.~\eqref{eq:laUCR7-2} is defined over $\Tilde{m}_1 \in [0, \frac{\theta^3 C_c}{C_s}]$. On the other hand, by taking derivative of the RHS function to $\Tilde{m}_1$, we have
\begin{equation}
\begin{aligned}
& \frac{d}{d \Tilde{m}_1} \left ( \frac{(1 + \theta + \theta^2) C_c + \Tilde{m}_1C_s (1 - \frac{1}{\theta^2})}{\theta^2 C_c + \Tilde{m}_1C_s(1 - \frac{1}{\theta})} \right ) \\
&\quad = \begin{aligned}[t]
    &\frac{C_s(1-\frac{1}{\theta^2})(\theta^2 C_c + \Tilde{m}_1C_s(1-\frac{1}{\theta}))}{(\theta^2 C_c + \Tilde{m}_1C_s(1-\frac{1}{\theta}))^2} \\
    &\quad -  \frac{C_s(1-\frac{1}{\theta})((1 + \theta + \theta^2) C_c + \Tilde{m}_1C_s (1 - \frac{1}{\theta^2}))}{(\theta^2 C_c + \Tilde{m}_1C_s(1-\frac{1}{\theta}))^2}
\end{aligned} \\
&\quad = \frac{C_s(1-\frac{1}{\theta^2})\theta^2C_c + C_s(1-\frac{1}{\theta})(1 + \theta + \theta^2) C_c}{(\theta^2 C_c + \Tilde{m}_1C_s(1-\frac{1}{\theta}))^2}
\end{aligned}
\end{equation}
As the denominator is non-negative and the numerator is independent of $\Tilde{m}_1$, the RHS function in Eq.~\eqref{eq:laUCR7-2} is monotone to $\Tilde{m}_1$.
Therefore, it takes the maximum only when $\Tilde{m}_1=0$ or $\Tilde{m}_1=\frac{\theta^3 C_c}{C_s}$.
When $\Tilde{m}_1 = 0$, we have
\begin{equation}
\begin{aligned}
    \frac{(1 + \theta + \theta^2) C_c + \Tilde{m}_1C_s (1 - \frac{1}{\theta^2})}{\theta^2 C_c + \Tilde{m}_1C_s(1 - \frac{1}{\theta})} &= \frac{(1 + \theta + \theta^2) C_c}{\theta^2 C_c} \\
    &= 1 + \frac{1}{\theta} + \frac{1}{\theta^2}.
\end{aligned}
\end{equation}
When $\Tilde{m}_1 = \frac{\theta^3 C_c}{C_s}$, we have
\begin{equation}
\begin{aligned}
    &\frac{(1 + \theta + \theta^2) C_c + \Tilde{m}_1 C_s (1 - \frac{1}{\theta^2})}{\theta^2 C_c + \Tilde{m}_1C_s(1 - \frac{1}{\theta})} \\
    &\quad = \frac{(1 + \theta + \theta^2) C_c + \frac{\theta^3 C_c}{C_s} C_s (1 - \frac{1}{\theta^2})}{\theta^2 C_c + \frac{\theta^3 C_c}{C_s}C_s(1 - \frac{1}{\theta})} \\
    &\quad= \frac{(1 + \theta +  \theta^2 + \theta^3 - \theta) C_c}{(\theta^2 + \theta^3 - \theta^2) C_c} \\
    &\quad= \frac{1+\theta+\theta^3}{\theta^3} \\
    &\quad = 1 + \frac{1}{\theta} + \frac{1}{\theta^3}  \ge 1 + \frac{1}{\theta} + \frac{1}{\theta^2}.
\end{aligned}   
\end{equation}
Therefore, combining Eqs.~\eqref{eq:laUCR2} and~\eqref{eq:la-com} we have
\begin{equation}
U_{\mathrm{CR}}(z^y_{\mathrm{std}}(m_1, n_1, n_2, x)) \le 1 + \frac{1}{\theta} + \frac{1}{\theta^3}.
\end{equation}

Next, we turn to $\lambda_c = C_c / \theta$, then Eq.~\eqref{eq:laUCR7} becomes
\begin{equation} \label{eq:laUCR8-1}
\begin{aligned}
& g(\Tilde{m}_1, \Tilde{n}_2, \Tilde{x}) \\
&\quad \le \frac{C_c/\theta + C_c +(\Tilde{m}_1 + \frac{C_c/\theta - \Tilde{m}_1C_s/\theta}{\theta C_s}) C_s}{\min \{C_c,  C_c/\theta - \Tilde{m}_1C_s (\frac{1}{\theta} - 1)\}}. 
\end{aligned}
\end{equation}
Because $\Tilde{m}_1$, $\Tilde{n}_2$, and $\Tilde{x}$ are non-negative and $\Tilde{m}_1 C_s/\theta +  \Tilde{n}_2\theta C_s + \Tilde{x} = C_c/\theta$, we have
\begin{equation}
    0 \le \Tilde{m}_1 = \frac{C_c/\theta - \Tilde{n}_2\theta C_s - \Tilde{x}}{C_s/\theta} \le \frac{C_c}{C_s}.
\end{equation}
Therefore, for the denominator in Eq.~\eqref{eq:laUCR8-1} we have
\begin{equation}
\frac{C_c}{\theta} - \Tilde{m}_1C_s (\frac{1}{\theta} - 1) \ge \frac{C_c}{\theta} - \frac{C_c}{C_s}C_s (\frac{1}{\theta} - 1) \ge C_c.
\end{equation}
Therefore, Eq.~\eqref{eq:laUCR8-1} becomes
\begin{equation}
\begin{aligned}
&g(\Tilde{m}_1, \Tilde{n}_2, \Tilde{x}) \\
&\quad \le  \frac{C_c/\theta + C_c +(\Tilde{m}_1 + \frac{C_c/\theta - \Tilde{m}_1C_s/\theta}{\theta C_s}) C_s}{C_c} \\
&\quad = (1 + \frac{1}{\theta} + \frac{1}{\theta^2}) + \frac{C_s (1-\frac{1}{\theta^2})}{C_c}\Tilde{m}_1 \\
&\quad \overset{(a)}{\le} 1 + \frac{1}{\theta} + \frac{1}{\theta^2} \le 1 + \frac{1}{\theta} + \frac{1}{\theta^3},
\end{aligned}   
\end{equation}
where $(a)$ is due to the term $\frac{C_s(1-\theta^{-2})}{C_c}\Tilde{m}_1 \le 0$. 
Therefore, combining Eqs.~\eqref{eq:laUCR2} and~\eqref{eq:la-com} we have
\begin{equation}
U_{\mathrm{CR}}(z^y_{\mathrm{std}}(m_1, n_1, n_2, x)) \le 1 + \frac{1}{\theta} + \frac{1}{\theta^3}.
\end{equation}
Now, we conclude that the CR of LADTSR is upper bounded by $1 + \theta^{-1} + \theta^{-3}$ in Case A.

\textbf{Case B):}  In this case, LADTSR does not make the combo purchase. We consider two subcases: I) If the optimal offline algorithm does not make the combo purchase. II) If it makes the combo purchase.

\textbf{Case B-I):} We show that the CR upper bound $1 + \theta^{-1} + \theta^{-3}$ holds for all the total demand directly without relying on the standard total demand. In this case, the problem degenerates into the ski-rental problem (i.e., $z_k$ is the ski day for item $k$ and $C_s$ is the purchase cost).
Therefore, we can directly apply Theorem~$2.2$ in their work, i.e., for each item $k$, the cost of LADTSR for item $k$ (denoted by $\mathit{ALG}_k(z)$) is at most $(1 + \theta^{-1}) \mathit{OPT}_k(z)$, where $\mathit{OPT}_k(z)$ is the cost of the optimal offline algorithm for item $k$. As the total cost of both algorithms in Case B-I is the sum of their cost for each item, we have
\begin{equation}
\begin{aligned}
\mathit{CR}(z) &= \frac{\sum_{k=1}^K \mathit{ALG}_k(z)}{\sum_{k=1}^K \mathit{OPT}_k(z)} \\
& \le (1 + \theta^{-1})  \frac{\sum_{k=1}^K \mathit{OPT}_k(z) }{\sum_{k=1}^K \mathit{OPT}_k(z)} \\
& = 1 + \frac{1}{\theta} \le 1 + \frac{1}{\theta} + \frac{1}{\theta^3}.
\end{aligned}
\end{equation}


\textbf{Case B-II):} We can show that this case is dominated by Case A. First, we have $\mathit{OPT}(z^y_{\mathrm{std}}(m_1, m_2, n_2, x)) = C_c$. For the cost of LADTSR, rewrite Eq.~\eqref{eq:lacase1} we have
\begin{equation}
\begin{aligned}
    &U_{\mathrm{ALG}}(z^y_{\mathrm{std}}(m_1, m_2, n_2, x)) \\
    \quad &= m_1(\frac{C_s}{\theta}+ C_s) + (m_2 + n_2)(\theta C_s + C_s) + x.
\end{aligned}
\end{equation}
Therefore, for $U_{\mathit{CR}}(z^y_{\mathrm{std}}(m_1, m_2, n_2, x))$ we have
\begin{equation} \label{eq:laA-II}
\begin{aligned}
    &U_{\mathit{CR}}(z^y_{\mathrm{std}}(m_1, m_2, n_2, x)) \\
    \quad &= \frac{m_1(\frac{C_s}{\theta}+ C_s) + (m_2 + n_2)(\theta C_s + C_s) + x}{C_c}.
\end{aligned}
\end{equation}
As LADTSR does not make the combo purchase, according to Lemma~\ref{lemma::combo}, we have
\begin{equation} \label{eq:laA-IIcon}
\begin{aligned}
&\sum_{k=1}^K \min \{z^y_{\mathrm{std}}(m_1, m_2, n_2, x)_k, \lambda_{s, k} \} \\
&\quad= m_1 \frac{C_s}{\theta} + (m_2 + n_2)\theta C_s + x \\
&\quad < \lambda_c.
\end{aligned}
\end{equation}
Then Eq.~\eqref{eq:laA-II} becomes
\begin{equation} 
    U_{\mathit{CR}}(z^y_{\mathrm{std}}(m_1, m_2, n_2, x)) < \frac{\lambda_c + (m_1 + m_2 + n_2)C_s}{C_c}.
\end{equation}
Then we can find another standard total demand $z'^{y'}_{\mathrm{std}}(m'_1, m'_2, n'_2, x')$ such that
\begin{equation} \label{eq:BIIcon}
\begin{aligned}
& m'_1 \ge m_1, m'_2 = m_2, n'_2 = n_2, x' = x, \\
& m'_1 \frac{C_s}{\theta} + (m'_2 + n'_2)\theta C_s + x' \ge \lambda_c.
\end{aligned}
\end{equation}
Note that $z'^{y'}_{\mathrm{std}}(m'_1, m'_2, n'_2, x')$ has a different prediction $y'$ as it has more items compared to $z^y_{\mathrm{std}}(m_1, m_2, n_2, x)$. 
For these extra items, we can assume their predicted total demand is $0$ because the single purchase threshold for these items is $C_s/\theta$, which implies their predicted total demand is less than $C_s$. For any other item, we assume the predicted total demand is the same as $y$. Let $K'$ be the number of items in $z'^{y'}_{\mathrm{std}}(m'_1, m'_2, n'_2, x')$. Then, we have
\begin{equation}
\sum_{k=1}^{K'} \min \{y'_k, C_s\} = \sum_{k=1}^{K} \min \{y_k, C_s\}.
\end{equation}
According to Lines~\ref{line:setthreshold_combo}-\ref{line:setthreshold_end} of LADTSR, it means LADTSR has the same combo purchase threshold on both $y$ and $y'$. According to the construction (i.e., Eq.~\eqref{eq:BIIcon}) and Lemma~\ref{lemma::combo}, LADTSR will make the combo purchase on $z'^{y'}_{\mathrm{std}}(m'_1, m'_2, n'_2, x')$.

Therefore, we have
\begin{equation}
\begin{aligned}
&U_{\mathit{CR}}(z'^{y'}_{\mathrm{std}}(m'_1, m'_2, n'_2, x')) \\
&\quad = \frac{\lambda_c + C_c + f(m'_1, m'_2+n'_2)C_s}{C_c}. 
\end{aligned}
\end{equation}
Next we show that $m_1 + m_2 + n_2 \le f(m'_1, m'_2+n'_2)$.
\begin{enumerate}[i)]
    \item First, from Eq.~\eqref{eq:laA-IIcon} we have
    \begin{equation}
    \begin{aligned}
        (m_1 + m_2 + n_2)\theta C_s &\le m_1 \frac{C_s}{\theta} + (m_2 + n_2)\theta C_s + x\\
        & < \lambda_c.
    \end{aligned}
    \end{equation}
    Therefore, we have $m_1 + n_1 + m_2 < \frac{\lambda_c}{\theta C_s}$.
    \item Second, as $m'_1 \ge m_1$, $m'_2 = m_2$, and $n'_2 = n_2$, we have
    \begin{equation}
        m_1 + m_2 + n_2 \le m'_1 + m'_2 + n'_2.
    \end{equation}
    \item Finally, from Eq.~\eqref{eq:laA-IIcon} we have
    \begin{equation}
        m_1 < \frac{\lambda_c - \theta C_s (m_2 + n_2)}{C_s/\theta}.
    \end{equation}
    Therefore, we have
    \begin{equation}
    \begin{aligned}
    m_1 + m_2 + n_2 &\le m_2 + n_2 + \frac{\lambda_c - \theta C_s (m_2 + n_2)}{C_s/\theta} \\
    &= m'_2 + n'_2 + \frac{\lambda_c - \theta C_s (m'_2 + n'_2)}{C_s/\theta}.
    \end{aligned}
    \end{equation}
\end{enumerate}
Therefore, we conclude that $m_1 + m_2 + n_2 \le f(m'_1, m'_2+n'_2)$. Furthermore, we have $U_{\mathit{CR}}(z^y_{\mathrm{std}}(m_1, m_2, n_2, x)) \le U_{\mathit{CR}}(z'^{y'}_{\mathrm{std}}(m'_1, m'_2, n'_2, x'))$. On the other hand, we have $\mathit{OPT}(z'^{y'}_{\mathrm{std}}(m'_1, m'_2, n'_2, x')) = C_c$ with increased $m'_1$. Therefore, we have
\begin{equation*}
U_{\mathit{CR}}(z'^{y'}_{\mathrm{std}}(m'_1, m'_2, n'_2, x')) \ge U_{\mathit{CR}}(z^{y}_{\mathrm{std}}(m_1, m_2, n_2, x)).   
\end{equation*}
As we have shown $U_{\mathit{CR}}(z'^{y'}_{\mathrm{std}}(m'_1, m'_2, n'_2, x')) \le 1 + \theta^{-1} + \theta^{-3}$ in Case A, we conclude the proof for Case B-II.

Finally, by combining Case A and B, we conclude that the CR of LADTSR is upper bounded by $1 + \theta^{-1} + \theta^{-3}$.
\end{proof}

\section{Proof of Lemma~\ref{lemma:la-zstd}} \label{app:la-zstd}
\begin{proof}
First, we show that the total demand $z'$ constructed by Algorithm~\ref{alg::la-construct} is a standard total demand. First, recall that there are two loops in Algorithm~\ref{alg::la-construct}, i.e., Lines~\ref{line:la-loop2start}-\ref{line:la-loop2end}, and  Lines~\ref{line:la-loop3start}-\ref{line:la-loop3end}. For each item, we consider two cases: $\lambda_{s, k} = C_s/\theta$ and $\lambda_{s, k}=\theta C_s$. 
\begin{itemize}

    \item If $\lambda_{s, k} = C_s/\theta$, after the first loop, the total demand $z'_k$ can be either $z'_k \ge \ceil{C_s/\theta}$ or $[0, C_s)$. That is, except $z'_k \ge \ceil{C_s/\theta}$, the total demand can only be within the interval $[0, \min\{C_s, \lambda_{s, k}\})$.
    \item If $\lambda_{s, k} = \theta C_s$, after the second loop, there are only three cases for the total demand $z'_k$: $z'_k \ge C_s$, $z'_k=\ceil{\theta C_s}$, and $0 \le z'_k < \theta C_s$. That is, except $z'_k \ge C_s$ and $z_k = \ceil{\theta C_s}$, the total demand can only be within the interval $[0, \min\{C_s, \lambda_{s, k}\})$.
\end{itemize}
Combining two cases, we conclude that $z'$ satisfies $(a)$, $(b)$, $(c)$ at the same time.

Second, we show that $U_{\mathrm{CR}}(z') \ge U_{\mathrm{CR}}(z)$. Initially, we have $z' = z$. Then, For each iteration in the first loop, we will show that $U_{\mathrm{CR}}(z')$ does not decrease in Lemma~\ref{la:std1}. For each iteration in the second loop, we will show that $U_{\mathrm{CR}}(z')$ does not decrease in Lemma~\ref{la:std2}. Then we can conclude that $U_{\mathrm{CR}}(z')\ge U_{\mathrm{CR}}(z)$.
\end{proof}

\begin{lemma} \label{la:std1}
Given any total demand $z$, if there exists item $j$ such that, $\lambda_{s,j} = C_s/\theta$ and $C_s \le z_j < \ceil{\lambda_{s,j}}$, then there exists another total demand $z'=\{z'_1, \dots, z'_K\}$ such that
\begin{equation} \label{la:std1:z'}
z'_k = 
\begin{cases}
\ceil{\lambda_{s,j}}, &k=j,\\
z_k, &k\neq j,
\end{cases}
\end{equation}
and we have $U_{\mathrm{CR}}(z) \le U_{\mathrm{CR}}(z')$.
\end{lemma}

\begin{proof}
The proof has two steps: 1) We show that $\mathit{OPT}(z) = \mathit{OPT}(z')$. 2) We show that $U_{\mathrm{ALG}}(z) \le U_{\mathrm{ALG}}(z')$.

\textbf{Step 1):} For $\mathit{OPT}(z)$ we have
\begin{equation}
\begin{aligned}
    \mathit{OPT}(z) &\overset{(a)}{=} \min \{C_c,\sum_{k=1}^K \min \{ C_s, z_k\} \} \\
    &\overset{(b)}{=} \min \{C_c, C_s+ \sum_{k=1, k \neq j}^K \min \{ C_s, z_k\} \} \\
    &\overset{(c)}{=} \min \{C_c, C_s + \sum_{k=1, k \neq j}^K \min \{ C_s, z'_k\} \} \\
    &\overset{(d)}{=} \min \{C_c, \sum_{k=1}^K \min \{ C_s, z'_k\} \}  \overset{(e)}{=} \mathit{OPT}(z'),
\end{aligned}    
\end{equation}
where $(a)$ is due to the definition of $\mathit{OPT}(z)$ (i.e., Eq.~\eqref{eq:OPTcost}), $(b)$ is due to $z_j \ge C_s$, $(c)$ is due to for all the items $k$ such that $k\neq j$ we have $z_k = z'_k$ (i.e., $\min \{C_s, z_k\} = \min \{C_s, z'_k\}$) according to Eq.~\eqref{la:std1:z'}, $(d)$ is due to for item $j$, we have $z'_j = \ceil{C_s/\theta} \ge C_s$ according to Eq.~\eqref{la:std1:z'}, and $(e)$ is due to the definition of $\mathit{OPT}(z)$.

\textbf{Step 2):} We consider two cases: A) If LADTSR makes the combo purchase on $z$. B) If LADTSR does not make the combo purchase on $z$.

\textbf{Case A):}
In this case, according to Lemma~\ref{lemma::combo}, we have $\sum_{k=1}^K \min \{z_k, \lambda_{s, k}\} \ge \lambda_c$. On the other hand, we have $\sum_{k=1}^K \min \{z'_k, \lambda_{s, k}\} \ge \sum_{k=1}^K \min \{z_k, \lambda_{s, k}\}$ since $z'_k \ge z_k$ for item $k=1,\dots,K$ according to Eq.~\eqref{la:std1:z'}. Therefore, LADTSR will make the combo purchase on $z'$ as well. 

Then, according to Lemma~\ref{lemma:LA-ALG}, we have
\begin{equation}
U_{\mathrm{ALG}}(z) = \lambda_c + C_c +f(s_1, s_2)C_s,
\end{equation}
where 
\begin{equation}
f(s_1, s_2) = \min \{\frac{\lambda_c}{\theta C_s}, s_1 + \frac{\lambda_c -\theta C_s s_1}{C_s/\theta}, s_1+s_2 \}. 
\end{equation}
Let $s'_1$ (resp. $s'_2$) denote the number of items in $z'$ such that its threshold is equal to $C_s/\theta$ (resp. $\theta C_s$) and its total demand exceeds its threshold. Because $s'_1 \ge s_1$ and $s'_2 \ge s_2$ since $z'_k \ge z_k$ for item $k=1, \dots, K$, we have $U_{\mathrm{ALG}}(z) \le U_{\mathrm{ALG}}(z')$.

\textbf{Case B):} In this case, we consider two subcases: I) If LADTSR makes the combo purchase on $z'$. II) If LADTSR does not make the combo purchase on $z'$.

\textbf{Case B-I):} 
Let $s_0$ denote the number of items that have a total demand less than its single purchase threshold. Clearly, we have $K=s_0 + s_1 + s_2$.
Suppose for item $k=1,\dots,s_0$, we have $z_k < \lambda_{s, k}$; for item $k=s_0+1,\dots, s_0+s_1$, we have $\lambda_{s, k} = C_s/\theta$ and $z_k\ge C_s/\theta$; and for item $k=s_0+s_1+1,\dots, s_0+s_1+s_2$, we have $\lambda_{s, k} = \theta / C_s$ and $z_k\ge \theta C_s$. 
According to lemma~\ref{lemma:LA-ALG} we have
\begin{equation} \label{eq:laBI1}
    U_{\mathrm{ALG}}(z)= \sum_{k=1}^{s_0} z_k + s_1 (C_s/\theta + C_s) + s_2(\theta C_s + C_s).
\end{equation}
Because LADTSR does not make the combo purchase on $z$, according to Lemma~\ref{lemma::combo}, we have
\begin{equation}\label{eq:laBI2}
    \sum_{k=1}^K \min \{z_k, \lambda_{s, k}\} = \sum_{k=1}^{s_0} z_k + s_1 C_s/\theta + s_2 \theta C_s <\lambda_c.
\end{equation}
Combining Eqs.~\eqref{eq:laBI1} and~\eqref{eq:laBI2} we have
\begin{equation}
    U_{\mathrm{ALG}}(z) < \lambda_c + (s_1 + s_2) C_s.
\end{equation}

On the other hand, for $z'$, because LADTSR makes the combo purchase on it, we have
\begin{equation}
    U_{\mathrm{ALG}}(z) = \lambda_c + C_c + f(s'_1, s'_2)C_s,
\end{equation}
where 
\begin{equation}
f(s'_1, s'_2)= \min \Big \{\frac{\lambda_c}{\theta C_s},s'_1 + s'_2, s'_2 + \frac{\lambda_c -\theta C_s s'_2}{C_s/\theta} \Big \}. 
\end{equation}
Therefore, we only need to show that $s_1 + s_2 \le f(s'_1, s'_2)$.
\begin{itemize}
    \item First, from Eq.~\eqref{eq:laBI2} we have
    \begin{equation}
        (s_1 + s_2)\theta C_s \le s_1 \frac{C_s}{\theta} + s_2\theta C_s + x < \lambda_c.
    \end{equation}
    Therefore, we have $s_1 + s_2 < \frac{\lambda_c}{\theta C_s}$.
    \item Second, we have $s_1 \le s'_1$ and $s_2 \le s'_2$ because $z'_k \ge z_k$ for item $k=1, \dots, K$. Therefore, we have
    \begin{equation}
        s_1 + s_2 \le s'_1 + s'_2.
    \end{equation}
    \item Finally, from Eq.~\eqref{eq:laBI2} we have
    \begin{equation}
        s_1 < \frac{\lambda_c - \theta C_c s_2 - \sum_{k=1}^{s_0}z_k}{C_s/\theta} \le \frac{\lambda_c - \theta C_s s_2}{C_s/\theta}.
    \end{equation}
    Therefore, we have
    \begin{equation}
        s_1 + s_2 \le s_2 + \frac{\lambda_c - \theta C_s s_2}{C_s/\theta} = (1-\theta^2) s_2 + \frac{\lambda_c}{C_s/\theta}.
    \end{equation}
    As $s_2 \le s'_2$ and $1-\theta^2 \ge 0$, we have
    \begin{equation}
        s_1 + s_2 \le (1-\theta^2) s'_2 + \frac{\lambda_c}{C_s/\theta} = s'_2 + \frac{\lambda_c - \theta C_s s'_2}{C_s/\theta}.
    \end{equation}
\end{itemize}
Therefore, we conclude that $s_1 + s_2 \le f(s'_1, s'_2)$ and $U_{\mathrm{ALG}}(z) \le U_{\mathrm{ALG}}(z')$.

\textbf{Case B-II):}
For item $j$, we have $U_{\mathrm{ALG}}^j(z') = C_s/\theta + C_s \ge C_s/\theta + 1 \ge \ceil{C_s/\theta}$ because $z'_j = \ceil{C_s/\theta} \ge C_s/\theta$ and $U_{\mathrm{ALG}}^j(z) = z_j < \ceil{C_s/\theta}$ because $z_j < C_s/\theta$. Therefore, we have $U_{\mathrm{ALG}}^j(z) \le U_{\mathrm{ALG}}^j(z')$.

For item $k$ such that $k\neq j$, we have $U_{\mathrm{ALG}}^k(z') = U_{\mathrm{ALG}}^k(z)$ since $z'_k = z_k$ according to Eq.~\eqref{la:std1:z'}.

Therefore, we have $U_{\mathrm{ALG}}^j(z) \le U_{\mathrm{ALG}}^j(z')$.

We conclude the proof by combining Cases A and B.
\end{proof}

\begin{lemma} \label{la:std2}
Given any total demand $z$, if there exists item $j$ such that $\lambda_{s, j} = \theta C_s$ and $\ceil{\lambda_{s,j}} \le z_j < C_s$, then there exists another total demand $z'=\{z'_1, \dots, z'_K\}$ such that
\begin{equation} \label{la:std2:z'}
z'_k = 
\begin{cases}
\ceil{\lambda_{s,j}} , &k=j,\\
z_k, &k\neq j,
\end{cases}
\end{equation}
and we have $U_{\mathrm{CR}}(z) \le U_{\mathrm{CR}}(z')$.
\end{lemma}

\begin{proof}
Similar to Lemma~\ref{std2}, we only need to show that 1) $\mathit{OPT}(z) \ge \mathit{OPT}(z')$ and 2) $U_{\mathrm{ALG}}(z) = U_{\mathrm{ALG}}(z')$.

\textbf{Step 1):} For $\textit{OPT}(z)$, we have
\begin{equation}
\begin{aligned}
    \mathit{OPT}(z) &\overset{(a)}{=} \min \{C_c,\sum_{k=1}^K \min \{ C_s, z_k\} \} \\
    &\overset{(b)}{=} \min \{C_c, z_j + \sum_{\substack{k=1, k \neq j}}^K \min \{ C_s, z_k\} \} \\
    &\overset{(c)}{\ge} \min \{C_c, z'_j + \sum_{\substack{k=1, k \neq j}}^K \min \{ C_s, z'_k\} \} \\
    &\overset{(d)}{=} \min \{C_c, \sum_{k=1}^K \min \{ C_s, z'_k\} \}  \overset{(e)}{=} \mathit{OPT}(z'),
\end{aligned}
\end{equation}
where $(a)$ is due to the definition of $\mathit{OPT}(z)$ (i.e., Eq.~\eqref{eq:OPTcost}), $(b)$ is due to $\min \{C_s, z_j\} = z_j$ because $z_j < C_s$, $(c)$ is due to $z'_k \le z_k$ for item $k=1,\dots,K$ according to Eq.~\eqref{la:std2:z'}, 
$(d)$ is due to $\min \{C_s, z'_j\} = z'_j$ because $z'_j = \ceil{\theta C_s} \le C_s$, and $(e)$ is due to the definition of $\mathit{OPT}(z)$.

\textbf{Step 2):} For LADTSR, we first show that the combo purchase decision of LADTSR is the same on $z$ and $z'$. According to Lemma~\ref{lemma::combo}, we only need to show that $\sum_{k=1}^K \min \{z_k, \lambda_{s, k}\} = \sum_{k=1}^K \min \{z'_k, \lambda_{s, k}\}$. 
This is due to two facts:
\begin{itemize}
    \item For item $j$, we have $\min \{z_j, \lambda_{s, j} \} = \min \{ z'_j, \lambda_{s, j} \} = \lambda_{s, j}$ since $z_j, z'_j \ge \ceil{\lambda_{s, j}} \ge \lambda_{s, j}$ according to Eq.~\eqref{la:std2:z'}.
    \item For item $k$ such that $k\neq j$, we have $z_k = z'_k$ according to Eq.~\eqref{la:std2:z'}.
\end{itemize}
Thus, we get $\sum_{k=1}^K \min \{z_k, \lambda_{s, k}\} = \sum_{k=1}^K \min \{z'_k, \lambda_{s, k}\}$.

Next, we consider two cases: whether LADTSR makes the combo purchase on $z$ and $z'$ or not.

If LADTSR makes the combo purchase, according to Lemma~\ref{lemma:LA-ALG} we have
\begin{equation}
U_{\mathrm{ALG}}(z) = \lambda_c + C_c +f(s_1, s_2)C_s,
\end{equation}
where 
\begin{equation}
f(s_1, s_2) = \min \{\frac{\lambda_c}{\theta C_s}, s_2 + \frac{\lambda_c -\theta C_s s_2}{C_s/\theta}, s_1+s_2 \}. 
\end{equation}
Let $s'_1$ (resp. $s'_2$) denote the number of items in $z'$ with a threshold equal to $C_s/\theta$ (resp. $\theta C_s$) and a total demand exceeds its threshold. Then we only need to show that $f(s'_1, s'_2) = f(s_1, s_2)$. 

For item $j$, we have $z_j \ge \lambda_{s, j}$ and $z'_j = \ceil{\lambda_{s, j}} \ge \lambda_{s, j}$ and for any other item $k$ such that $k\neq j$, we have $z_k = z'_k$ according Eq.~\eqref{la:std2:z'}. Therefore, we have $s'_1 = s_1$ and $s'_2 = s_2$. Furthermore, we have $U_{\mathrm{ALG}}(z)=U_{\mathrm{ALG}}(z')$.

If LADTSR does not make the combo purchase, in Eq.~\eqref{eq:lacase1}, we have
\begin{equation}
U_{\mathrm{ALG}}^k(z) = 
\begin{cases}
    z_k, & z_k < \lambda_{s,k},\\
    \lambda_{s,k} + C_s, & \text{otherwise}.
\end{cases}     
\end{equation}
For item $j$, we have $z_j, z'_j \ge \lambda_{s, j}$ and $U_{\mathrm{ALG}}^j(z) = U_{\mathrm{ALG}}^j(z') = \lambda_{s,j} + C_s$; for item $k$ such that $k\neq j$, we have $z_k = z'_k$ and hence $U_{\mathrm{ALG}}^k(z) = U_{\mathrm{ALG}}^k(z')$. Therefore, we have $U_{\mathrm{ALG}}(z) = \sum_{k=1}^K U_{\mathrm{ALG}}^k(z) = U_{\mathrm{ALG}}^k(z')$.

Combining two cases, we have $U_{\mathrm{ALG}}(z) = U_{\mathrm{ALG}}(z')$. 
\end{proof}
\begin{figure*}[!t]
\begin{subfigure}{0.16\linewidth}
\centering
\includegraphics[width=0.95\textwidth]{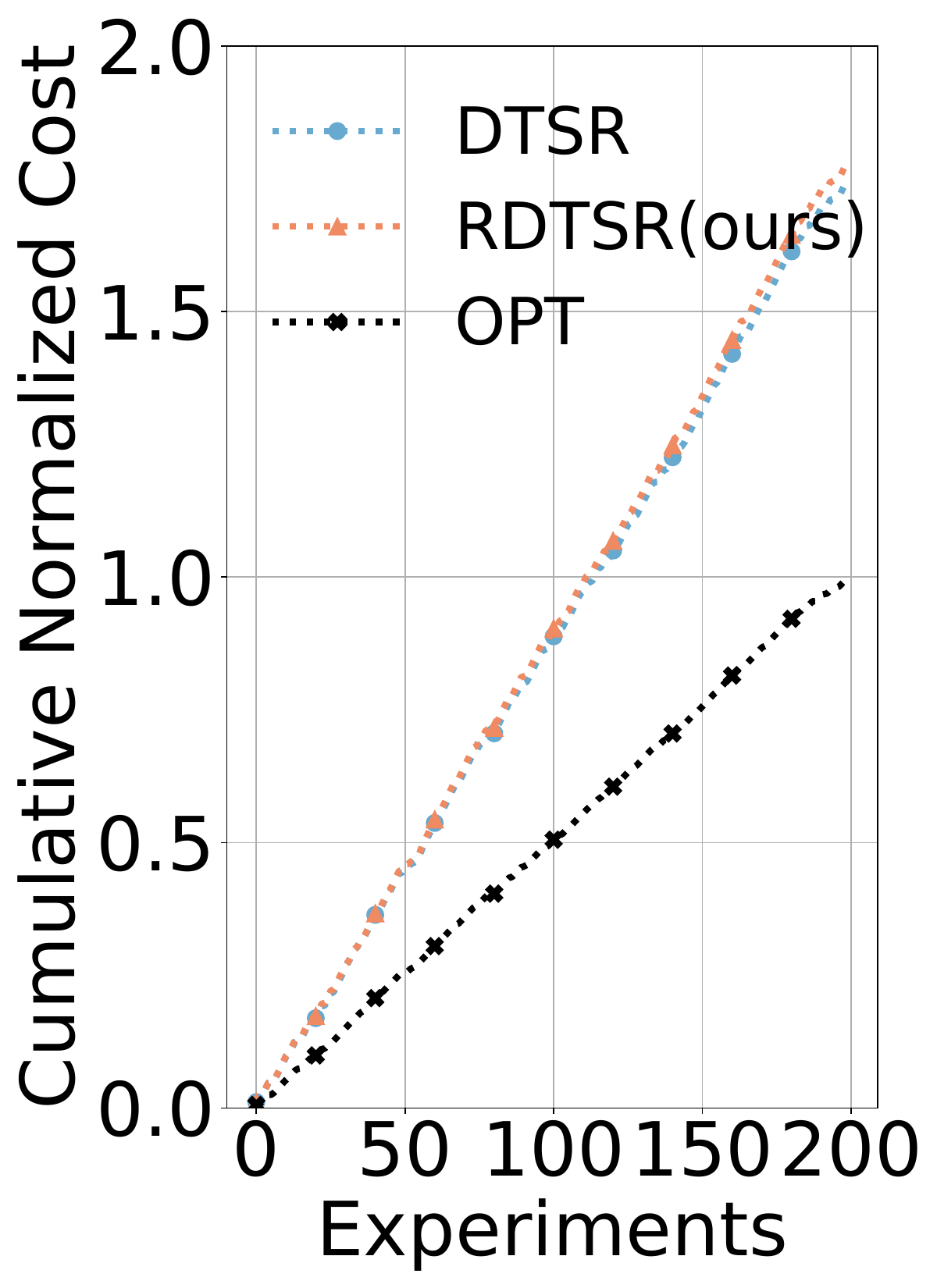}
 \caption{$C_c=20$}
\label{fig:u20}
\end{subfigure}
\begin{subfigure}{0.16\linewidth}
\centering
\includegraphics[width=0.95\textwidth]{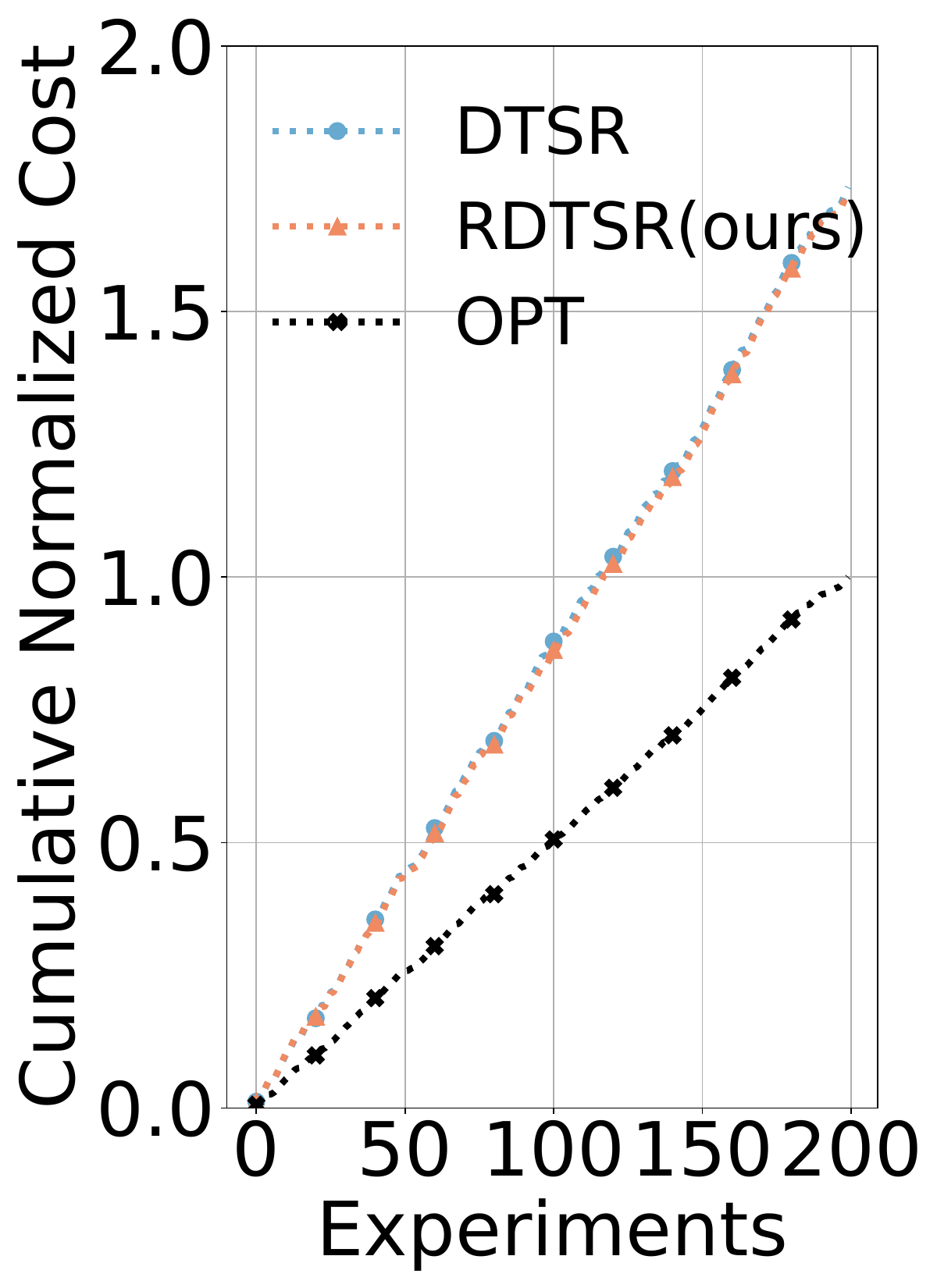}
 \caption{$C_c=25$}
\label{fig:u25}
\end{subfigure}
\begin{subfigure}{0.16\linewidth}
\centering
\includegraphics[width=0.95\textwidth]{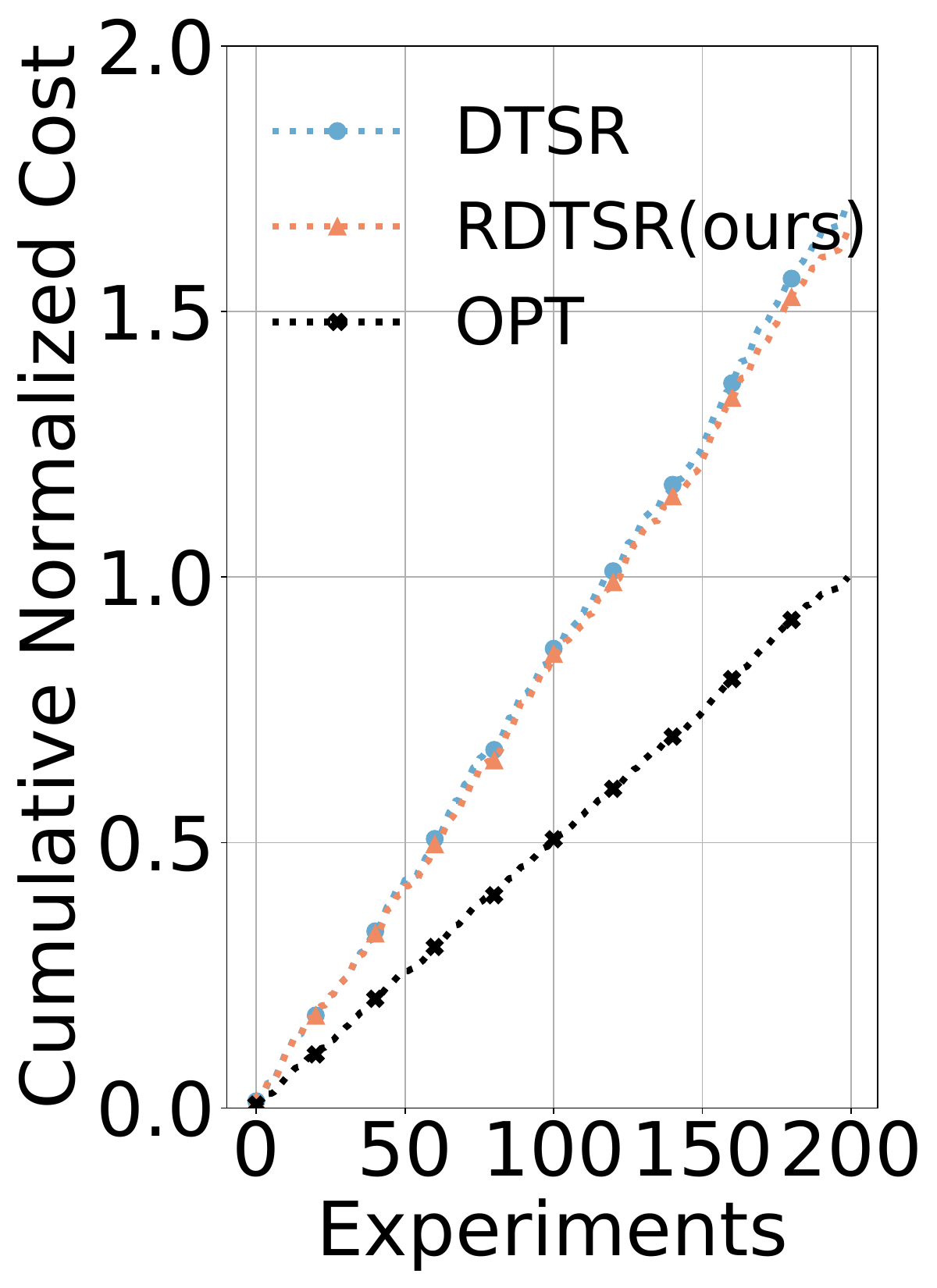}
 \caption{$C_c=30$}
\label{fig:u31}
\end{subfigure}
\begin{subfigure}{0.16\linewidth}
\centering
\includegraphics[width=0.95\textwidth]{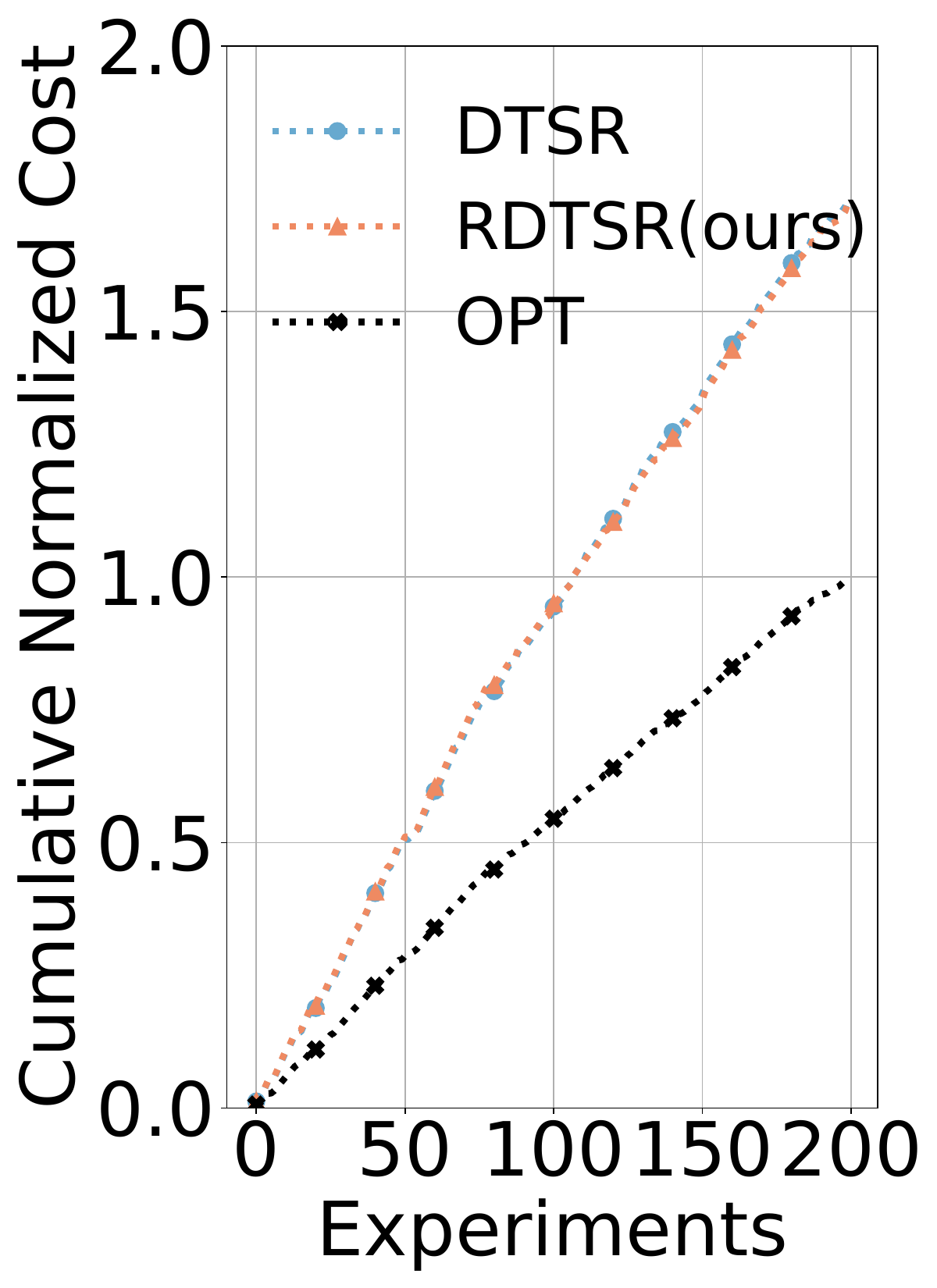}
 \caption{$C_c=20$}
\label{fig:m20}
\end{subfigure}
\begin{subfigure}{0.16\linewidth}
\centering
\includegraphics[width=0.95\textwidth]{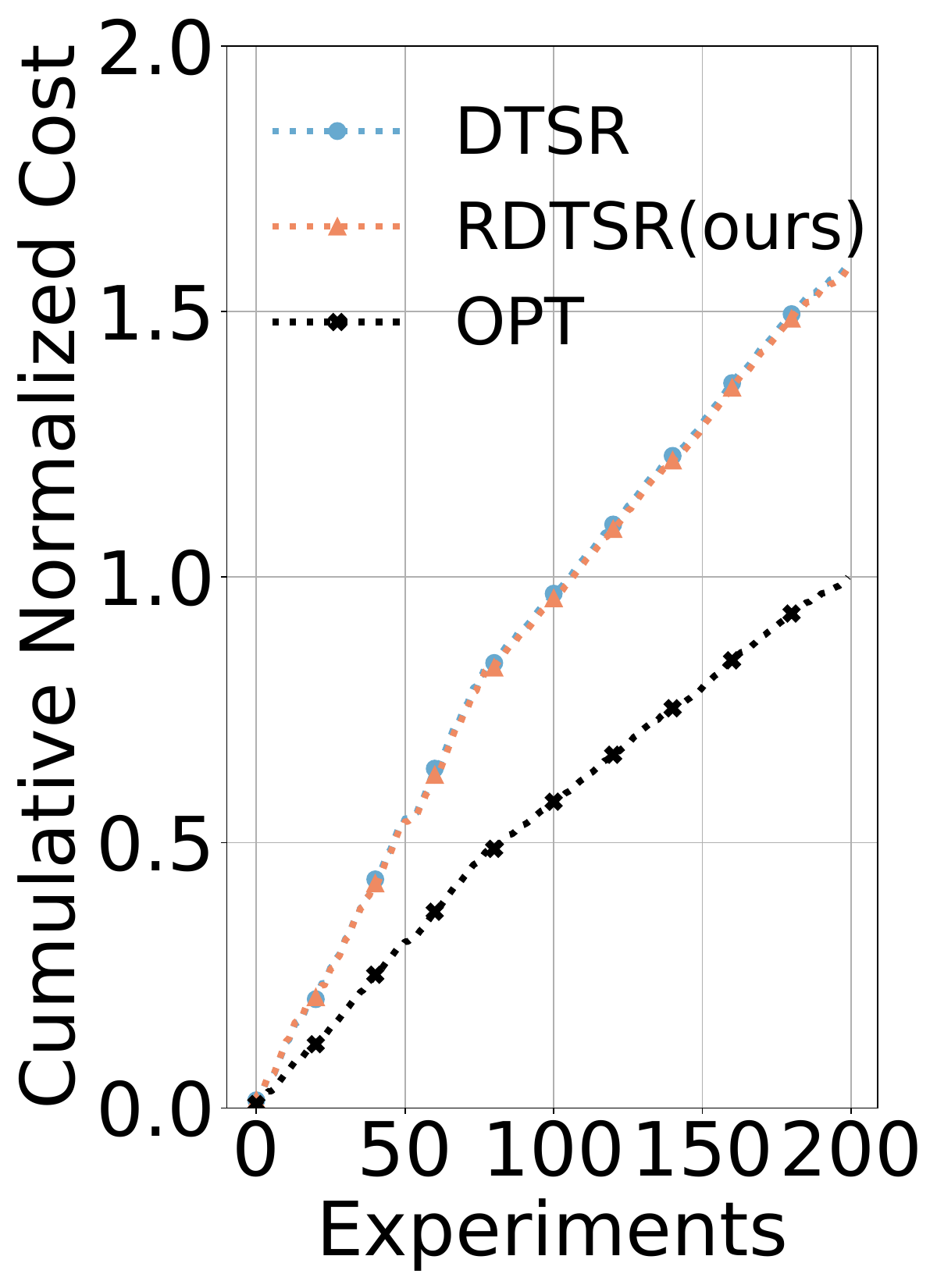}
 \caption{ $C_c=25$}
\label{fig:m25}
\end{subfigure}
\begin{subfigure}{0.16\linewidth}
\centering
\includegraphics[width=0.95\textwidth]{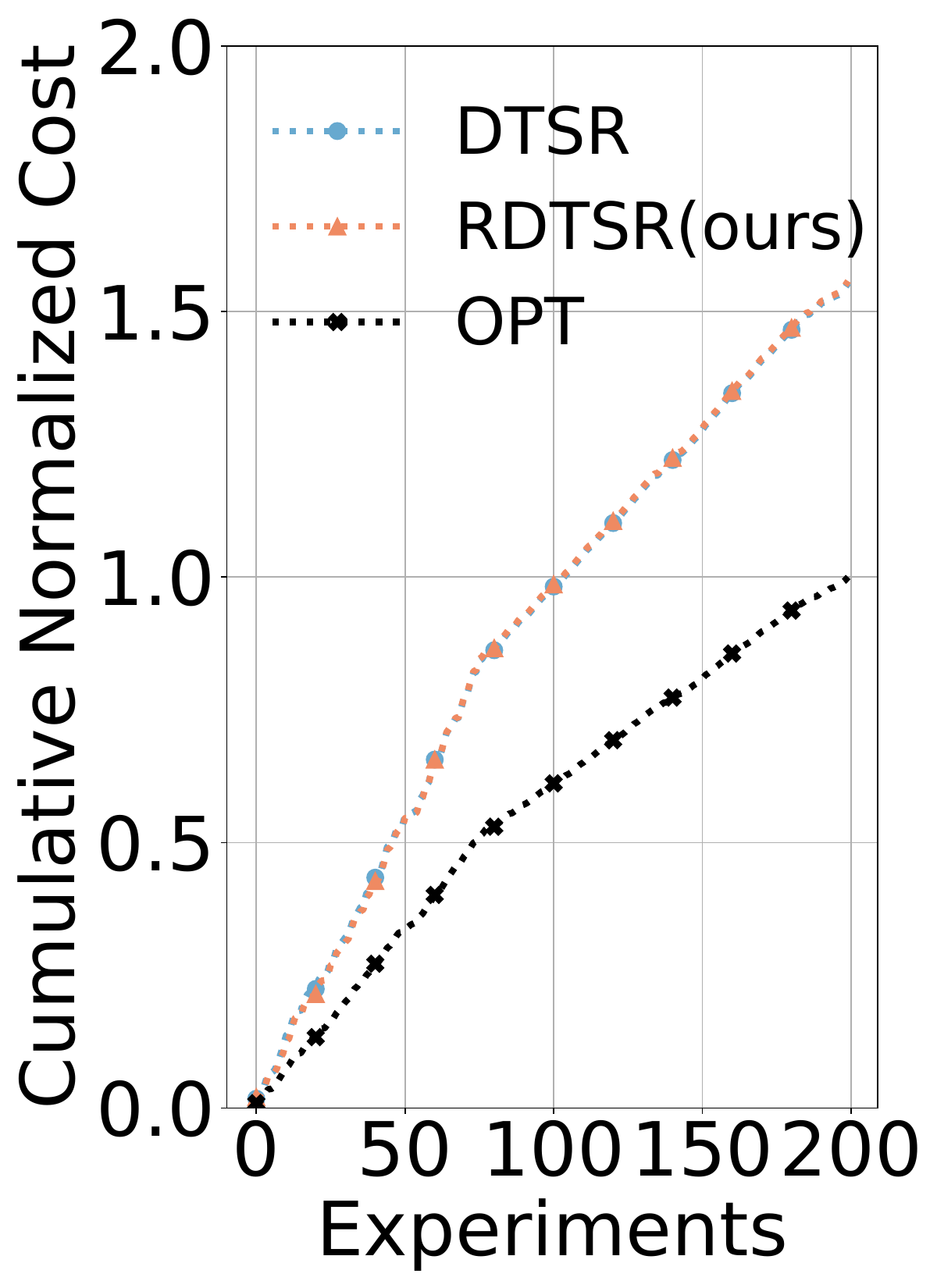}
 \caption{$C_c=30$}
\label{fig:m33}
\end{subfigure}
\caption{Average performance on the synthetic dataset (unit demands).  Figs.~\ref{fig:u20},~\ref{fig:u25}, and~\ref{fig:u31} are under the uniform sequence setting; Figs.~\ref{fig:m20},~\ref{fig:m25}, and~\ref{fig:m33} are under the mixed sequence setting.}
\label{fig:avg_result}
\end{figure*}

\begin{figure*}[!t]
\begin{subfigure}{0.16\linewidth}
\centering
\includegraphics[width=0.95\textwidth]{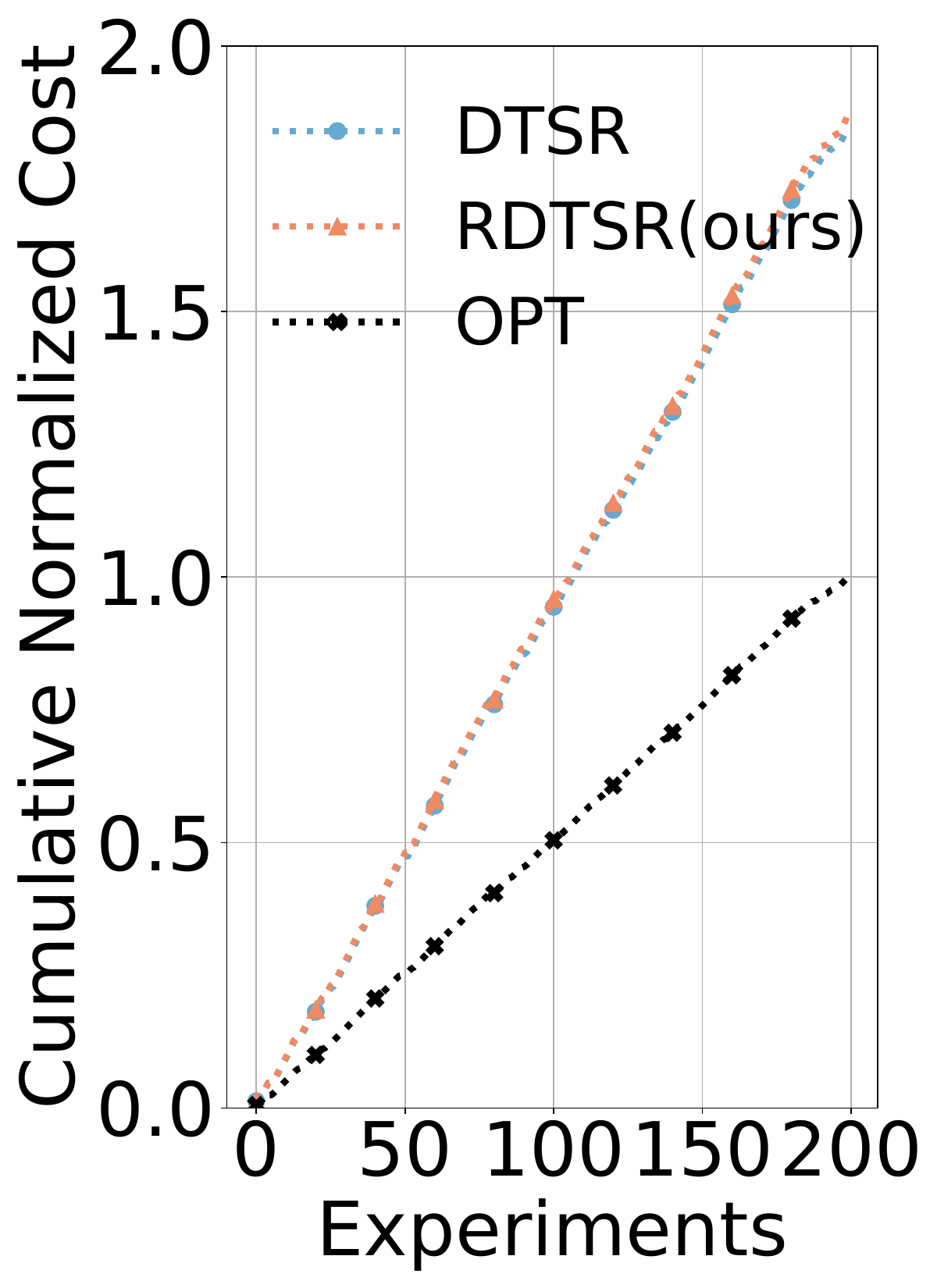}
 \caption{$C_c=20$}
\label{fig:mu20}
\end{subfigure}
\begin{subfigure}{0.16\linewidth}
\centering
\includegraphics[width=0.95\textwidth]{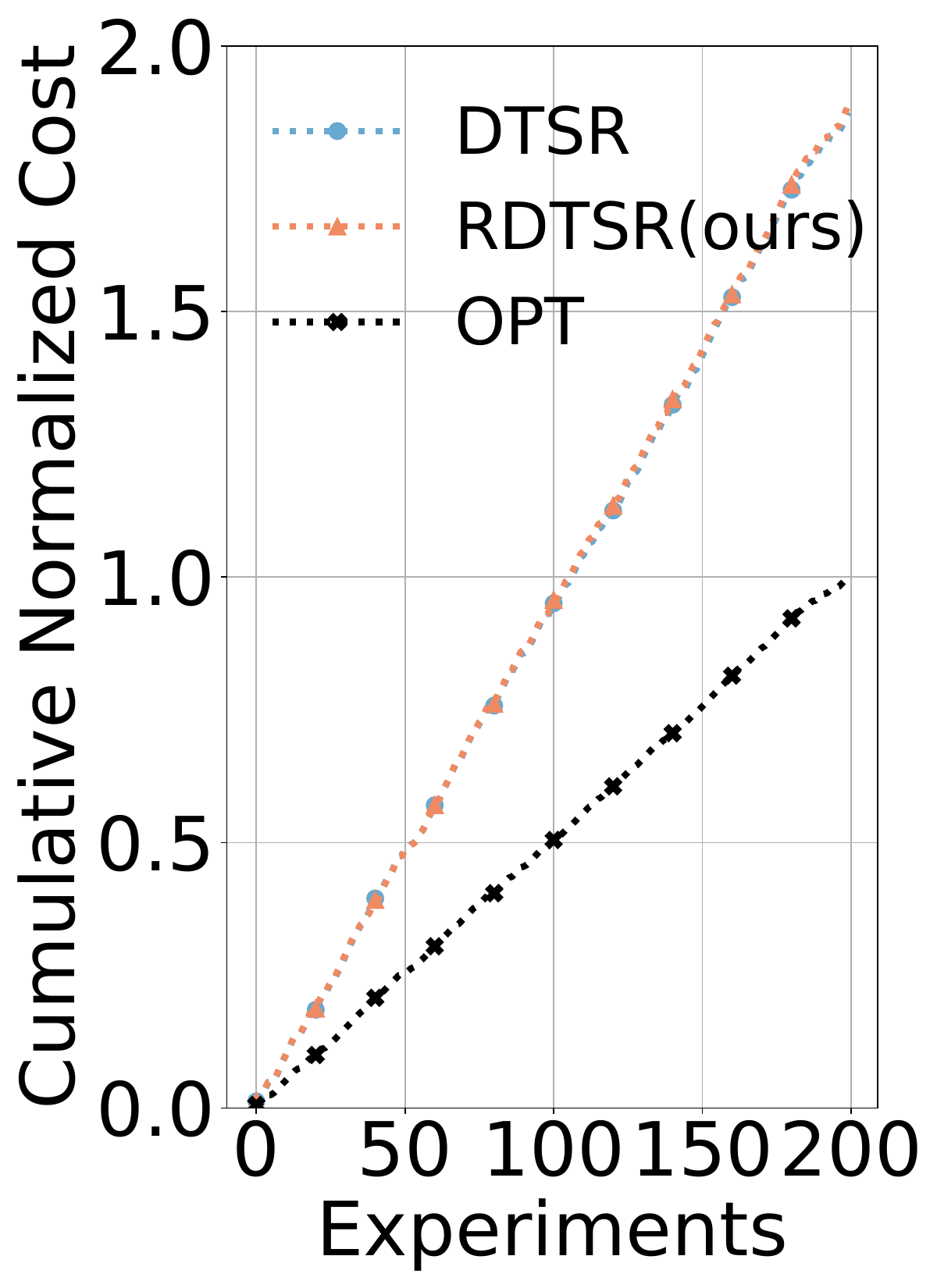}
 \caption{$C_c=25$}
\label{fig:mu25}
\end{subfigure}
\begin{subfigure}{0.16\linewidth}
\centering
\includegraphics[width=0.95\textwidth]{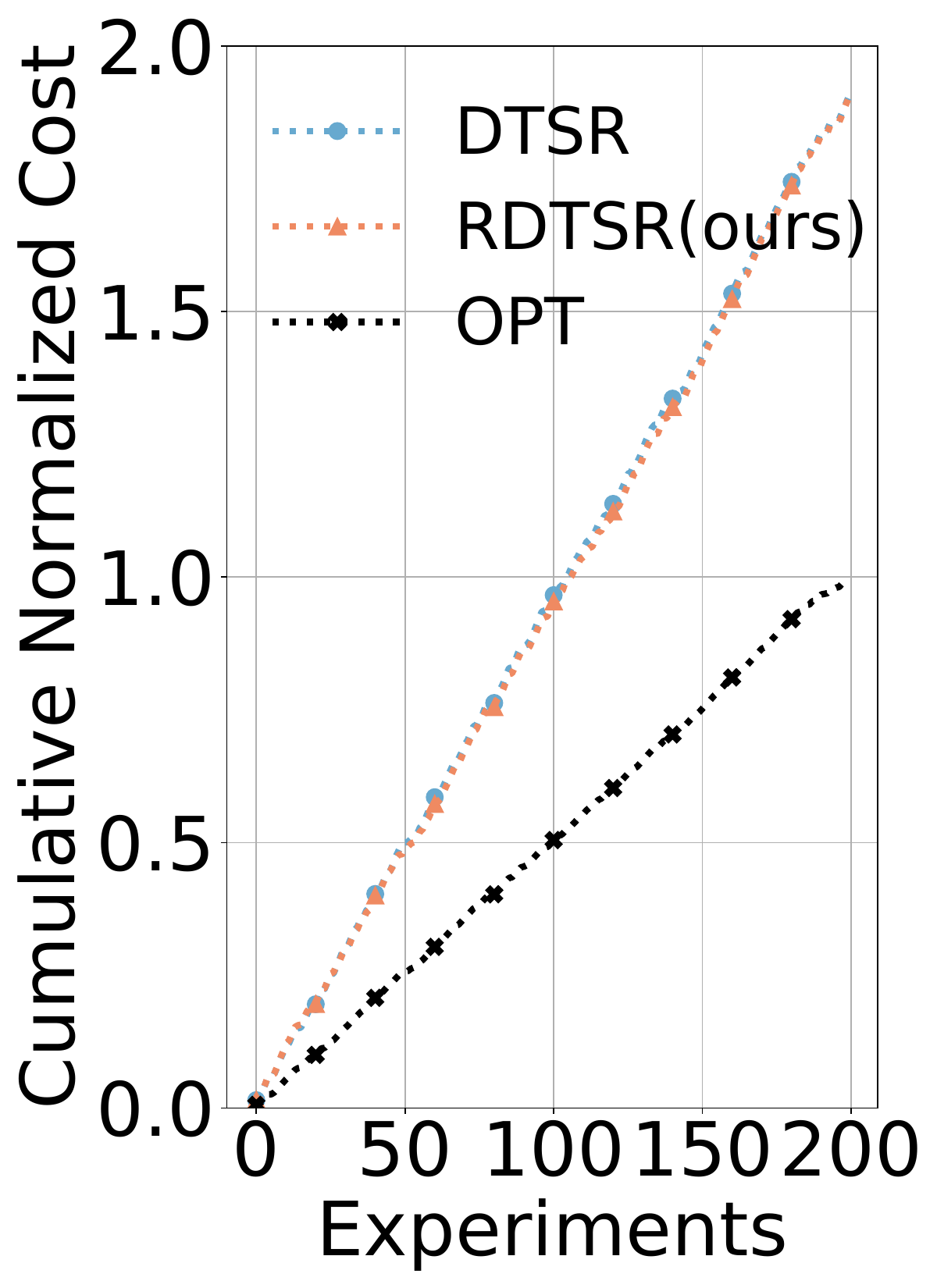}
 \caption{$C_c=30$}
\label{fig:mu31}
\end{subfigure}
\begin{subfigure}{0.16\linewidth}
\centering
\includegraphics[width=0.95\textwidth]{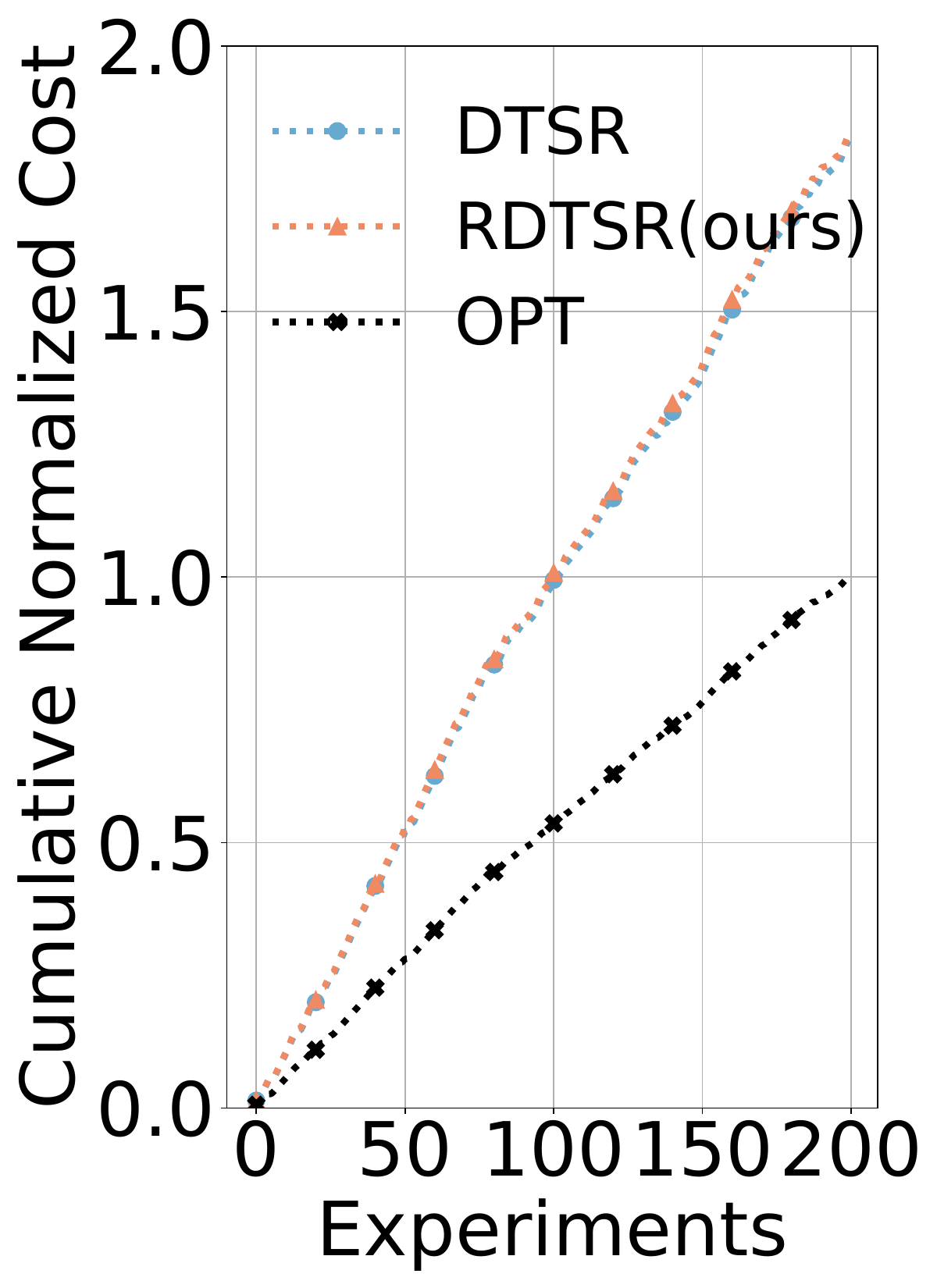}
 \caption{$C_c=20$}
\label{fig:mm20}
\end{subfigure}
\begin{subfigure}{0.16\linewidth}
\centering
\includegraphics[width=0.95\textwidth]{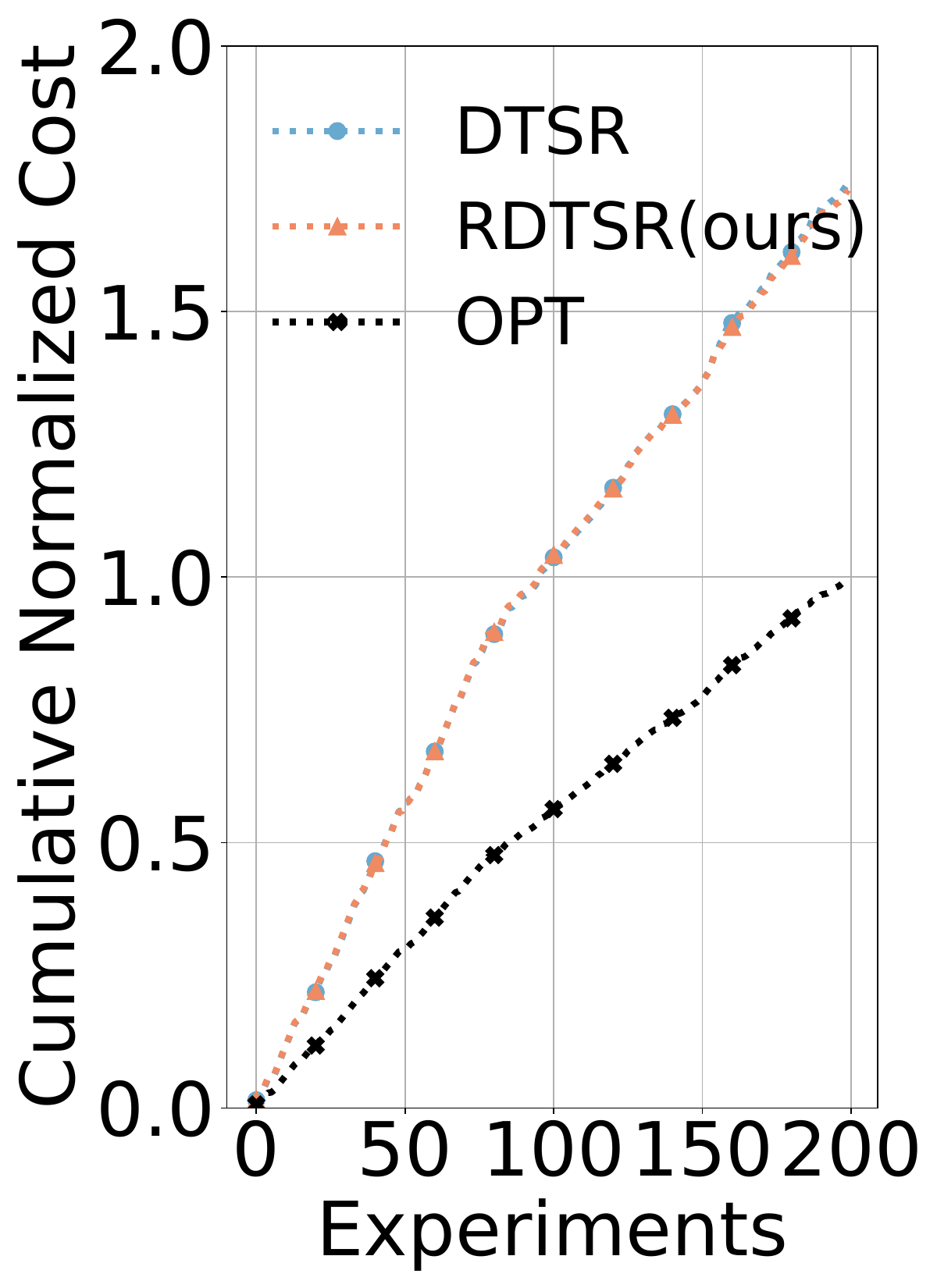}
 \caption{$C_c=25$}
\label{fig:mm25}
\end{subfigure}
\begin{subfigure}{0.16\linewidth}
\centering
\includegraphics[width=0.95\textwidth]{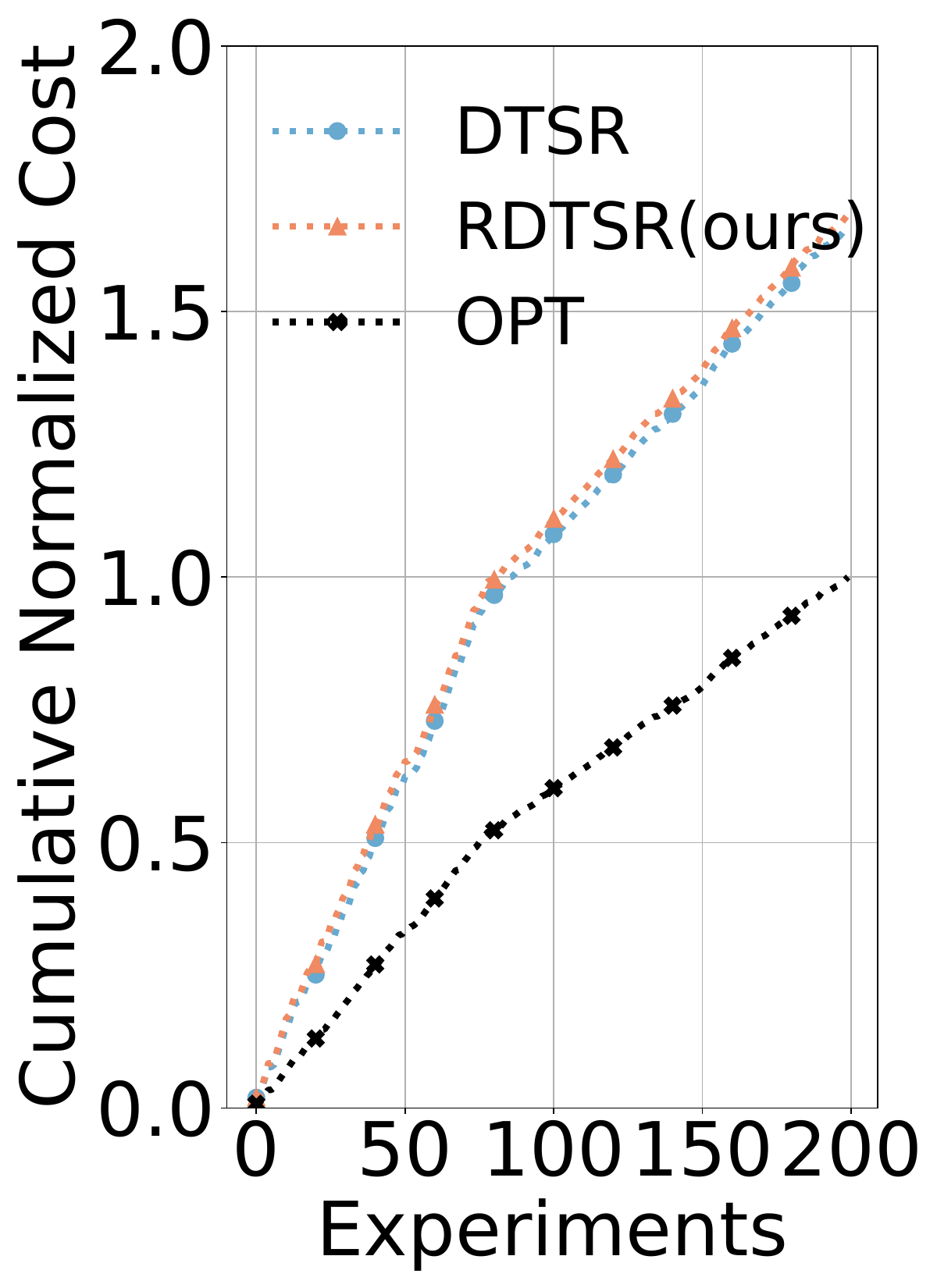}
 \caption{$C_c=30$}
\label{fig:mm33}
\end{subfigure}
\caption{Average performance on the synthetic dataset (multi-unit demands). Figs.~\ref{fig:mu20},~\ref{fig:mu25}, and~\ref{fig:mu31} are under the uniform sequence setting; Figs.~\ref{fig:mm20},~\ref{fig:mm25}, and~\ref{fig:mm33} are under the mixed sequence setting.}
\label{fig:avg_result_multiple}
\end{figure*}

\begin{figure*}[!t]
\begin{subfigure}{0.16\linewidth}
\centering
\includegraphics[width=0.95\textwidth]{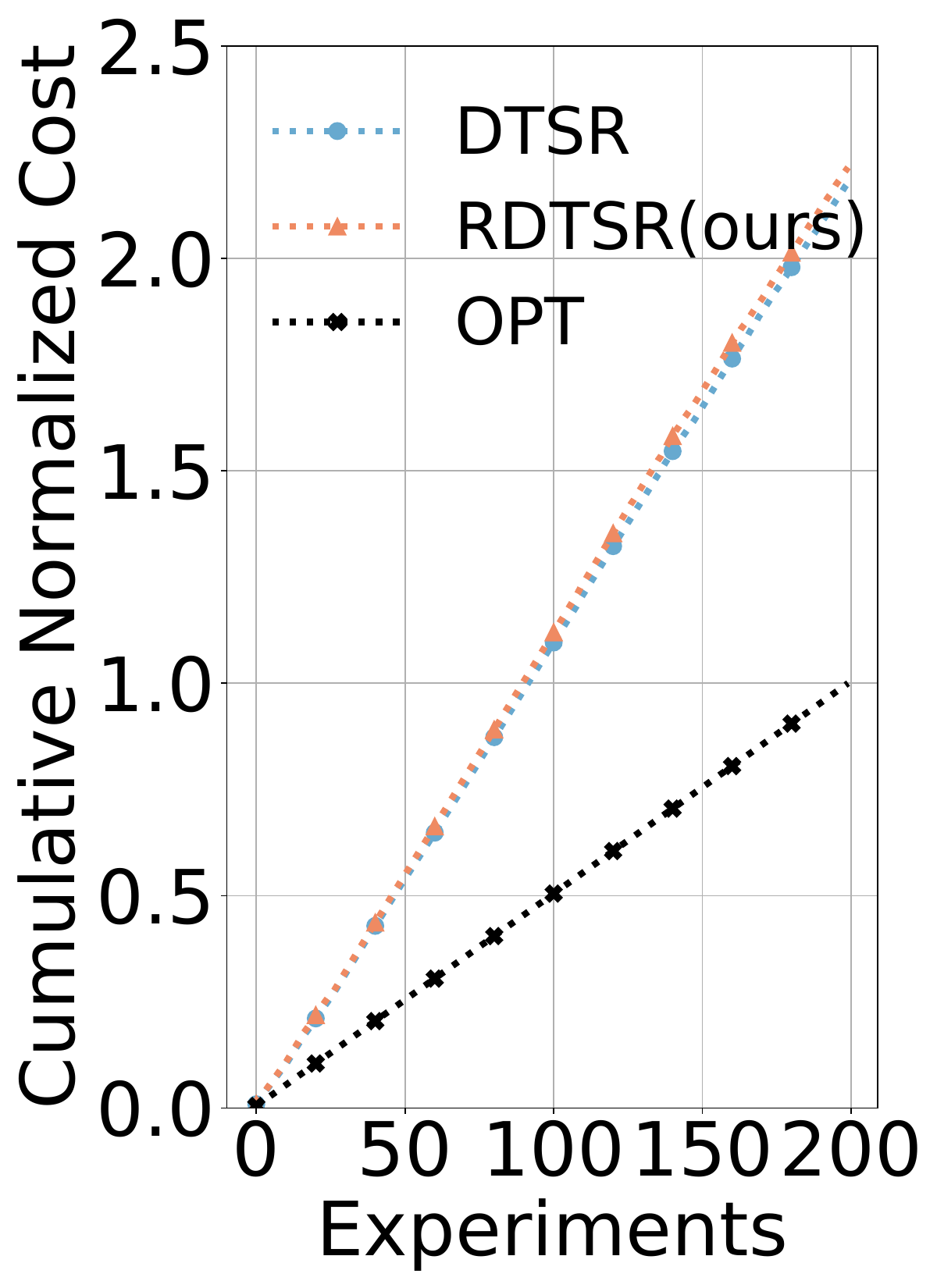}
 \caption{$C_c=20$}
\label{fig:azure20}
\end{subfigure}
\begin{subfigure}{0.16\linewidth}
\centering
\includegraphics[width=0.95\textwidth]{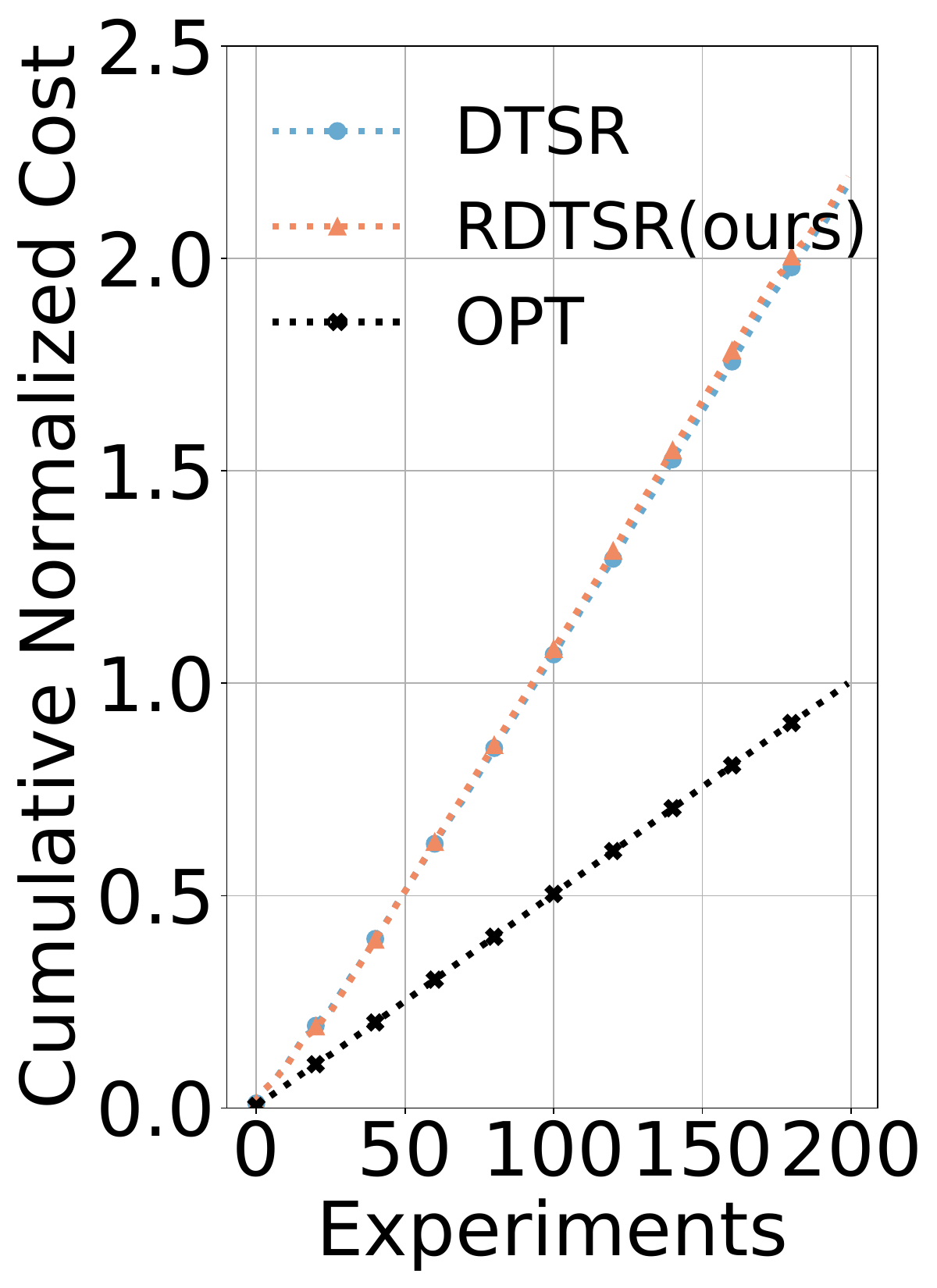}
 \caption{$C_c=25$}
\label{fig:azure25}
\end{subfigure}
\begin{subfigure}{0.16\linewidth}
\centering
\includegraphics[width=0.95\textwidth]{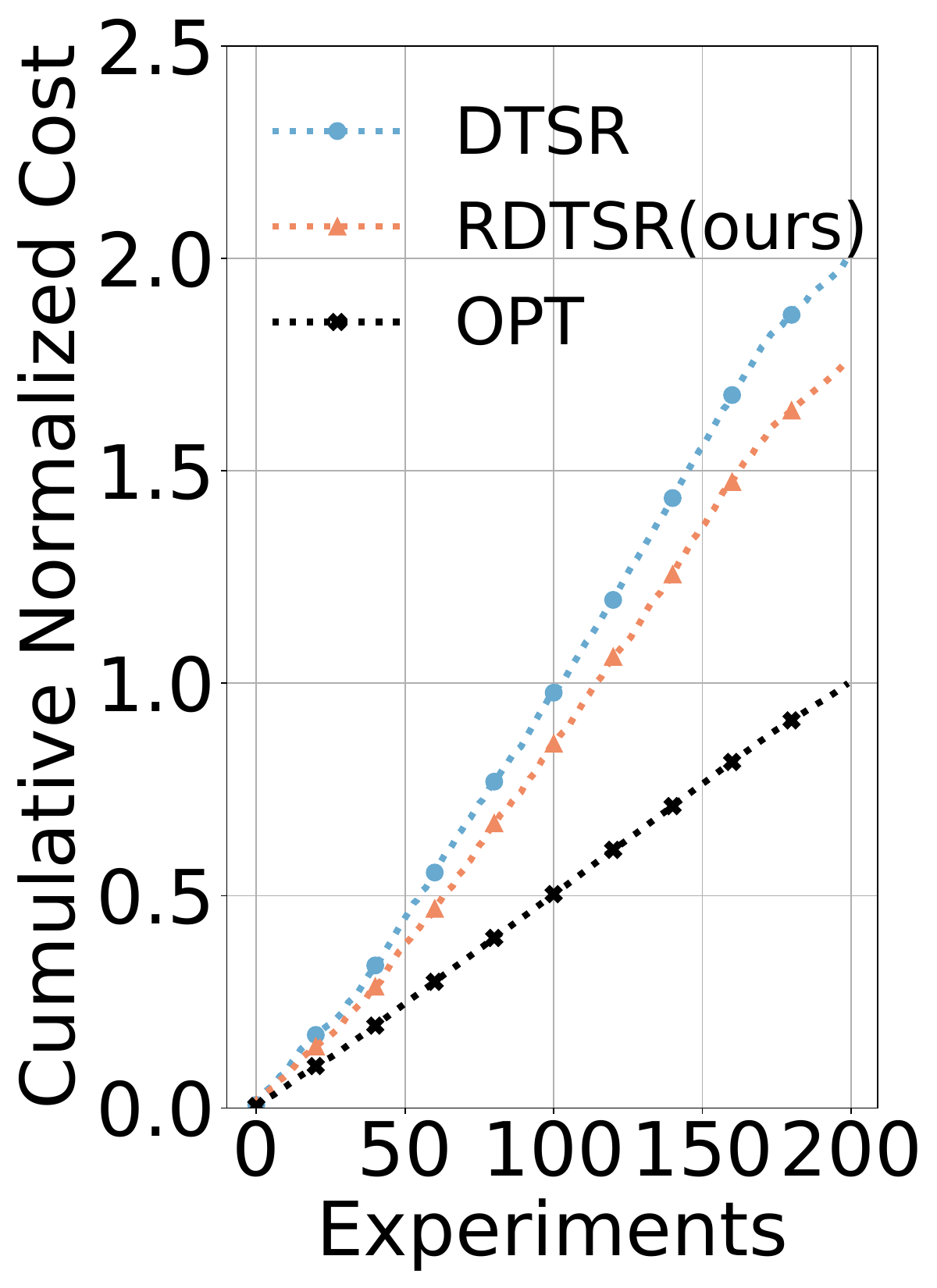}
 \caption{$C_c=30$}
\label{fig:azure30}
\end{subfigure}
\begin{subfigure}{0.16\linewidth}
\centering
\includegraphics[width=0.99\textwidth]{tech_figs/cnc_mix_multiple20.pdf}
 \caption{$C_c=20$}
\label{fig:app20}
\end{subfigure}
\begin{subfigure}{0.16\linewidth}
\centering
\includegraphics[width=0.95\textwidth]{tech_figs/cnc_mix_multiple25.pdf}
 \caption{$C_c=25$}
\label{fig:app25}
\end{subfigure}
\begin{subfigure}{0.16\linewidth}
\centering
\includegraphics[width=0.95\textwidth]{tech_figs/cnc_mix_multiple33.pdf}
 \caption{$C_c=30$}
\label{fig:app30}
\end{subfigure}
\caption{Average performance on the real-world datasets. Figs.~\ref{fig:azure20},~\ref{fig:azure25}, and~\ref{fig:azure30} are under the Azure dataset; Figs.~\ref{fig:app20},~\ref{fig:app25}, and~\ref{fig:app30} are under the AppUsage dataset.}
\label{fig:avg_result_real}
\end{figure*}

\section{Average Performance Evaluation} \label{sec:average}
In this section, we compare the average performance of RDTSR with DTSR. We adopt the metric cumulative normalized cost used in~\cite{wu2021competitive}, which is defined as the sum of the algorithm's cost divided by the sum of the offline optimal cost over all the sequences. We vary the value of the combo purchase price $C_c$ according to~\cite{wu2021competitive}  and for each value of $C_c$, we run our experiments for 200 independent demand sequences.
In Figs.~\ref{fig:avg_result},~\ref{fig:avg_result_multiple}, and~\ref{fig:avg_result_real}, we show the average performance comparison under all the three datasets, i.e., Synthetic (unit demands and multi-unit demands), Azure, and AppUsage datasets.  For the synthetic dataset, in the uniform setting, all 200 sequences are uniform sequences; in the mixed setting, 40\% of sequences are uniform and 60\% of sequences are long-tailed. From Fig.~\ref{fig:avg_result}, we can observe that our algorithm RDTSR has a similar performance compared to DTSR. This indicates that over the long term (in this experiment, 200 rounds), two algorithms have the closing total costs. Similar results can also be observed from two real-world datasets in Figs.~\ref{fig:avg_result_multiple} and~\ref{fig:avg_result_real}.